\documentclass[a4paper,noarxiv,onecolumn,accepted=2024-03-27]{quantumarticle}

\usepackage[top=1 in,bottom=1in, left=1.25 in, right=1.25 in]{geometry}

\usepackage[english]{babel}
\usepackage{color}
\usepackage{graphicx}
\usepackage{framed}
\usepackage[normalem]{ulem}
\usepackage{mathtools}
\usepackage{amsmath}
\usepackage{amsthm}
\usepackage{amssymb}
\usepackage{amsfonts}
\usepackage{enumerate}
\usepackage{algorithm}
\usepackage[noend]{algpseudocode}
\usepackage{appendix}
\usepackage{stmaryrd}
\usepackage{hyperref}
\usepackage{url}

\newcounter{ex}

\theoremstyle{plain}
\newtheorem{theorem}{Theorem}[section]
\newtheorem{lemma}[theorem]{Lemma}

\newtheorem{corollary}[theorem]{Corollary}

\newtheorem{example}[ex]{Example}
\newtheorem*{remark}{Remark}

\theoremstyle{definition}
\newtheorem{definition}{Definition}[section]

\newcommand{\ket}[1]{|#1\rangle}
\newcommand{\bra}[1]{\langle#1|}

\newcommand{\lcm}{\text{lcm}}
\newcommand{\hwgrp}[0]{\overline{\mathcal{P}}_n}

\definecolor{darkgreen}{rgb}{0,0.5,0}
\definecolor{cyan}{rgb}{0,0.5,0.5}
\definecolor{orange}{rgb}{1.0,0.4,0}

\newcommand{\edit}[1]{{\color{black}{#1}}}
\newcommand{\editt}[1]{{\color{black}{#1}}}

\makeatletter
\newtheorem*{rep@theorem}{\rep@title}
\newcommand{\newreptheorem}[2]{%
\newenvironment{rep#1}[1]{%
 \def\rep@title{#2 \ref{##1}}%
 \begin{rep@theorem}}%
 {\end{rep@theorem}}}
\makeatother

\newreptheorem{theorem}{Theorem}
\newreptheorem{lemma}{Lemma}
\newreptheorem{corollary}{Corollary}

\date{}

\begin{document}

\title[Structural aspects of the qudit  Pauli group]{The qudit Pauli group: non-commuting pairs, non-commuting sets, and structure theorems}
%

\author{Rahul Sarkar}
\address{Institute for Computational and Mathematical Engineering, Stanford University, Stanford, CA 94305}
\email{rsarkar@stanford.edu}
\thanks{Rahul Sarkar was funded by the Stanford Exploration Project for the duration of this study.}

\author{Theodore J.~Yoder}
\address{IBM T.J. Watson Research Center, Yorktown Heights, NY}
\email{ted.yoder@ibm.com}
\thanks{}

\date{\today}

\maketitle

\begin{abstract}
Qudits with local dimension $d>2$ can have unique structure and uses that qubits ($d=2$) cannot. Qudit Pauli operators provide a very useful basis of the space of qudit states and operators. We study the structure of the qudit Pauli group for any, including composite, $d$ in several ways. \edit{To cover composite values of $d$, we work with modules over commutative rings, which generalize the notion of vector spaces over fields.} For any specified set of commutation relations, we construct a set of qudit Paulis satisfying those relations. We also study the maximum size of sets of Paulis that mutually non-commute and sets that non-commute in pairs. Finally, we give methods to find near minimal generating sets of Pauli subgroups, calculate the sizes of Pauli subgroups, and find bases of logical operators for qudit stabilizer codes. Useful tools in this study are normal forms from linear algebra over commutative rings, including the Smith normal form, alternating Smith normal form, and Howell normal form of matrices. \editt{Possible applications of this work include the construction and analysis of qudit stabilizer codes, entanglement assisted codes, parafermion codes, and fermionic Hamiltonian simulation.}
\end{abstract}

\section{Introduction}
\label{sec:intro} 

Typical models for quantum computation assume that quantum information is encoded into qubits, which are $d=2$ level systems analogous to classical bits. However, systems with $d>2$ levels, called qudits, also have theoretical advantages, including information storage \cite{greentree2004maximizing, grassl2004optimal,rather2022thirty}, computation \cite{nielsen2002universal,moussa2015quantum,bocharov2017factoring}, the fault-tolerant creation of magic states \cite{campbell2012magic,krishna2019towards}, and simpler realizations of topological phases \cite{ellison2022pauli,ellison2022pauli2}. Experimental implementations of qudits exist. For instance qutrits ($d=3$) with entangling operations can be implemented in superconducting \cite{goss2022high,luo2022experimental} and photonic \cite{nisbet2013photonic} systems. Larger, composite values of $d$ exist as well across superconducting \cite{kues2017chip,fischer2022towards}, photonic \cite{wang2018proof}, and molecular systems \cite{moreno2018molecular,chizzini2022quantum}. We assume throughout that $d=p_0^{\alpha_0}p_1^{\alpha_1}\dots p_{m-1}^{\alpha_{m-1}}$ is the prime factorization of a positive integer $d \geq 2$ into unique primes $p_i$.

Central to the understanding of qubits is the Pauli group on $n$ qubits, a subgroup of $\mathrm{U}(2^n)$ generated by products and tensor products of $X=\left(\begin{smallmatrix}0&1\\1&0\end{smallmatrix}\right)$ and $Z=\left(\begin{smallmatrix}1&0\\0&-1\end{smallmatrix}\right)$. The Pauli group has myriad uses in quantum theory, including facilitating the description of stabilizer states, error-correcting codes, and quantum errors \cite{gottesman1997stabilizer}. The qudit Pauli group, or Heisenberg-Weyl group, plays a similar role in understanding qudits \cite{gottesman1998fault}. Here we take
\begin{equation}\label{eq:qudit_paulis}
X=\sum_{j=0}^{d-1}\ket{j+1}\bra{j},\quad Z=\sum_{j=0}^{d-1}\omega^j\ket{j}\bra{j},
\end{equation}
where $\omega=e^{2\pi i/d}$, as generators of the $n$-qudit Pauli group, a subgroup of $\mathrm{U}(d^n)$. All matrices $X^aZ^b$ for $a,b\in\mathbb{Z}_d$ are distinct. If we take the group commutator \edit{$[A,B]:=ABA^{-1}B^{-1}$} of two single-qudit Paulis, we find $[X^aZ^b,X^{a'}Z^{b'}]=\omega^{a'b-b'a}I$, for identity matrix $I$. Thus, while qubit Paulis $P,Q$ only either commute $[P,Q]=I$ or anticommute $[P,Q]=-I$, qudit Paulis can fail to commute in any one of $d-1$ distinct ways, $[P,Q]=\omega^cI$ for any $c=1,\dots,d-1$. This is one reason there is interesting structure in the qudit Pauli group that is absent for qubits.

In this paper we communicate several structural theorems about sets and groups of qudit Pauli operators. We think it is illustrative to present our results in comparison with the corresponding ideas for qubits, which might be familiar to some readers.

{\it Non-commuting pairs.} We say Paulis $s_0,\dots, s_{k-1}$ and $t_0,\dots, t_{k-1}$ are non-commuting pairs if all pairs $(s_i,t_i)$ do not commute with each other, but do commute with all other $s_j$ and $t_j$ for $j\neq i$. For qubits (and in fact any prime $d$), the maximum number of non-commuting pairs on $n$ qudits is simply $k=n$, for instance $(X,Z)$ on one qubit. However, for composite $d$, we can have more. For example, $d=6$ permits $(X^2,Z^2)$ and $(X^3,Z^3)$ as a collection of $k=2$ non-commuting pairs on one qudit.  In Section~\ref{sec:max-pairs}, we will show that the maximum number of non-commuting pairs is $k=mn$, where $m$ is the number of unique primes in the factorization of $d$. \edit{We use this result to answer some further questions in Sections~\ref{sec:achievable-patterns} and \ref{sec:group_theory}. The key tool used to obtain this result is Lemma~\ref{lem:obs-lem}, which also finds use in Section~\ref{sec:achievable-patterns}.}

{\it Non-commuting sets.} A set of Paulis $\{q_0,q_1,\dots,q_{h-1}\}$ is non-commuting if $[q_i,q_j]\neq I$ for every $i\neq j$. On a single qubit, an example largest non-commuting set is $\{Z,X,XZ\}$ and, more generally for prime $d$, a largest set has size $h=d+1$. In contrast, on a single $d=6$ qudit, the largest non-commuting sets have size $h=12$, e.g.~
\begin{equation}
\{Z,X,XZ,XZ^2,XZ^3,XZ^4,XZ^5,X^2Z,X^2Z^3,X^2Z^5,X^3Z,X^3Z^2\}.
\end{equation}
We will show in Section~\ref{ssec:noncomm-set-max-size} that $h=\Psi(d)\ge d+1$ is the maximum size non-commuting set on a single qudit, where $\Psi$ is the Dedekind Psi function from number theory. The situation of non-commuting sets on more than one qudit is complicated, and we discuss some bounds on the size in that setting. We also note here that if one additionally requires the Paulis in the non-commuting set to all have the same group commutator value, the size of the maximum set is $2n+1$, independent of $d$ \cite{gungordu2014parafermion}.

{\it Qudits needed to create commutation patterns.} Suppose $C\in\mathbb{Z}^{l\times l}$ has zero diagonal, $C_{ii}=0$ for all $i$, and is anti-symmetric, $C=-C^T$. These two properties define what is known as an {\it alternating} matrix. What is the minimum number of qudits $n$ such that there exists a list of $n$-qudit Paulis $\{q_0,q_1,\dots,q_{l-1}\}$ with $[q_i,q_j]=\omega^{C_{ij}}I$? In Appendix B of \cite{sarkar2021graph} and in \cite{gunderman2023transforming}, it was shown that the number of qubits required is half the rank of $C$ as a matrix over the field $\mathbb{F}_2$. For instance, if
\begin{equation}
C=\left(\begin{array}{cccc}0&3&0&0\\-3&0&0&0\\0&0&0&5\\0&0&-5&0\end{array}\right),
\end{equation}
then two qubits are necessary and sufficient via the set $\{X\otimes I,Z\otimes I,I\otimes X,I\otimes Z\}$. However, just one $d=15$ qudit suffices: $\{X^3,Z^9,X^5,Z^5\}$. In \edit{Theorem~\ref{thm:min-qubits-comm-matrix}} of Section~\ref{sec:achievable-patterns}, we show that in general, the number of qudits is half the minimal number of vectors generating the column-space of $C$ as a matrix over $\mathbb{Z}_d$, which can differ from other definitions of matrix rank when $d$ is composite. We also give an efficient algorithm to find a satisfying set of Paulis using a matrix decomposition called the alternating Smith normal form \cite{kuperberg2002kasteleyn}. \edit{This approach effectively reduces the problem to the construction of collections of non-commuting pairs, solved in Section~\ref{sec:max-pairs}.}

{\it Minimal and Gram-Schmidt generating sets.} Suppose $A^{\otimes3}$ is shorthand for $A\otimes A\otimes A$, $S=\{X^{\otimes 3},Z^{\otimes 3},\omega(XZ)^{\otimes 3}\}$ is a set of qudit Paulis, and $G=\langle S\rangle$ is the group generated by $S$. What is the smallest generating set of $G$? We answer this question in general in Section~\ref{ssec:minimal-gen-sets} using the Smith normal form of the matrix over $\mathbb{Z}_d$ representing $S$. For $d=2$, $\omega=-1$ and a smallest generating set for this example is $\{X^{\otimes 3},Z^{\otimes 3}\}$. However, for $d=6$, the last generator contributes a phase unobtainable from the first two generators alone and so any minimal generating set is larger, e.g.~$\{X^{\otimes 3},Z^{\otimes 3},\omega I^{\otimes 3}\}$. A Gram-Schmidt generating set for $G$ is one that can be written as some number of non-commuting pairs $(s_i,t_i)$ and additional elements $u_i$ that commute with all elements of $G$. If $G$ is the centralizer of a qudit stabilizer group, then the non-commuting pairs in Gram-Schmidt generating set provide a basis for the logical operators of the stabilizer code, as is done for qubits in \cite{wilde2009logical}. For instance, $\{X^{\otimes 3},Z^{\otimes 3},\omega I^{\otimes 3}\}$ is also a Gram-Schmidt generating set for the example $G$ with $d=6$, where the first two elements are a non-commuting pair. Our construction of Gram-Schmidt generating sets in Section~\ref{ssec:canonical-form-gen-set} guarantees the minimal number of non-commuting pairs, and again makes use of the alternating Smith normal form.

{\it Subgroups generated by maximum non-commuting pairs.} For qubits (and other qudits with prime $d$), maximum collections of non-commuting pairs generate the full Pauli group. This is clearly the case, for instance, for the maximum collection $(X,Z)$ on a single prime-$d$ qudit. However, for composite $d$ this is no longer necessarily the case. For instance, for $d=12$, $(X^3,Z^6)$ and $(X^4,Z^4)$ generate only $\langle \omega I, X,Z^2\rangle$, a proper subgroup of the full single-qudit Pauli group, despite being a maximum size collection of non-commuting pairs. We show in Section~\ref{ssec:subgroups-noncomm-pairs} that this can only happen when $d$ is not square-free. We also go further and find the size of groups generated by arbitrary generating sets as well.

\edit{

\subsection{Related work}

In the remainder of this section, we discuss prior works and similar results in the existing literature, and mention applications and contributions of our current results.

One of the important uses of qubit based quantum computers is the simulation of Hamiltonians of quantum systems, such as atoms or molecules. The way this is achieved is by mapping the algebra of fermions or bosons, depending on whether one wants to simulate a fermionic or bosonic system, to the algebra of qubits. Popular encoding techniques to achieve such a transformation are the Jordan-Wigner transform \cite{jordan1993paulische}, the Bravyi-Kitaev encoding \cite{bravyi2002fermionic}, the Verstraete-Cirac mapping \cite{verstraete2005mapping}, Fenwick trees \cite{havlivcek2017operator}, and ternary trees \cite{jiang2020optimal}, but other options also exist in the literature \cite{bravyi2017tapering,setia2019superfast,setia2020reducing,seeley2012bravyi,steudtner2017lowering}. One desirable property during this mapping process is to use as few qubits as possible; in fact this has been the focus of the methods described in \cite{bravyi2017tapering,setia2019superfast,setia2020reducing,steudtner2017lowering}. Equivalently, this problem is the same as finding a set of Pauli operators using the minimum number of qubits, where the Paulis must satisfy a certain set of commutation relations --- for example in the case of fermionic simulation, the Paulis that the fermionic modes are mapped to must anticommute with one another. This exact question for the qubit case has been answered in \cite{gunderman2023transforming}, and the result there generalizes to all qudits of prime dimension. Here, we cover the case of qudits of composite dimension. 

The question of finding Paulis using a minimum number of qubits that satisfy a prescribed set of commutation relations also comes up in the context of entanglement assisted quantum error-correcting codes (EAQECC) \cite{brun2006correcting,hsieh2008entanglement,wilde2008optimal}, and quantum convolutional codes \cite{houshmand2012minimal}. A typical EAQECC is described as $[[n,k,d;c]]$, where $n$ qubits are used to encode $k$ logical qubits with code distande $d$, and using $c$ ebits (or extended qubits). The setup for EAQECC is as follows: we initially have a set of Paulis $\{p_0, p_1,\dots,p_{k-1}\}$ on $n$ qubits which do not necessarily commute, and then one seeks to create stabilizer codes out of them by using a set of Paulis $\{q_0, q_1,\dots,q_{k-1}\}$ supported on $c$ extended qubits, such that the set of Paulis $\{p_0 \otimes q_0, p_1 \otimes q_1,\dots,p_{k-1} \otimes q_{k-1}\}$ on $n+c$ qubits form a valid set of stabilizer generators (meaning that they must mutually commute). Clearly, if one wants to be efficient with respect to the number of physical qubits used to encode a fixed number $k$ of logical qubits, one would want to minimize the number $c$ of extended qubits used in the process (we assume $n$ is fixed). This leads to the task of finding these Paulis $\{q_0, q_1,\dots,q_{k-1}\}$ such that $[p_i,p_j]=[q_i,q_j]$, for all $0 \leq i,j \leq k-1$, while at the same time minimizing $c$. Now, one may use the same procedure to construct EAQECC over qudits, and this has been previously considered in \cite[Remark~1]{wilde2008optimal}, and the result there again applies to prime dimensional qudits only. In the case of quantum convolutional codes, an exactly analogous thing happens --- there one seeks to find Paulis supported on a minimum number $m$ of memory qubits that satisfy a memory commutativity matrix specified by the input-output commutation relations \cite[Theorem~2]{houshmand2012minimal}. In order to generalize these convolutional codes to the qudit case, one needs to be able to answer the same question for any qudit dimension $d$. Our results in Section~\ref{sec:achievable-patterns} achieve exactly this.

Yet another instance where the problem of finding Paulis with specified commutation relations comes up is in the setting of parafermion codes. Parafermion codes \cite{gungordu2014parafermion} are the higher dimensional generalization of Majorana codes \cite{kitaev2001unpaired,vijay2017quantum}, just as qudit stabilizer codes generalize qubit stabilizer codes. A key part of the translation between Majorana codes and qubit codes is the construction of sets of qubit Paulis with certain commutation relations, matching those of the underlying Majorana operators \cite{sarkar2021graph}. We anticipate a similar translation between parafermion codes and qudit stabilizer codes to involve constructions of qudit Paulis with specified commutation relations that we provide in Section~\ref{sec:achievable-patterns}. It is to be noted that much of the results in Sections~\ref{sec:max-pairs} and~\ref{ssec:noncomm-set-max-size} set up the necessary machinery to be able to resolve this central question of Section~\ref{sec:achievable-patterns}.

In \cite{gheorghiu2014standard}, Georghiu provides a construction of the standard form of a qudit stabilizer group $S$, i.e.~an abelian subgroup of the qudit Pauli group. This is similar to Gottesman’s standard form \cite{gottesman1997stabilizer} for prime $d$-dimensional qudits and likewise makes use of Clifford transformations of the group. Our constructions of minimal and Gram-Schmidt generating sets apply to any Pauli subgroup and do not involve applying Cliffords. In the qubit case, such a Gram-Schmidt generating set is useful for improving simulations of stabilizer circuits \cite{aaronson2004improved}. We can also efficiently calculate the size of a qudit stabilizer group and thus the dimension of the stabilized codespace using the results in Section~\ref{ssec:subgroups-noncomm-pairs}. Qudit stabilizer codes have also been studied previously in \cite{gunderman2023stabilizer,lei2023qudit} in the context of quantum codes with exotic local dimensions, qudit surface codes, and qudit hypermap codes, and some of our results in Section~\ref{sec:group_theory} could be applicable to those works also.
}








\section{Preliminaries}
\label{sec:prelim}

In this section, we introduce some background material, some definitions and notation to be used throughout in the paper. We also review some necessary facts from ring theory and establish a few basic results that are helpful later on. Throughout the paper, entries of matrices and vectors (both row and column) will follow zero-based indexing.

\subsection{Qudit Pauli group}
\label{ssec:qudit-pauli-group}

Consider a collection of $n$ qudits, each with dimension $d$. Thus each qudit corresponds to a complex $d$-dimensional Hilbert space $\mathbb{C}^d$ with a basis of orthonormal states:  $\ket{0},\ket{1},\dots,\ket{d-1}$. The Pauli operators on one qudit are the invertible operators $X$ and $Z$ from Eq.~\eqref{eq:qudit_paulis} which we will think of concretely as being given by their $d \times d$ matrix representations with respect to the chosen basis. Thus $X,Z \in \mathbb{C}^{d \times d}$, and notice that $X^d = Z^d = (-1)^{d-1}(XZ)^d = I$, where $I$ is the $d \times d$ identity matrix. We will always use $I$ to denote a context-appropriately-sized identity matrix and specify its shape only when there is a chance for confusion. No smaller number $0<a<d$ results in $X^a$, $Z^a$, or $(XZ)^a$ even being proportional to $I$.

To define Pauli operators on $n$ qudits, start with the disjoint union space $\mathcal{Z} := \bigsqcup_{n=1}^{\infty} \mathbb{Z}_d^{2n}$, and the map $P: \mathcal{Z} \rightarrow \bigsqcup_{n=1}^{\infty} \mathbb{C}^{d^n \times d^n}$ by
\begin{equation}
P(v):=X^{v_0}Z^{v_{n}}\otimes X^{v_1}Z^{v_{n+1}}\otimes\dots\otimes X^{v_{n-1}}Z^{v_{2n-1}} \in \mathbb{C}^{d^n \times d^n}, \;\; v\in\mathbb{Z}_d^{2n}.
\end{equation}
Then the set of $n$-qudit Pauli operators is $\mathcal{P}_{n} :=\{P(v):v\in\mathbb{Z}_d^{2n}\} \subseteq \mathbb{C}^{d^n \times d^n}$. We call an $n$-qudit Pauli $P(v)$ $X$-type if $v_{n}=v_{n+1}=\dots=v_{2n-1}=0$, and $Z$-type if $v_0=v_1=\dots=v_{n-1}=0$. These are abbreviated as $X(u)$ and $Z(u)$ respectively, where $u\in\mathbb{Z}_d^n$ specifies just the nonzero entries of $v$.

The set $\mathcal{P}_{n}$ does not form a group. However, the set $\overline{\mathcal{P}}_n := \{\omega^j P(v) : v \in \mathbb{Z}_d^{2n}, j \in \mathbb{Z}_d\}$ does form a group with the operation of matrix multiplication, called the \textit{Heisenberg-Weyl Pauli} group. Defining $K=\{\omega^j I:j=0,1,\dots,d-1\}$, we may identify $\mathcal{P}_n$ in one-to-one fashion with the quotient group $\overline{\mathcal{P}}_n/K$ in the obvious way by ignoring phase factors.

Many of our results arise from interest in commutativity relations of elements of $\overline{\mathcal{P}}_n$. For this, we start by noticing that $XZ = \omega^{-1}ZX$. Defining the group commutator of two invertible matrices $[A,B] := ABA^{-1}B^{-1}$, we then have $[X,Z]=\omega^{-1}I$. This implies that for any two elements $\omega^j P(u)$ and $\omega^k P(v)$ of $\overline{\mathcal{P}}_n$
\begin{equation}
\label{eq:Lambda-def}
[\omega^j P(u), \omega^k P(v)] = [P(u),P(v)]=\omega^{u^T\Lambda v}I,\quad \Lambda=\left(\begin{array}{cc}0&-I\\I&0\end{array}\right),
\end{equation}
where $\Lambda$ is written in four $n\times n$ blocks. The first equality of this equation suggests that it suffices to work with the set $\mathcal{P}_n$ for studying commutation relations, and this is largely what we will do outside of Section~\ref{sec:group_theory}, where some more notation will be introduced just for the material in that section. To isolate the phase, we also define the notations
\begin{equation}\label{eq:commutator}
\llbracket P(u),P(v)\rrbracket := u^T\Lambda v\in\mathbb{Z},\quad \llbracket P(u),P(v)\rrbracket_d := u^T\Lambda v \mod{d} \in\mathbb{Z}_d.
\end{equation}
Two Pauli operators $P(u), P(v) \in \mathcal{P}_n$ are said to commute if and only if $\llbracket P(u),P(v)\rrbracket_d = 0$, because that implies $P(u)P(v) = P(v)P(u)$.

We now introduce some key definitions that underlies much of our results. First, in Section~\ref{sec:max-pairs}, we determine the maximum size of collections of non-commuting pairs.

\begin{definition}
\label{def:noncomm-pairs}
A collection of $k$ \textit{non-commuting pairs} on $n$ qudits consists of two ordered sets $S=\{s_0,s_1,\dots,s_{k-1}\}$, and $T=\{t_0,t_1,\dots,t_{k-1}\}$, where $S,T\subseteq \mathcal{P}_{n}$, for all $i,j\in\{0,1,\dots,k-1\}$, $\llbracket s_i,s_j\rrbracket_d=0$ and $\llbracket t_i,t_j\rrbracket_d=0$, and moreover $\llbracket s_i,t_j\rrbracket_d=0$ if and only if $i\neq j$. A collection of $k$ \textit{non-commuting CSS pairs} is a collection of non-commuting pairs where every element of $S$ is $X$-type and every element of $T$ is $Z$-type.
\end{definition}

Thus, in a collection of $k$ non-commuting pairs, there are a total of $2k$ Paulis, and they all commute except for the $k$ pairs $(s_i,t_i)$. An easy consequence of this definition is that all elements of $S$ and $T$ are necessarily distinct from one another and not equal to $I$.

Later, in Section~\ref{sec:achievable-patterns}, we will also be concerned with what values the commutators of a collection of non-commuting pairs can have. 
\begin{definition}
\label{def:achievable-relations}
Suppose we are given $n$ qudits, and a $k$-tuple of numbers $f := (f_0,\dots,f_{k-1}) \in (\mathbb{Z}_d \setminus \{0\})^{k}$. We say that $f$ is an \textit{achievable non-commuting pair relation} on $n$ qudits if there exist non-commuting pairs $S=\{s_0,s_1,\dots,s_{k-1}\}$ and $T=\{t_0,t_1,\dots,t_{k-1}\}$ of size $k$, such that $\llbracket s_i,t_i\rrbracket_d = f_i$, for all $i=0,\dots,k-1$. We say that $S, T$ \textit{achieves} $f$.
\end{definition}





In Section~\ref{ssec:noncomm-set-max-size}, we count the elements of non-commuting sets.

\begin{definition}
\label{def:noncommuting-set}
A non-empty subset $S \subseteq \mathcal{P}_n$ is called a \textit{non-commuting set} on $n$ qudits if and only if $\llbracket p,q\rrbracket_d \neq 0$ whenever $p,q \in S$ are distinct.
\end{definition}

It should be clear from this definition that a non-commuting subset $S \subseteq \mathcal{P}_n$ cannot contain $I$ and all elements of $S$ must be distinct.

\subsection{Modules over commutative rings}
\label{ssec:modules-rings}

Our next goal is to provide a short review of some basic results in linear algebra over the commutative ring $\mathbb{Z}_d$ with $d \geq 2$, for the reader who is unfamiliar with these notions. The main difficulties in working with $\mathbb{Z}_d$ and matrices over $\mathbb{Z}_d$ arise when $d$ is composite, because then and only then is $\mathbb{Z}_d$ a ring and not a field. We introduce some ring-theoretic terminology here that will find use in our proofs, and some well-known canonical matrix decompositions such as the Smith normal form and the Howell normal form, which for certain applications generalize the utility of the singular value decomposition and reduced row echelon form, respectively, to certain classes of commutative rings. Separately, we dedicate a section (Appendix~\ref{app:alternating-smith-normal-form}) to the alternating Smith normal form, which is central to this paper and provides a symmetry-respecting Smith normal form for matrices with alternating symmetry, to be defined later. The related material is mostly borrowed from \cite[Chapters~II, III, X, XIII, XV]{lang2012algebra}, and \cite[Chapters~I-V, XIV,XV]{brown1993matrices}.

\subsubsection{Basic definitions}
\label{sssec:basic-def}

Generically, we denote a ring by $R$, and assume, unless stated otherwise, it is non-zero, commutative, and contains a multiplicative identity $1$ and additive identity $0$. Largely, we work over the ring of integers $\mathbb{Z}$ or the ring of integers modulo $d$, $\mathbb{Z}_d$. If $a,b \in R$ such that $a = bc$ for some $c \in R$, we say $b$ \textit{divides} $a$, and we write $b \mid a$. An element $x \in R$ is called \textit{invertible} (or \textit{unit}) if it has a multiplicative inverse. The subset of units of a ring form a multiplicative group called the \textit{group of invertible elements}. Ring $R$ is called a \textit{field} if every non-zero element is a unit. $R$ is called an \textit{integral domain} if $ab = 0$ for $a,b \in R$, implies either $a=0$ or $b=0$. All finite integral domains, such as $\mathbb{Z}_d$ for prime $d$, are fields by Wedderburn's little theorem \cite{kaczynski1964another}, but this is not so for infinite rings, like $\mathbb{Z}$. For a ring $R$ that is not an integral domain, an element $0 \neq x \in R$ is called a \textit{zero divisor} if $xy=0$ for some $R \ni y \neq 0$. The set of zero divisors of $R$ is denoted $\mathcal{Z}(R)$.

An ideal $J$ is an additive subgroup of $R$ that is also closed under multiplication by $R$, i.e.~$J=RJ:=\{rj:j\in J,r\in R\}$. An ideal $J \subseteq R$ is called a \textit{principal ideal} if it is generated by a single element, i.e. $J = xR := \{xy : y \in R\}$, for some $x \in R$. A commutative ring $R$ is called a \textit{principal ideal ring} (or PIR), if all its ideals are principal ideals.  It is well-known that $\mathbb{Z}_d$ is a PIR \cite[Chapter~2.2, Ex.~8]{ash2013basic}. An integral domain that is also a PIR is called a \textit{principal ideal domain} (or PID). For example $\mathbb{Z}$ is a PID, as it can be shown that any ideal of $\mathbb{Z}$ is of the form $x\mathbb{Z}$ (hence a principal ideal), for some $x \in \mathbb{Z}$.

Given a ring $R$, a \textit{left module} over $R$, or simply a left $R$-module $M$ is an abelian group together with an operation $\cdot : R \times M \rightarrow M$, such that for all $a,b \in R$ and $x,y \in M$ we have (i) $(a + b) \cdot x = a \cdot x + b \cdot x$, (ii) $a \cdot (x + y) = a \cdot x + a \cdot y$, (iii) $(ab) \cdot x = a \cdot (b \cdot x)$, and (iv) $1 \cdot x = x$. One can analogously define a \textit{right module} over $R$, but any right $R$-module can be viewed as a left $R$-module and vice-versa, and for us this distinction will not be needed. Thus when the ring $R$ is understood, we will simply say ``a module $M$''. A ring $R$ is also a $R$-module, with the $\cdot$ operation being the same as ring multiplication. The most important modules we will have to deal with are the $\mathbb{Z}_d$-module $\mathbb{Z}_d^k$ and the $\mathbb{Z}$-module $\mathbb{Z}^k$, i.e.~the set of all $k$-tuples of elements of the ring $\mathbb{Z}_d$ and $\mathbb{Z}$ respectively, where the operation $\cdot$ is that of ring multiplication applied to each component of the tuple. Given a $R$-module $M$, a subset $S \subseteq M$ is called a \textit{generating set} (over $R$) of $M$ if any $y \in M$ can be expressed as $y = \sum_{i=1}^{k} a_i x_i$ for some finite $k$, $a_1,\dots,a_k \in R$, and $x_1,\dots,x_k \in S$. We say that $M$ is \textit{finitely generated} if it has a finite generating set. A non-empty subset $S \subseteq M$ is called \textit{linearly independent} (over $R$) if and only if the following condition holds for every finite positive integer $k$, $a_1,\dots,a_k \in R$, and distinct elements $x_1,\dots,x_k \in S$:
\begin{equation}
    \sum_{i=1}^{k} a_i \cdot x_i = 0 \implies a_i = 0 \;\; \text{for every } i=1,\dots,k.
\end{equation}
A linearly independent subset $S \subseteq M$ that is also a generating set of $M$, is called a \textit{basis} of $M$. We say that a $R$-module $M$ is a \textit{free $R$-module} if $M$ has a basis, or is the zero module. For example, it is easy to show that both $\mathbb{Z}_d^k$ and $\mathbb{Z}^k$ are finitely generated and free: in both cases the set $S = \{e_1,e_2,\dots,e_k\}$ is a basis for these modules, where $(e_j)_i = 1$ if $i=j$, else $0$.

\subsubsection{Minimal generating sets}
\label{sssec:min-gensets}

A commutative ring $R$ with a multiplicative identity element is an \textit{invariant basis number ring} (see \cite[Chapter~4, Definition~2.8]{hungerford2012algebra} and \cite[Section~1A]{lam2012lectures} for a discussion on invariant basis number rings). This means that any finitely generated free $R$-module $M$ has a unique finite dimension given by the size of any (hence all) basis of $M$, which we will call the \textit{rank} of $M$ where, by convention, the zero module has zero rank. For a proof of this statement, the reader is referred to \cite[Corollary~5.13]{brown1993matrices}, and the discussions thereafter. For example, the modules $\mathbb{Z}_d^k$ and $\mathbb{Z}^k$ discussed above both have rank $k$. Given a $R$-module $M$, a subset $M' \subseteq M$ is called a \textit{submodule} if it is also a $R$-module. A key fact is that a submodule of a free module need not be free; for example, consider the submodule $\{0, 2\}$ of $\mathbb{Z}_4$ which does not have a basis. Thus, apriori, the rank of the submodule $M'$ is not defined. 

In order to be able to extend the notion of rank to submodules of a finitely generated free $R$-module, we need the notion of minimal generators. Given a finitely generated $R$-module $M$, a generating set $S \subseteq M$ is called a \textit{minimal generating set} of $M$ if there does not exist another generating set $S' \subseteq M$ such that $|S'| < |S|$. The quantity $|S|$ is called the \textit{minimal number of generators} of $M$, and we will denote it as $\Theta(M)$. For example, for the submodule $\{0,2\}$ of $\mathbb{Z}_4$, the generating set $\{2\}$ is a minimal generating set, and so $\Theta (\{0,2\})=1$. If $M$ is free, it is a well-known result that $\Theta (M)$ equals the rank of $M$, since $R$ is commutative. So $\Theta(M)$ is indeed a generalization of rank to non-free modules.

Our main interest in the minimal number of generators is motivated by the need to define a quantity analogous to the column rank of a matrix with entries in a field, for the situation we encounter in this paper where the entries of the matrix belong to a commutative ring $R$ (in particular $\mathbb{Z}$ or $\mathbb{Z}_d$). Suppose we have a matrix $C \in R^{k \times t}$, and we label each of its $t$ columns as $c_0,c_1,\dots,c_{t-1}$. Each column is an element of the $R$-module $R^{k}$, and together they generate a submodule of $R^{k}$, namely $M_{C} := \left\{\sum_{i=0}^{t-1} a_i c_i : a_i \in R \; \forall \; i \right\}$. Now suppose $A \in R^{k \times k}$ and $B \in R^{t \times t}$ are invertible matrices, and recall that a square matrix with entries in $R$ is invertible if and only if the determinant of the matrix is a unit in $R$ \cite[Corollary~2.21]{brown1993matrices}. Let $\bar{C} := ACB \in R^{k \times t}$, let its columns be $\bar{c}_0,\bar{c}_1,\dots,\bar{c}_{t-1}$, and let $M_{\bar{C}}$ be the submodule of $R^k$ generated by the columns of $\bar{C}$. The important property of the minimal number of generators that we will need is the following lemma, which is well-known (although we also provide a proof in Appendix~\ref{app:proofs_sec2}):

\begin{lemma}
\label{lem:min-genarators-modules}
The minimal number of generators of the submodules $M_{C}$ and $M_{\bar{C}}$ of $R^k$ are equal.
\end{lemma}

We also note an important special case in the next lemma, still assuming that $R$ is a commutative ring with a multiplicative identity, where the minimal number of generators for a matrix with special structure is known. We thank Luc Guyot for an outline of this proof in \cite{guyot}. See Appendix~\ref{app:proofs_sec2} for the proof.

\begin{lemma}
\label{lem:min-number-gens-special-matrix}
Let $C \in R^{k \times t}$ be a matrix such that each row and column has at most one non-zero element, and suppose one of those non-zero elements is divisible by all the others. Then the minimal number of generators of the $R$-module $M_{C}$, generated by the columns of $C$, is equal to the number of non-zero elements of the matrix $C$.
\end{lemma}

\subsubsection{Smith normal form}
\label{sssec:smith-normal-forms}

Let $R$ be a commutative ring with multiplicative identity. Following \cite[Chapter~15]{brown1993matrices}, $R$ is called an \textit{elementary divisor ring} if for every $k, \ell \geq 1$ and for every $C \in R^{k \times \ell}$, there exists invertible matrices $B \in R^{k \times k}$ and $\tilde{B} \in R^{\ell \times \ell}$, such that (i) $A := B C \tilde{B}$ is a diagonal matrix, and (ii) $a_0 \mid a_1 \mid \dots \mid a_{r-1}$, where $r= \min \{k,\ell\}$, and $a_i := A_{ii}$ for every $i$. The diagonal matrix $A$ is called the \textit{Smith normal form} (SNF) of the matrix $C$, and we say that $C$ admits a \textit{diagonal reduction}. The non-zero elements of the diagonal matrix $A$ are called the \textit{invariant factors} of $C$.

We recall the following well-known result about the Smith normal form of a matrix with entries from a principal ideal ring $R$:

\begin{theorem}[{\cite[Theorems~15.9, 15.24]{brown1993matrices}}]
\label{thm:smith-normal-form}
If a commutative ring $R$ is a principal ideal ring, then for all $k, \ell \geq 1$, every matrix $C \in R^{k \times \ell}$ admits a diagonal reduction $A := B C \tilde{B}$, where $B \in R^{k \times k}$ and $\tilde{B} \in R^{\ell \times \ell}$ are invertible matrices, and $A$ is the Smith normal form of $C$. The Smith normal form $A$ is unique up to multiplication of the diagonal elements of $A$ by units of $R$.
\end{theorem}

The following result about the SNF of a matrix and its relation to the minimal number of generators of the module generated by its columns, will be useful in Section~\ref{sec:group_theory}, and follows as a direct consequence of Lemma~\ref{lem:min-genarators-modules} and Lemma~\ref{lem:min-number-gens-special-matrix}:

\begin{lemma}
\label{lem:snf-nonzeros}
Let $R$ be a commutative principal ideal ring. Suppose $C \in R^{k \times \ell}$ has a Smith normal form $A = B C \tilde{B}$, for invertible matrices $B,\tilde{B}$ and $A$ a diagonal matrix of appropriate shape. Let $M_C$ be the submodule generated by the columns of $C$. Then $\Theta(M_C)$ equals the number of non-zero elements of $A$.
\end{lemma}

Choosing integer $n$ so that $k,\ell,r=O(n)$, finding the Smith normal form of a matrix takes $\tilde{O}(n^{\theta+1})$ time \cite{storjohann2000algorithms} where $\tilde{O}$ hides log factors and assuming $2<\theta\le3$ is the exponent of matrix multiplication.


\subsubsection{Howell normal form}
\label{sssec:howell-normal-forms}
 
We now introduce the Howell normal form of a matrix with entries in the ring $\mathbb{Z}_d$, which will be useful for us in Section~\ref{sec:group_theory}. The material here is borrowed from \cite{howell1986spans,webster2022xp}. Although we only discuss the $\mathbb{Z}_d$ case here, the Howell normal form is also valid for a matrix over a principal ideal ring, and the reader is referred to \cite{storjohann2000algorithms,fieker2014computing} for details. Also note that we present the Howell normal form in terms of column spans of a matrix, instead of row spans as done in prior literature (of course they are equivalent).

A matrix $H \in \mathbb{Z}_d^{k \times \ell}$ is said to be in \textit{reduced column echelon} form if $H$ satisfies the following properties:
\begin{enumerate}[(i)]
    \item If $H$ has $r$ non-zero columns, then its first $r$ columns are non-zero.
    \item For $0 \leq i \leq r-1$, let $j_i$ be the row index of the first non-zero entry of column $i$. Then $j_0 < j_1 < \dots < j_{r-1}$.
    \item For $0 \leq i \leq r-1$, the matrix entry $H_{j_i i}$ is a divisor of $d$ over integers.
    \item For each $0 \leq t < i \leq r-1$, we have that $0 \leq H_{j_i t} < H_{j_i i}$.
\end{enumerate}
If the matrix $H$ only satisfied properties (i) and (ii) above, then it can be made to satisfy properties (iii) and (iv) also by post-multiplying $H$ on the right by an invertible matrix over $\mathbb{Z}_d$. We say that $H$ is in \textit{Howell normal form} (HNF), if in addition to properties (i)-(iv), the matrix $H$ satisfies the following additional property:
\begin{enumerate}[(v)]
    \item Let $M_H$ be the submodule of $\mathbb{Z}_d^k$ generated by the columns of $H$. For each $0 \leq i \leq r-1$, consider the submodule $\overline{M}_i := \{v \in M_H: v_t = 0, \; \forall \; 0 \leq t < j_i\}$. Then $\overline{M}_i$ is generated by columns $i$ to $r-1$ of $H$, for every $0 \leq i \leq r-1$.
\end{enumerate}

Let us now state the main result about the Howell normal form. Given $A \in \mathbb{Z}_d^{k \times \ell}$, we construct $\overline{A} := [A \;\; 0]$, by adding $k$ zero columns to $A$. Then there exists an invertible matrix $L \in \mathbb{Z}_d^{(\ell+k) \times (\ell+k)}$, such that $\overline{A}L=H$, where $H\in\mathbb{Z}_d^{k\times(k+\ell)}$ is a matrix in Howell normal form. Moreover, the matrix $H$ is uniquely determined. 

As part of the calculation of the Howell normal form, some algorithms (e.g.~\cite{storjohann2000algorithms}) include computation of a matrix $G\in\mathbb{Z}_d^{(k+\ell)\times(k+\ell)}$, whose columns generate the kernel of $H$, $\text{Ker}(H)=\{x\in\mathbb{Z}_d^{\ell}:Hx=0\}$. So $HG=0$. Moreover, the columns of $LG$ provides a generating set for the kernel of $A$. Choosing integer $n$ so that $k,\ell,r=O(n)$, finding both $G$ and $H$ takes $\tilde{O}(n^{\theta+1})$ time.

With the Howell normal form $H$ and kernel $G$, we can test membership in the column space of $A$, or equivalently solve linear systems, $\overline{A}x=b$, where we are given $b\in\mathbb{Z}_d^k$ and want to determine whether $x\in\mathbb{Z}^{k+\ell}$ exists. Note that by invertibility of $L$, this is the same as finding a $y=L^{-1}x$ such that $Hy=b$. Also, if $y$ is a solution then so is $y+Gv$ for any $v\in\mathbb{Z}_d^{\ell}$. Solving the inhomogenous equation $Hy=b$ is made easy by properties (i-v). In principle, one can observe, \cite{storjohann2000algorithms}
\begin{equation}
\bigg(\underbrace{\begin{array}{c|c}1&0\\\hline b&H\end{array}}_{H_b}\bigg)\bigg(\underbrace{\begin{array}{c|c}1&0\\\hline -y&I\end{array}}_{L_b}\bigg)=\bigg(\begin{array}{c|c}1&0\\\hline b'&H\end{array}\bigg)
\end{equation}
where the righthand-side is the Howell normal form of $H_b$ and the invertible matrix $L_b$ relating them necessarily has the indicated form. Now, $b'=0$ if and only if $Hy=b$. In practice, one does not need to run the Howell normal form algorithm on $H_b$, but can instead achieve the same result by solving $Hy=b$ one entry at a time using the fact $H$ is in reduced row echelon form and has the Howell property (v).
\section{Maximum number of non-commuting pairs of Paulis}
\label{sec:max-pairs}

In this section, we give a precise count of the maximum number of non-commuting pairs that one can achieve with $n$ qudits. This generalizes the well-known result in the qubit ($d=2$) case, where this count is equal to the number of qubits \cite{gottesman1997stabilizer} \editt{and prepares for our construction in Section~\ref{sec:achievable-patterns} of Pauli sets satisfying given commutation relations using the minimum number of qudits. }

It is easy to establish the lower bound $nm$ \editt{on the maximum number of non-commuting pairs} as illustrated by the following example.

\begin{example}
\label{ex:one-qudit-example}
Define the sets $S=\{s_i=X^{d/p_i^{\alpha_i}}:i\in\{0,1,\dots,m-1\}\}$ and $T=\{t_i=Z^{d/p_i^{\alpha_i}}:i\in\{0,1,\dots,m-1\}\}$. Then $S, T$ is a collection of $m$ non-commuting CSS pairs on a single qudit. On $n$ qudits, this example can be applied individually to each qudit to get a collection of non-commuting CSS pairs of size $nm$.
\end{example}

We now show that $nm$ is also the upper bound. The key result that underlies the proof is the following number-theoretic obstruction:

\begin{lemma}[Obstruction lemma]
\label{lem:obs-lem}
Suppose $\ell = p^\gamma$, for some prime number $p$, and positive integer $\gamma \geq 1$. For some $k \geq 2$, let $A$ be a $k \times k$ matrix with entries in $\mathbb{Z}$. Then the following conditions cannot hold simultaneously
\begin{enumerate}[(i)]
    \item $\det(A) = 0$,
    \item $A$ contains exactly $k$ entries that are not divisible by $\ell$, such that each row and column of $A$ contains exactly one such element.
\end{enumerate}
\end{lemma}

\begin{proof}
For contradiction, assume that such a matrix $A$ exists satisfying both conditions (i) and (ii). Let us denote the columns of $A$ as $c_0,\dots,c_{k-1}$. We claim that there exist integers $x_0,\dots,x_{k-1}$, not all zeros, such that $\sum_{i=0}^{k-1} x_i c_i = 0$, which we prove at the end. Now pick any row $j$. This row contains an element $A_{jq}$ which is not divisible by $\ell$, while all other elements of this row are divisible by $\ell$. From the condition $\sum_{i=0}^{k-1} x_i c_i = 0$, we thus we have 
\begin{equation}
\label{eq:mat-fact-1}
    x_q A_{jq} = - \sum_{\substack{i=0 \\ i\neq q}}^{k-1} x_i A_{ji}.
\end{equation}
The right hand side of Eq.~\eqref{eq:mat-fact-1} is divisible by $\ell$, and $A_{jq}$ is not divisible by $\ell$, so we deduce that $x_q$ is divisible by $p$. By repeating this argument for each row, and using condition (ii), we deduce that $x_i$ is divisible by $p$ for each $i=0,\dots,k-1$. Defining $x'_i = x_i / p$ for each $i$, we now obtain $\sum_{i=0}^{k-1} x'_i c_i = 0$. We can keep repeating the argument, giving an infinite sequence of integers
\begin{equation}
    x_i, \frac{x_i}{p}, \frac{x_i}{p^2}, \frac{x_i}{p^3},\dots, \;\;\; \text{ for every } i = 0,\dots,k-1.
\end{equation}
But this implies $x_i = 0$ for every $i$, giving a contradiction.

We now prove the claim made in the previous paragraph. We may regard $A$ as a matrix over the set of rational numbers $\mathbb{Q}$, and then by condition (i) we have that $\det(A)=0$ over $\mathbb{Q}$. Thus the columns $c_0,\dots,c_{k-1}$ are linearly dependent as column vectors over the field $\mathbb{Q}$, which implies there exists $y_0,\dots,y_{k-1} \in \mathbb{Q}$, not all zeros, such that $\sum_{i=0}^{k-1} y_i c_i = 0$. Multiplying through by the least common denominator of the $y_i$ immediately proves the claim.
\end{proof}

We can now prove the following theorem:
\begin{theorem}[\edit{Largest size non-commuting pairs}]
\label{thm:largest-noncomm-pairs}
The largest size of a collection of non-commuting pairs on $n$ qudits is $nm$.
\end{theorem}
\begin{proof}
Suppose by way of contradiction that we have a collection of $k>nm$ non-commuting pairs $S=\{s_0,\dots,s_{k-1}\}$ and $T=\{t_0,\dots,t_{k-1}\}$. Then for every $j$, $\llbracket s_j,t_j\rrbracket$ is not divisible by $d$. By the pigeonhole principle, there is some prime $p$ (from the factorization of $d$) and a set $J \subseteq \{0,1,\dots,k-1\}$ of $n+1$ indices so that for each $j\in J$, $\llbracket s_j,t_j\rrbracket$ is not divisible by $p^\alpha$, where $\alpha$ is the largest integer such that $p^\alpha$ divides $d$. Without loss of generality, one may assume that $J=\{0,1,\dots,n\}$.

Form the matrices $M^{(1,1)}_{ij}=\llbracket s_i,s_j\rrbracket\edit{=0}$, $M^{(1,2)}_{ij}=\llbracket s_i,t_j\rrbracket$, $M^{(2,1)}_{ij}=\llbracket t_i,s_j\rrbracket$, and $M^{(2,2)}_{ij}=\llbracket t_i,t_j\rrbracket\edit{=0}$ for $i,j\in J$. Write $M=\left(\begin{smallmatrix}M^{(1,1)}&M^{(1,2)}\\M^{(2,1)}&M^{(2,2)}\end{smallmatrix}\right)\in\mathbb{Z}^{2(n+1)\times2(n+1)}$. We note that every row and column of $M$ contains exactly one element not divisible by $p^\alpha$ (i.e.~the diagonal elements of $M^{(1,2)}$ and $M^{(2,1)}$).

However, we also claim $\det(M)=0$. Thus, Lemma \ref{lem:obs-lem} gives our contradiction. To prove that the determinant is zero, we take vectors $u_0, \dots, u_{2n+1} \in\mathbb{Z}^{2n}_d$, so that $s_i=P(u_i)$ and $t_i=P(u_{i+n+1})$, for every $i=0,\dots,n$. Then $M_{ij}=u_i^T\Lambda u_j \in \mathbb{Z}$, where $\Lambda$ is defined in Eq.~\eqref{eq:Lambda-def}. We now lift $M$ and the vectors $u_i$ to the field of real numbers, i.e. we may treat $M \in \mathbb{R}^{2(n+1) \times 2(n+1)}$, and $u_i \in \mathbb{R}^{2n}$ for every $i$. Because there are $2(n+1)$ vectors $u_i$ for $i=0,\dots,2n+1$, it now follows that there are some real numbers $x_0,\dots,x_{2n+1}$, such that $\sum_{i=0}^{2n+1} x_i u_i = 0$. From this it follows that the columns of $M$ (and the rows) are also linearly dependent and $\det(M)=0$ over the reals, and hence also over $\mathbb{Z}$.
\end{proof}

Note a consequence of the theorem above:
\begin{corollary}
Let $S,T$ be a collection of non-commuting pairs of size $nm$ on $n$ qudits. Let $G = \{r \in \mathcal{P}_n : \llbracket r,r'\rrbracket_d = 0, \;\forall \; r' \in S \cup T \}$. Then $G$ does not contain a non-commuting pair of Paulis.
\end{corollary}

\begin{proof}
If $G$ contained such a non-commuting pair $(s,t)$, then the ordered sets $S' := S \cup \{s\}$ and $T':=T \cup \{t\}$ would be a non-commuting pair of size $nm + 1$, which contradicts Theorem~\ref{thm:largest-noncomm-pairs}.
\end{proof}

\section{Maximum size of a non-commuting set}
\label{ssec:noncomm-set-max-size}

It is well-known in the qubit case $(d=2)$, that the maximum size of an anticommuting set on $n$ qubits is $2n+1$ (see \cite[Lemma~8]{sarkar2021sets}, \cite[Appendix~G]{bonet2020nearly}, \cite[Theorem~1]{hrubevs2016families}). In \editt{Section~\ref{sssec:one-qudit-case}}, we generalize this result to a single qudit of arbitrary dimension $d$, and in Section~\ref{sssec:general-bounds} we provide bounds for the multi-qudit case. \editt{The case of constructing Pauli sets satisfying the commutation relations of parafermion operators \cite{gungordu2014parafermion} is a special case of the results in this section, which we specifically solve later in Corollary~\ref{cor:non-comm-all-the-same}.}


\subsection{Non-commuting sets on a single qudit}
\label{sssec:one-qudit-case}

In this section, we evaluate $\gcd$ of elements of $\mathbb{Z}$ (or $\mathbb{Z}_d$ considered as elements of $\mathbb{Z}$) over $\mathbb{Z}$ and define it to be positive to ensure that it is unique. Also whenever we write $a/b$ for $a,b \in \mathbb{Z}_d$ (or $\mathbb{Z}$), we will always mean division over integers (in particular this must mean that $b$ divides $a$ over integers). Our goal in this section is to prove the following theorem:

\begin{theorem}[\edit{Largest size non-commuting set on 1 qudit}]
\label{thm:largest-noncomm-set-one-qudit}
Let $d = p_0^{\alpha_0}p_1^{\alpha_1}\dots p_{m-1}^{\alpha_{m-1}}$ be the prime factorization of $d \geq 2$. Then the size of a largest non-commuting subset of $\overline{\mathcal{P}}_1$ on one $d$-dimensional qudit is given by the Dedekind psi function $\Psi(d) := d \prod_{j=0}^{m-1}(1+1/p_j)$.
\end{theorem}
The Dedekind psi function is bounded by $\Psi(d)\ge d+1$, with equality when $d$ is prime, and $\Psi(d)<e^\gamma d\log\log(d)$ for all $d>30$ \cite{sole2010extreme}. 

Recall from the paragraph below Definition~\ref{def:noncommuting-set} that it suffices to find the size of a largest non-commuting subset of \edit{$\mathcal{P}_1$}. The basic strategy of the proof is now to reduce this problem to finding a maximum clique\footnote{In an undirected graph $G(V,E)$, a clique is a subset of vertices with the property that any two distinct vertices have an edge connecting them.} in a graph. To achieve this, we start by noticing that if $\left(\begin{smallmatrix}a \\ b\end{smallmatrix}\right), \left(\begin{smallmatrix}s \\ t\end{smallmatrix}\right) \in \mathbb{Z}_d^2$, then for the Paulis $P(\left(\begin{smallmatrix}a \\ b\end{smallmatrix}\right))$ and $P(\left(\begin{smallmatrix}s \\ t\end{smallmatrix}\right))$ on $1$ qudit represented by these vectors, Eq.~\eqref{eq:commutator} simplifies to
\begin{equation}
\llbracket P(\left(\begin{smallmatrix}a \\ b\end{smallmatrix}\right)),P(\left(\begin{smallmatrix}s \\ t\end{smallmatrix}\right))\rrbracket_d = -\det\left(\begin{array}{cc}a&b\\s&t\end{array}\right) \; \bmod{d}.
\end{equation}
As a consequence, finding the largest set of non-commuting Paulis on 1 qudit is equivalent to finding a maximum clique $C_m$ in the undirected graph $G=(V,E)$ with vertices $V=(\mathbb{Z}_d\times\mathbb{Z}_d) \setminus \{(0,0)\}$ and edge $((a,b),(s,t))\in E$ if and only if $\det\left(\begin{smallmatrix}a&b\\s&t\end{smallmatrix}\right) \not \equiv 0 \mod{d}$. This commutation graph (or its complement) has been studied before in relation to projective geometry in e.g.~\cite{planat2008pauli,planat2011pauli,havlicek2007projectiveB} and mutually unbiased bases \cite{planat2007qudits}. We will show in several steps below that $|C_m| = \Psi(d)$, thus proving Theorem~\ref{thm:largest-noncomm-set-one-qudit}.

The first step is to define a special subset of vertices $W\subseteq V$, where $(a,b)\in W$ if and only if $\gcd(a,b,d)=1$. \edit{We prove the following Lemma in Appendix~\ref{app:proofs_sec4}.}

\begin{lemma}\label{lem:G_properties}
The graph $G=(V,E)$ has the following properties.
\begin{enumerate}[(i)]
\item If $C\subseteq V$ is a clique in $G$, then there is a clique $C'\subseteq W$ of the same size, $|C|=|C'|$.  
\item If $v_0,v_2\in V$, $v_1\in W$, $(v_0,v_1)\not\in E$, and $(v_1,v_2)\not\in E$, then $(v_0, 
v_2)\not\in E$.
\item If $(a,b)$ and $(s,t)$ are in $W$, then $((a,b),(s,t))\not\in E$ if and only if there exists $u\in\mathbb{Z}$ with $\gcd(u,d)=1$ such that $(a,b)=(us,ut)$, the right hand side evaluated modulo $d$.
\end{enumerate}
\end{lemma}

Lemma \ref{lem:G_properties}(i) implies there is a maximum clique that is just a subset of $W$. Part (ii) implies transitivity of an equivalence relation $u\sim v$ for $u,v\in W$, where $u$ and $v$ are said to be equivalent if $(u,v)\not\in E$. Part (iii) says equivalence classes of $W$ are all the same size, exactly the size of the group of units of $\mathbb{Z}_d$, which is $\phi(d)=d\prod_{j=0}^{m-1}(1-1/p_j)$, the Euler's totient function. Combining these facts, the size of a maximum clique of $G$ is $|C_m|=|W|/\phi(d)$. Thus the remaining step is to count $|W|$.

In the remainder of the proof of Theorem~\ref{thm:largest-noncomm-set-one-qudit}, we will adopt the following notation: any product of the form $\prod_{p \mid r}$, where $r$ is a positive integer, will mean that the product is taken over all unique primes $p$ that divide $r$. To count $|W|$, we let $(a,b)\in W$ and allow $a\in\mathbb{Z}_d$ to be arbitrary. Now $b\in\mathbb{Z}_d$ is restricted by the choice of $a$. Namely, if $\gcd(a,d)>1$, $b$ cannot be $0$ or share any prime factors with $\gcd(a,d)$. Each prime factor $p$ of $\gcd(a,d)$ reduces the allowable set of $b$ by a factor of $(p-1)/p=(1-1/p)$. We thus have the following equality over rationals:
\begin{equation}
\frac1d|W|=\sum_{a\in\mathbb{Z}_d}\prod_{p\mid\gcd(a,d)}\left(1-\frac1p\right)=: H(d).
\end{equation}

It turns out $H$ is multiplicative\edit{, which we show in Appendix~\ref{app:proofs_sec4}.}
\begin{lemma}\label{lem:H_is_multiplicative}
If $d_0,d_1\in\mathbb{Z}$ are relatively prime, then $H(d_0d_1)=H(d_0)H(d_1)$.
\end{lemma}
Therefore, it is easy to evaluate $H$ for prime powers: $H(p^k)=p^k(1-1/p^2)$. Consequently, $H(d)=d\prod_{p\mid d}(1-1/p^2)$. We note for curiosity's sake that $|W|=dH(d)=J_2(d)$ is the second of Jordan's totient functions \cite{dickson1919history}. With $H(d)$ in hand, we now conclude the calculation:
\begin{equation}
|C_m|=|W|/\phi(d)=dH(d)/\phi(d)=d\prod_{p\mid d}(1+1/p)=\Psi(d).
\end{equation}
This completes the proof of Theorem~\ref{thm:largest-noncomm-set-one-qudit}.

\subsection{\texorpdfstring{Bounds on the size of non-commuting sets on $n$ qudits}{}}
\label{sssec:general-bounds}

Suppose we have a non-commuting set $S=\{s_i:i=0,\dots,|S|-1\}$ on $n$ qudits and another non-commuting set $S'=\{s'_i:i=0,\dots,|S'|-1\}$ on $n'$ qudits. Then,
\begin{equation}
S\circ S':=\{s_i\otimes I:i=1,\dots,|S|-1\}\cup\{s_0\otimes s'_i:i=0,\dots,|S'|-1\}
\end{equation}
is a non-commuting set of size $|S|+|S'|-1$ on $n+n'$ qudits. We call this the Jordan-Wigner composition of sets $S$ and $S'$, in analogy to the Jordan-Wigner encoding of fermions into qubits \cite{jordan1993paulische}. The Jordan-Wigner composition gives a simple lower-bound on the size of non-commuting sets on $n$ qudits.

\begin{corollary}\label{cor:noncomm_lower_bound}
The size of a maximum non-commuting set on $n$ qudits is at least $(\Psi(d)-1)n+1$.
\end{corollary}
\begin{proof}
The Jordan-Wigner composition of $t$ non-commuting sets $S_i$, $i=0,\dots,t-1$, gives a non-commuting set of size $\sum_{i=0}^{t-1}|S_i|-(t-1)$. With $n$ qudits, each supporting a maximum size non-commuting set $S_i$ with size $|S_i|=\Psi(d)$ (Theorem~\ref{thm:largest-noncomm-set-one-qudit}), this gives a lower bound of $(\Psi(d)-1)n+1$ on the size of the maximum non-commuting set.
\end{proof}

If larger non-commuting sets are found in special cases, the lower bound can be improved. For instance, for $d=3$, maximum non-commuting sets on $n=1$ and $n=2$ qutrits have sizes $4$ and $7$ (the case $n=2$ was found via exhaustive computer search), respectively, matching the $dn+1$ lower bound that Corollary~\ref{cor:noncomm_lower_bound} implies. But for $n=3$, we have a computer-verified, size-13, maximum non-commuting set, where the Paulis are $P(c)$ for each column $c$ of the matrix
\begin{equation}
\left(\begin{array}{ccccccccccccc}
0&0&0&0&0&0&0&1&1&1&1&1&1\\
0&0&0&0&1&1&1&0&0&0&1&1&2\\
0&1&1&1&1&2&2&1&1&2&1&2&1\\
0&0&0&1&0&1&1&0&2&1&0&1&0\\
0&0&1&0&0&1&2&0&1&2&1&0&1\\
1&0&1&2&1&2&2&2&2&1&1&1&1
\end{array}\right).
\end{equation}
The Jordan-Wigner composition therefore implies a lower-bound on non-commuting sets of qutrit Paulis of $12\lfloor n/3\rfloor+3(n\text{ mod }3)+1$.

Similarly, for $d=4$ and $n=2$, a maximum non-commuting set has size $20$. This implies the lower bound of $19\lfloor n/2\rfloor+5(n\text{ mod }2)+1$ for the maximum size of non-commuting sets on $n$ qudits with $d=4$.
\section{\edit{Achievable commutation relations}}
\label{sec:achievable-patterns}

We established in Section~\ref{sec:max-pairs} that the maximum number of non-commuting pairs that can exist on $n$ qudits is $nm$. Thus, if $(f_0,\dots,f_{k-1}) \in (\mathbb{Z}_d\setminus\{0\})^{k}$ is an achievable non-commuting pair relation on $n$ qudits, we must have $1 \leq k \leq nm$. Our first goal in this section is to further characterize exactly what tuples $(f_0,\dots,f_{k-1})$ are achievable non-commuting pair relations on $n$ qudits. Then, we will use this characterization to construct Pauli multisets with given commutation relations.

\subsection{Achievable tuples}\label{ssec:achievable_tuples}

Note that if $k \leq n$, then every $f = (f_0,\dots,f_{k-1}) \in (\mathbb{Z}_d\setminus\{0\})^{k}$ is an achievable non-commuting pair relation on $n$ qudits. \edit{On one qubit, $\llbracket Z^{f_i},X\rrbracket_d=f_i$, and so we can use a separate qubit to achieve each of the $k$ non-commuting pair relations. See also Lemma~\ref{lem:helper-lem-1}(iii) for a more formal argument.}

Our first result obtains a lower bound on the minimum number of qudits needed to achieve a given non-commuting pair relation $(f_0,\dots,f_{k-1}) \in (\mathbb{Z}_d\setminus\{0\})^{k}$, that improves on the lower bound $\lceil k/m \rceil$ \edit{implied by Theorem~\ref{thm:largest-noncomm-pairs}}.

\begin{lemma}
\label{lem:lower-bound}
Let $S=\{s_0,s_1,\dots,s_{k-1}\}$ and $T=\{t_0,t_1,\dots,t_{k-1}\}$ be a collection of non-commuting pairs of size $k$ on $n$ qudits, and for every $i$, let $f_i := \llbracket s_i,t_i\rrbracket$. Define the multiset $F_j := \{f_i : f_i \not\equiv 0 \pmod{p_j^{\alpha_j}}, \; i \in \{0,\dots,k-1\}\}$, where $d = p_0^{\alpha_0}p_1^{\alpha_1}\dots p_{m-1}^{\alpha_{m-1}}$. Then $|F_j| \leq n$, for every $j = 0,\dots,m-1$, and thus $n \geq \max_j |F_j|$.
\end{lemma}

\begin{proof}
If $k \leq n$, then there is nothing to prove. So we assume $k > n$, and let $j$ be arbitrary. For the sake of contradiction, suppose that $\ell := |F_j| > n$, and without loss of generality, one may then further assume that $F_j = \{f_0,f_1,\dots,f_{\ell - 1}\}$. Define $J:=\{0,\dots,\ell-1\}$, and then define the matrix $M \in \mathbb{Z}^{2\ell \times 2 \ell}$, as in the proof of Theorem~\ref{thm:largest-noncomm-pairs}. Then by exactly the same argument there, and using $\ell > n$, we conclude that (i) $\det(M)=0$, and (ii) every row and column of $M$ has exactly one element not divisible by $p_j^{\alpha_j}$. But this is impossible by Lemma~\ref{lem:obs-lem}, which gives us the desired contradiction.
\end{proof}

\begin{remark}
In Lemma~\ref{lem:lower-bound}, one can replace $f_i := \llbracket s_i,t_i\rrbracket$ by $f_i := \llbracket s_i,t_i\rrbracket_d$, and the conclusion still holds. This is because, for every $i$ and $j$, we have $\llbracket s_i,t_i\rrbracket \equiv 0 \pmod{p_j^{\alpha_j}}$ if and only if $\llbracket s_i,t_i\rrbracket_d \equiv 0 \pmod{p_j^{\alpha_j}}$.
\end{remark}

\edit{Remarkably, $n=\max_j|F_j|$ is also a sufficient number of qubits to achieve the non-commuting pair relation $(f_0,\dots,f_{k-1})$. We prove this below, but before that, we present an example of the construction.}

\edit{\begin{example}
Suppose $d=30=2\times3\times5$ (so $m=3$ factors), $n=3$, and $f=(2,5,6,11,15)$. Then, $F_0=\{5,11,15\}$ (elements of $f$ not divisible by 2), $F_1=\{2,5,11\}$ (elements not divisible by 3), and $F_2=\{2,6,11\}$ (elements not divisible by 5). We fill a $|f|\times m$ matrix $Q$ whose elements are from $\mathbb{Z}_{n+1}$ subject to the constraints (1) $Q_{ij}=0$ if and only if $f_i\not\in F_j$ and (2) all non-zero entries in each column are unique. Because $n=\max_j|F_j|$, constraint (2) is satisfiable. For instance,
\begin{equation}
Q=\left(\begin{array}{ccc}0&1&1\\1&2&0\\0&0&2\\2&3&3\\3&0&0\end{array}\right).
\end{equation}
Rows $0,\dots,4$ of $Q$ are used to construct $p_0,\dots,p_{4}$, $X$-type Paulis, and $q_0,\dots,q_{4}$, $Z$-type Paulis, that form non-commuting pairs achieving $f$. To construct $p_i$, a non-zero entry $Q_{ij}$ says that the prime factor labeled $j$ should be excluded from the power of $X$ on qubit $Q_{ij}$. We do the same for $q_i$, but with powers of $Z$ and additional factors of $-\beta_i\in\mathbb{Z}_d$ in the exponents. For this example,
\begin{align}
\begin{array}{ll}
p_0=X^{30/(3\times5)}\otimes X^{30}\otimes X^{30}=X^2\otimes I\otimes I,\quad& q_0=Z^{-2\beta_0}\otimes I\otimes I\\
p_1=X^{30/2}\otimes X^{30/3}\otimes X^{30}=X^{15}\otimes X^{10}\otimes I,\quad& q_1=Z^{-15\beta_1}\otimes Z^{-10\beta_1}\otimes I\\
p_2=X^{30}\otimes X^{30/5}\otimes X^{30}=I\otimes X^{6}\otimes I,\quad& q_2=I\otimes Z^{-6\beta_2}\otimes I\\
p_3=X^{30}\otimes X^{30/2}\otimes X^{30/(3\times5)}=I\otimes X^{15}\otimes X^2,\quad& q_3=I\otimes Z^{-15\beta_3}\otimes Z^{-2\beta_3}\\
p_4=X^{30}\otimes X^{30}\otimes X^{30/2}=I\otimes I\otimes X^{15},\quad& q_4=I\otimes I\otimes Z^{-15\beta_4}.
\end{array}
\end{align}
Note that $\llbracket p_i,q_j\rrbracket_d=0$ for $i\neq j$ regardless of the values of the $\beta_i$. The $\beta_i$ must be chosen so that $\llbracket p_i,q_i\rrbracket_d=f_i$. Namely, $\beta_0=8$, $\beta_1=5$, $\beta_2=1$, $\beta_3=29$, and $\beta_4=1$. 
\end{example}}

\edit{It was not a coincidence we could find values of $\beta_i$ to make the desired non-commutation relations in the previous example. Lemma~\ref{lem:helper-lem-2} in the appendix gives the necessary statement in general. This also enables the following constructive proof that $n=\max_j|F_j|$ qubits are sufficient in general to achieve a non-commuting pair relation.}





\edit{\begin{theorem}[Minimum qudits achieving a non-commuting pair relation]
\label{thm:css-construction-hard-case}
Suppose that we are given $f := (f_0,\dots,f_{k-1}) \in (\mathbb{Z}_d\setminus\{0\})^{k}$. For every $j=0,\dots,m-1$, define the multisets $F_j := \{f_i : f_i \not\equiv 0 \pmod{p_j^{\alpha_j}}, \; i \in \{0,\dots,k-1\}\}$. Then the minimum number of qudits needed for which $f$ is an achievable non-commuting pair relation is $\max_j |F_j|$. The non-commuting pairs generating $f$ can be chosen to be CSS.
\end{theorem}}

\begin{proof}
\edit{By Lemma~\ref{lem:lower-bound} $n\ge\max_j|F_j|$ qudits are necessary to achieve $f$. To show that $\max_j|F_j|$ qudits are sufficient,} consider a matrix $Q\in\mathbb{Z}_{n+1}^{k\times m}$, where $Q_{ij}\neq0$ if and only if $f_i\in F_j$. \edit{Because $n=\max_j|F_j|$, by choosing non-zero entries from $\mathbb{Z}_{n+1} \setminus \{0\}$},  we can arrange that within each column of $Q$, every non-zero entry is unique. We interpret these non-zero entries as qudit labels for qudits $1,2,\dots,n$. 

For each row $i=0,\dots,k-1$, we will now show how to construct a pair of Paulis $X(u_i)$ and $Z(v_i)$ that are supported only on qudits indicated in that row and $\llbracket X(u_i),Z(v_i)\rrbracket_d=f_i$. Then the ordered sets $\{X(u_0),X(u_1),\dots,X(u_{k-1})\}$ and $\{Z(v_0),Z(v_1),\dots,Z(v_{k-1})\}$ are CSS non-commuting pairs generating $f$. Define the set of indices in row $i$ where $Q$ takes value $h$: $S_{ih} :=\{j:Q_{ij}=h\}$ and $\overline S_{ih}:=\{j:Q_{ij}\neq h\}$. By construction, $p_j^{\alpha_j}$ divides $f_i$ for all $j\in S_{i0}$, and so $f_i=\gamma\prod_{j\in S_{i0}}p_j^{\alpha_j}$ for some integer $\gamma$. 

Set $(u_i)_{h-1}=\prod_{j\in\overline S_{ih}}p_j^{\alpha_j}$ and $v_i = -\beta u_i$, where $\beta$ is some integer (that depends on $i$) to be determined. So, we have
\begin{align}
\llbracket X(u_i),Z(v_i)\rrbracket=\beta\sum_{h\in\mathbb{Z}_{n+1} \setminus \{0\}} (\prod_{j\in\overline S_{ih}}p_j^{2\alpha_j} ) = \beta \prod_{j\in S_{i0}}p_j^{2\alpha_j} \sum_{h\in\mathbb{Z}_{n+1} \setminus \{0\}}(\prod_{j\in\overline S_{ih}\setminus S_{i0}}p_j^{2\alpha_j}).
\end{align}
Choose $\beta=\gamma\beta'$ for another integer $\beta'$. Apply Lemma~\ref{lem:helper-lem-2} with $d'=f_i$ and
\begin{equation}\label{eq:d''}
d''=(\prod_{j\in S_{i0}}p_j^{\alpha_j}) \sum_{h\in\mathbb{Z}_{n+1} \setminus \{0\}}(\prod_{j\in\overline S_{ih}\setminus S_{i0}}p_j^{2\alpha_j}),
\end{equation}
where we note that for any $j\not\in S_{i0}$ (i.e.~exactly the set of indices for which $p_j^{\alpha_j}$ does not divide $f_i$), $p_j$ does not divide $d''$. To elaborate,  $p_j$ divides all but one term in the sum in Eq.~\eqref{eq:d''}. Therefore, Lemma~\ref{lem:helper-lem-2} implies there is some integer $\beta'$ so that $\llbracket X(u_i),Z(v_i)\rrbracket=\beta'd'd''\equiv f_i\text{\space(mod $d$)}$.

Lastly, we show that $\llbracket X(u_i),Z(v_j)\rrbracket_d=0$ for $i\neq j$.
\begin{equation}
\llbracket X(u_i),Z(v_j)\rrbracket=\beta\sum_{h\in\mathbb{Z}_{n+1} \setminus \{0\}}\prod_{l\in\overline S_{ih}}p_l^{2\alpha_l}\prod_{k\in\overline S_{jh}}p_k^{2\alpha_k}.
\end{equation}
Because within each column of $Q$ the entries are unique, $\overline S_{ih}\cup\overline S_{jh}=\{0,1,\dots,m-1\}$. Thus, $d$ divides $\llbracket X(u_i),Z(v_j)\rrbracket$ and $\llbracket X(u_i),Z(v_j)\rrbracket_d=0$.
\end{proof}



\edit{Achievable non-commuting pair relations $f$ of maximum size, i.e.~$|f|=nm$, can be characterized more directly. We do so in Appendix~\ref{ssec:case-max-number-pairs}.}

\subsection{Qudits needed to achieve a matrix of commutation relations}
\label{ssec:comm-matrix-relations}

If $R$ is a principal ideal ring, and $C \in R^{k \times k}$ is a matrix satisfying $C_{ii} = 0$ and $C_{ij} = -C_{ji}$, for all $i,j=0,1,\dots,k-1$, then such a matrix is called an \textit{alternating} matrix over $R$. Suppose we are given an alternating matrix
$C\in\mathbb{Z}_d^{k\times k}$ and wish to find $P\in\mathbb{Z}_d^{k\times 2n}$ such that 
\begin{equation}
P\Lambda P^T = C.
\end{equation}
Here, rows of $P$ can be interpreted as $n$-qudit Paulis possessing the commutation relations specified by $C$. The goal of this section is to answer the following question: what is the minimum number of qudits $n$ for which such a $P$ can be found?

The following lemma, proved constructively in appendix~\ref{app:alternating-smith-normal-form}, is useful in answering this question.


\begin{lemma}[Alternating Smith Normal Form]\label{lem:zdasnf}
Given an alternating matrix $A\in\mathbb{Z}_d^{k\times k}$, there are matrices $L,B\in\mathbb{Z}_d^{k\times k}$, where $B$ is alternating with at most one non-zero entry per row and column and $L$ is invertible, such that $A=LBL^T$. We may further arrange $B$ so that it is non-zero only in the top-left $2r \times 2r$ block which has the form $\bigoplus_{i=1}^r\left(\begin{smallmatrix}0&\beta_i\\-\beta_i&0\end{smallmatrix}\right)$ for integer $r=\Theta(M_A)/2$, where $M_A$ is the $\mathbb{Z}_d$-submodule generated by the columns of $A$, and each $\beta_i\in\mathbb{Z}_d$ non-zero, satisfying $\beta_i\mid\beta_{i+1}$ for all $i < r$. Also, for all $i = 1,\dots, r$, $\beta_i$ is uniquely determined up to multiplication by a unit by the formula $\beta_i =d_{2i} / d_{2i-1} \bmod d$ (or, alternatively, the formula $\beta_i =d_{2i-1}/ d_{2i-2} \bmod{d}$), where $d_j$ is the greatest common divisor of all $j\times j$ minors of $A$ (and $d_0:=1$).
\end{lemma}

\begin{remark}
Note that in the formulas for $\beta_i$ in the above lemma, both the minors and the greatest common divisor are first evaluated over integers, and then the division is also performed over integers. Finally the modulo $d$ operation gives back an element of $\mathbb{Z}_d$. Similar to the remark following Lemma~\ref{lem:asnf}, one should note that the smallest integer $j$ for which $d_j / d_{j-1} \equiv 0 \pmod{d}$ is $2r+1$, and thus an odd integer (the fact that this integer is odd for any $d$ is itself quite an interesting fact). In fact, if $k$ is odd, then such an odd integer $j$ must exist as $d_k = 0$. We can also easily deduce that $d_{j'} \neq 0$ for all $j' < j$. Furthermore, it follows from Lemma~\ref{lem:asnf} that for all $j' \geq j$ such that $d_{j'-1} \neq 0$, we also have $d_{j'} / d_{j'-1} \equiv 0 \pmod{d}$.
\end{remark}

Apply the Lemma to $C$, finding invertible $L$ and alternating $B$ such that $C=LBL^T$. Then defining $Q := L^{-1}P \in\mathbb{Z}_d^{k\times 2n}$ we have
\begin{equation}
(L^{-1}P)\Lambda(L^{-1}P)^T=Q\Lambda Q^T=B,
\end{equation}
and let $r = \Theta(M_C)/2$ with $M_C$ denoting the submodule generated by the columns of $C$. Now, since $B$ is alternating with at most one non-zero entry per row and column, the Paulis represented by the first $2r$ rows of $Q$ are simply non-commuting pairs (the last $k-2r$ rows can be chosen to be all 0s, representing identity Paulis). In Theorem~\ref{thm:css-construction-hard-case}, we concluded that the necessary and sufficient number of qudits needed to achieve a set of non-commuting relations $\{\beta_1,\beta_2,\dots,\beta_r\}$ is $\max_j|F_j|$, where $F_j=\{\beta_i:\beta_i\not\equiv 0\pmod{p_j^{\alpha_j}},i\in\{1,2,\dots,r\}\}$. Now there exists some $p_j$ so that $\beta_r\not\equiv0\pmod{p_j^{\alpha_j}}$, because otherwise we would have $\beta_r = 0$. Since $\beta_i|\beta_{i+1}$ for all $i<r$, if $\beta_r\not\equiv0\pmod{p_j^{\alpha_j}}$, then $\beta_i\not\equiv0\pmod{p_j^{\alpha_j}}$ for all $i$. Thus, $\max_j|F_j|=r$ qudits are necessary and sufficient to construct $Q$ and also $P=LQ$. 

The above establishes the following theorem:
\begin{theorem}[\edit{Minimum qudits achieving commutation relations}]
\label{thm:min-qubits-comm-matrix}
If $C\in\mathbb{Z}_d^{k\times k}$ is an alternating matrix, and $Q\in\mathbb{Z}_d^{k\times 2n}$ satisfies $Q\Lambda Q^T=C$, then $n\ge\Theta(M_C)/2$, where $M_C$ is the submodule generated by the columns of $C$. Moreover, there exists a matrix $P\in\mathbb{Z}_d^{k\times\Theta(M_C)}$, such that $C=P\Lambda P^T$. Rows of $P$ indicate $k$ $(\Theta(M_C)/2)$-qudit Paulis possessing the commutations relations specified by $C$.
\end{theorem}

\edit{For some particular forms of alternating matrix $C\in\mathbb{Z}_d^{k\times k}$, one can compute $\Theta(M_C)$ analytically.} We mentioned one such case already before in Lemma~\ref{lem:min-number-gens-special-matrix}. Another case happens when all entries in the upper triangular part of $C$, except the diagonal, are equal. That is, suppose $C$ is alternating with $C_{ij}=t \in \mathbb{Z}_d \setminus \{0\}$, for all $j > i$. Then applying Lemma~\ref{lem:zdasnf} again gives matrices $L,B \in \mathbb{Z}_d^{k\times k}$, with $L$ invertible and $B$ alternating and of the form given by the lemma. Let us calculate the quantities $d_0, d_1, \dots, d_k$ as defined in Lemma~\ref{lem:zdasnf}. First suppose $t=1$. Then we have:
\begin{enumerate}[(i)]
    \item $d_0 = 1$ by definition.
    \item For every even $j$, the determinant (evaluated over integers) of the top-left $j \times j$ block of the matrix $C$ is $1$, which is easily verified by bringing $C$ to upper-triangular form using row (or column) operations. In particular, if $k$ is even, then $d_k=\det(C)=1$.
    \item If $k$ is odd, $d_k=\det(C)=0$ as $C$ is an alternating matrix. Also by (ii) we get $d_{k-1}=1$, as at least one $(k-1) \times (k-1)$ minor is $1$.
\end{enumerate}
Then by the chain of divisibilities condition mentioned in the remark following Lemma~\ref{lem:asnf}, we conclude that for $t=1$, we have $d_j = 1$ for all $j$ if $k$ is even, while if $k$ is odd, then we have $d_k = 0$ and $d_j = 1$ for all $j < k$. Now for arbitrary $t \in \mathbb{Z}_d \setminus \{0\}$, we simply note that $d_j$ equals $t^j$ times the value of $d_j$ for the case $t=1$. Combining these facts, and using the formulas in Lemma~\ref{lem:zdasnf}, we obtain the following:
\begin{enumerate}[(i)]
    \item If $k$ is even, then $\beta_i = t$ for all $1 \leq i \leq k/2$. Thus $\Theta(M_C) = k$.
    \item If $k$ is odd, then $\beta_i = t$ for all $1 \leq i \leq \lfloor k/2 \rfloor$. Thus $\Theta(M_C) = k-1$.
\end{enumerate}

From these observations, we immediately obtain the following corollary as a direct consequence of Theorem~\ref{thm:min-qubits-comm-matrix}:
\begin{corollary}
\label{cor:non-comm-all-the-same}
    Let $t \in \mathbb{Z}_d \setminus \{0\}$. Then the largest size of a non-commuting set $S$ on $n$-qudits, such that $\llbracket p, q \rrbracket_d = \pm t$, for every distinct $p,q \in S$, is $2n+1$.
\end{corollary}

\edit{
A construction of such a non-commuting set of qudit Paulis can be found in \cite{gungordu2014parafermion}. To describe it, we let $X_i$ (resp.~$Z_i$) denote the single-qudit Pauli $X$ (resp.~$Z$) on the qudit $i$. For $j=1,2,\dots,n$ let
\begin{align}
\gamma_{2j-1}=Z_j\prod_{i=1}^{j-1}X_i,\quad \gamma_{2j}=\omega^{(d+1)/2}X_jZ_j\prod_{i=1}^{j-1}X_i.
\end{align}
Also, let $\gamma_0=\prod_{j=1}^{n}\gamma_{2j-1}\gamma_{2j}^\dag$. Then $\llbracket\gamma_i,\gamma_j\rrbracket_d=\omega$ for all $i<j$, and $S=\{\gamma_0,\gamma_1,\dots,\gamma_{2n}\}$ forms the claimed non-commuting set of size $2n+1$.
}
\section{Some group theoretic results}
\label{sec:group_theory}

In this section, we depart from the previous sections where we studied elements of $\mathcal{P}_n$, and instead we take up the study of the Heisenberg-Weyl Pauli group $\overline{\mathcal{P}}_n$ for a $d$-dimensional qudit, without ignoring the phases. \editt{We begin this section by establishing 
the notion of equivalent generating sets in Theorem~\ref{thm:equiv-gen}. Using this, we give in Section~\ref{ssec:minimal-gen-sets} a characterization of a near minimal generating set of any subgroup of $\overline{\mathcal{P}}_n$ and also an algorithm to compute such a generating set. We also provide an (inefficient) algorithm to compute a minimal generating set of any subgroup of $\overline{\mathcal{P}}_n$ from a near minimal generating set. Next, in Section~\ref{ssec:canonical-form-gen-set}, we provide a way to compute a Gram-Schmidt generating set of any subgroup of $\overline{\mathcal{P}}_n$. This generalizes the well-known stabilizer-destabilizer decomposition \cite{aaronson2004improved} of the qubit Pauli group to the qudit case. Finally, in Section~\ref{ssec:subgroups-noncomm-pairs} we develop some results to compute the size of any subgroup of $\overline{\mathcal{P}}_n$. This generalizes similar results in \cite{gheorghiu2014standard}, where only the special case of stabilizer subgroups was studied. Of particular note in this subsection is the square-free theorem (Theorem~\ref{thm:square-free-max-pairs}), which says that maximal sets of non-commuting pairs generate the qudit Pauli group when $d$ is square-free.}

Let us first introduce some notation that will help the discussion. All products in this section will be ordered, unless mentioned otherwise. Recall from Section~\ref{ssec:qudit-pauli-group} that any element $p \in \hwgrp$ has the form $\omega^j P(v)$ for some $v \in \mathbb{Z}_d^{2n}$ and $j \in \mathbb{Z}_d$, where $\omega = e^{2 \pi i/d}$, and for a given $p$, the corresponding values of $j$ and $v$ are uniquely determined. Thus we can equivalently represent $p$ by the tuple $(j, v)$, and the map $p \mapsto (j,v)$ sets up a bijection $\hwgrp \rightarrow \mathbb{Z}_d \times \mathbb{Z}_d^{2n}$. For any $p \in \hwgrp$, we define $p^0 := I$. Representing a Pauli as an element of $\mathbb{Z}_d \times \mathbb{Z}_d^{2n}$, we define the projection maps $\pi_1$ and $\pi_2$ onto the first and second factors respectively, i.e. $\pi_1((j,v)) = j$ and $\pi_2((j,v))=v$, for a Pauli represented by the tuple $(j,v)$. Similarly, for an ordered multiset $S:=\{q_0,q_1,\dots,q_{k-1}\}$, we use the notation $\pi_1(S):=\{\pi_1(q_0),\pi_1(q_1),\dots,\pi_1(q_{k-1})\}$, and $\pi_2(S):=\{\pi_2(q_0),\pi_2(q_1),\dots,\pi_2(q_{k-1})\}$. We may unambiguously associate $\pi_2(S)$ with a $2n \times k$ matrix with elements in $\mathbb{Z}_d$, where the $j^{\text{th}}$ column is $\pi_2(q_j)$, and we will also denote the matrix by $\pi_2(S)$ when there is no chance for confusion.

Now suppose we have two Paulis $p,q \in \hwgrp$ represented by $(j,v)$ and $(k,w)$ respectively. Then it is an easy exercise to check that $pq$ is represented by $(\ell, v + w)$, where $\ell = j + k + v^T \left( \begin{smallmatrix}
    0 & 0 \\ I & 0
\end{smallmatrix}\right)  w$, evaluated modulo $d$, where $I$ here is an $n \times n$ identity matrix. From this, one can also easily show that for $t \geq 1$, $p^t$ is represented by $(\ell,tv)$, where $\ell = tj + \frac{t(t-1)}{2} 
 v^T \left( \begin{smallmatrix}
    0 & 0 \\ I & 0
\end{smallmatrix}\right)  v$, evaluated modulo $d$, and thus computing both $p+q$ and $p^t$ takes $O(n)$ operations (for $t$ bounded). The following well-known lemma is easy to establish (see also \cite{planat2011pauli,planat2007qudits}):

\begin{lemma}
\label{lem:max-degree}
Let $p=\omega^jP(v) \in \hwgrp$. If $d$ is odd, then $p^d = I$. If $d$ is even, then $p^d = \pm I$, with $p^d = I$ if and only if $v^T \left( \begin{smallmatrix}
    0 & 0 \\ I & 0
\end{smallmatrix}\right)  v$ is even.
\end{lemma}


Thus in the case of odd $d$, the maximum possible order of any Pauli is $d$, while in the case of even $d$, the maximum possible order of a Pauli is $2d$ (for example, on $1$ qudit, the Pauli $XZ$ has order $2d$ when $d$ is even). 
 
Next we would like to figure out a way to transform a generating set of a subgroup of $\hwgrp$ to another generating set of the same group with the same number of generators. Before we present the main theorem on this, let us consider the case of one generator. Suppose $p \in \hwgrp$ generates a group $G$. For some unit $s \in \mathbb{Z}_d$, let us define $q := p^s$. We can then show that $q$ also generates $G$. To see this let $t \in \mathbb{Z}_d$ be the unit such that $st=1$ in $\mathbb{Z}_d$, i.e. treating $s,t$ as elements of $\mathbb{Z}$ we have $st = kd + 1$ for some non-negative integer $k$. Then $q^t = (p^s)^t = p^{st} = p^{kd+1} = p^{kd}p$. Now if $d$ is odd, then $p^{kd}=I$, and if $d$ is even, then $p^{kd} = \pm I$ by Lemma~\ref{lem:max-degree}, and thus $q^t = \pm p$. If $q^t =p$, then $q$ generates $G$. The case $q^t = -p$ is more interesting. This is precisely the case when $p^d = -I$ and $d$ is even, and thus $q^{t(d+1)}=q^{td} q^t = p$, which again shows that $q$ generates $G$. Thus we have proved that $\langle p \rangle = \langle q \rangle = G$. Theorem~\ref{thm:equiv-gen} generalizes this observation to multiple generators. We will also need Lemma~\ref{lem:helper-lem-phase-mult}, whose proof is easy and is left to the reader.

\textbf{Notation:} Recall that $[p,q]$ is the group commutator of $p,q \in \hwgrp$. Let $S \subseteq \hwgrp$ be a set or multiset. We will denote $I_S := \{p \in \langle S \rangle : \pi_2(p) = 0\}$ to be the set of all elements of $\langle S \rangle$ that are proportional to $I$. We will also denote $J_S := \{[p,p'] : p,p' \in S\} \cup \{p^d :p \in S\}$. Then $I_S$ is a subgroup and moreover we have $\langle J_S \rangle \subseteq I_S \subseteq \langle S \rangle$.

\begin{lemma}
\label{lem:helper-lem-phase-mult}
Suppose $S := \{q_0,q_1,\dots,q_{k-1}\} \subseteq \hwgrp$ is a multiset. Then
\begin{enumerate}[(i)]
    \item $\langle J_S \rangle \subseteq \langle S \rangle$.
    \item Let $T \subseteq \langle S \rangle$. Then $\langle J_T \rangle \subseteq \langle J_S \rangle$. Moreover, if $\langle T \rangle = \langle S \rangle$, then $\langle J_T \rangle = \langle J_S \rangle$.
    \item If $S' \subseteq S$ is chosen such that $q \in S \setminus S'$ implies $\pi_2(q)=0$, then $\langle J_S \rangle = \langle J_{S'} \rangle$.
    \item Let $T:= \{q'_0,q'_1,\dots,q'_{k-1}\} \subseteq \hwgrp$ be a different multiset such that $q'_j = q_j \omega^{\beta_j}$, for some $\beta_j \in \mathbb{Z}_d$ for each $j$. Then $\langle J_T \rangle = \langle J_S \rangle$.
\end{enumerate}
\end{lemma}

\begin{theorem}[\edit{Equivalent generating sets}]
\label{thm:equiv-gen}
Suppose $S := \{q_0,q_1,\dots q_{k-1}\}$ is an ordered multiset of elements of $\hwgrp$. Let $A \in \mathbb{Z}_d^{k \times k}$ be an invertible matrix, and consider the ordered multiset $T := \{q'_0,q'_1,\dots,q'_{k-1}\} \subseteq \hwgrp$, where we define $q'_i := \prod_{j=0}^{k-1} q_j^{A_{ij}}$. Then $\langle S \rangle = \langle T \rangle$.
\end{theorem}

\begin{proof}
This follows a similar line as the one generator example above. Keeping track of the phases that is facilitated by Lemma~\ref{lem:helper-lem-phase-mult}. A full proof is given in Appendix~\ref{app:proofs_sec6}.
\end{proof}


\edit{Our next goal is to provide a complete characterization of the subgroup of phases $I_S$, given a generating set $S \subseteq \hwgrp$. For instance, for any $p,q\in S$, both $[p,q]$ and $p^d$ are proportional to identity and thus members of $I_S$. There may be additional products of generators that are proportional to identity as well, and these can be read off the kernel of $\pi_2(S)$. That these are all the elements of $I_S$ is the content of the next lemma.

Recall that given a matrix $A \in \mathbb{Z}_d^{k \times \ell}$, the kernel of $A$ is a submodule of $\mathbb{Z}_d^{\ell}$ defined as $\ker(A):= \{v \in \mathbb{Z}_d^{\ell}: Av=0\}$.
}

\begin{lemma}
\label{lem:KS-generators}
Suppose $S := \{q_0,q_1,\dots,q_{k-1}\} \subseteq \hwgrp$ is an ordered multiset. Consider the matrix $\pi_2(S) \in \mathbb{Z}_d^{2n \times k}$, and let $\overline{K} \in \mathbb{Z}_d^{k \times \ell}$ be such that the columns of $\overline{K}$ is a generating set for $\ker(\pi_2(S))$. Then we have the following:
\begin{enumerate}[(i)]
    \item Let $v \in \ker(\pi_2(S))$, and define $q' := \prod_{j=0}^{k-1}q_j^{v_j}$. Then $\pi_2(q')=0$.
    \item Define the set $N_S := \left \{ \prod_{j=0}^{k-1}q_j^{v_j} : v \in \ker(\pi_2(S)) \right\}$. Then $I_S = \langle N_S, J_S \rangle$.
    \item Define the multiset $\overline{K}_S := \left \{\prod_{j=0}^{k-1} q_j^{\overline{K}_{ji}} : i=0,1,\dots,\ell-1 \right\}$. Then $I_S = \langle \overline{K}_S, J_S \rangle$.
\end{enumerate}
\end{lemma}

\begin{proof}
\edit{See Appendix~\ref{app:proofs_sec6}.}
\end{proof}

\edit{The group $K=\langle\omega I\rangle$ is isomorphic to the additive group on $\mathbb{Z}_d$. Subgroups of $\mathbb{Z}_d$ have the property that they can always be generated by one element (in other words, $\mathbb{Z}_d$ is a principal ideal ring). Thus, the following lemma says the same is true of subgroups of $K$.}

\begin{lemma}
\label{lem:identity-gen}
 Let $S := \{\omega^{\mu_j} I: j=0,1,\dots,t, \; \mu_j \in \mathbb{Z}_d\} \subset \hwgrp$. Let $\mathcal{I}$ be the ideal generated by $\{\mu_0,\dots,\mu_t\}$. Then
 \begin{enumerate}[(i)]
     \item $\langle S \rangle = \{\omega^\mu I : \mu \in \mathcal{I}\}$.
     \item There exists $\mu \in \mathcal{I}$ such that $\mathcal{I}$ is generated by $\mu$, and $\langle S \rangle = \langle \omega^{\mu}I \rangle$. One can choose $\mu = \gcd(\mu_0,\mu_1,\dots,\mu_t,d)$, where $\gcd$ is the greatest common divisor evaluated over integers.
 \end{enumerate}
\end{lemma}



\begin{algorithm}
\caption{Return a generator for $I_S$ given $S \subseteq \hwgrp$}
\begin{algorithmic}[1]
    \Procedure{Identity\_Generator}{$S:=\{q_0,q_1,\dots,q_{k-1}\}, d$}
        \State $\mu \leftarrow d$
        \If{$d$ is even} \Comment{Generate phases of $I$ from $d^{\text{th}}$ powers of generators}
            \For{$0 \leq j \leq k-1$}
                \State $t \leftarrow \pi_1(q_j^d)$
                \If{$t=d/2$}
                    \State $\mu \leftarrow d/2$; \; Exit for loop
                \EndIf
            \EndFor
        \EndIf
        \If{$\mu = 1$}
            \State \Return $\omega I$
        \EndIf
        \State
        \For{$0 \leq j < \ell \leq k-1$} \Comment{Generate phases of $I$ from group commutators}
            \State $t \leftarrow \pi_1([q_j,q_{\ell}])$
            \State $\mu \leftarrow \gcd(t,\mu)$
            \If{$\mu = 1$}
                \State \Return $\omega I$
            \EndIf
        \EndFor
        \State
        \State Compute $\overline{K} \in \mathbb{Z}_d^{k \times s}$ such that $\pi_2(S)\overline{K}=0$, columns of $\overline{K}$ span $\ker(\pi_2(S))$
        \State
        \For{$0 \leq j \leq s-1$} \Comment{Generate all other phases of $I$}
            \State $q \leftarrow \prod_{\ell=0}^{k-1} q_{\ell}^{\overline{K}_{\ell j}}$
            \State $t \leftarrow \pi_1(q)$
            \State $\mu \leftarrow \gcd(t,\mu)$
            \If{$\mu = 1$}
                \State \Return $\omega I$
            \EndIf
        \EndFor
        \State
        \State \Return $\omega^{\mu}I$ 
    \EndProcedure
\end{algorithmic}
\label{alg:generator-IS}
\end{algorithm}

Lemma~\ref{lem:KS-generators} and Lemma~\ref{lem:identity-gen} leads to a simple algorithm (Algorithm~\ref{alg:generator-IS}) to find a generator for $I_S$ given a multiset $S \subseteq \hwgrp$. The basic idea is to get a generating set for $I_S$ and then use Lemma~\ref{lem:identity-gen}(ii). In the algorithm, $\gcd(a,b)$ refers to the evaluation of the greatest common divisor over integers, for integers $0 \leq a,b \leq d$, which takes $O(M(d)\log(\log (d)))$ operations, where $M(d)$ is the number of operations needed to multiply two integers no greater than $d$ (so $M(d)=O(\log^2 d)$ for elementary-school arithmetic, for instance) \cite{storjohann2000algorithms}. In order to facilitate early termination, instead of evaluating all the generators and then evaluating the $\gcd$ using Lemma~\ref{lem:identity-gen}(ii) of the whole list, we use the property $\gcd(a,b,c)= \gcd(a,\gcd(b,c))$ recursively, to update the $\gcd$ everytime we have a new generator. If at any stage the $\gcd$ becomes $1$, we can terminate the algorithm (Lines~8-9, Lines~14-15, Lines~23-24). In Lines~3-7, we use the fact from Lemma~\ref{lem:max-degree} that for odd $d$, any $p \in \hwgrp$ satisfies $p^d=I$, and for even $d$ satisfies $p^d = \pm I$. Lines~6-7 exploits this fact, and terminates checking the $d^{\text{th}}$ powers of the remaining generators, if a generator $q_j \in S$ is detected such that $q_j^d = -I$. In Lines~11-15, notice that we exploit the fact that $\pi_1([q_j,q_{\ell}]) + \pi_1([q_{\ell},q_j]) \equiv 0 \pmod{d}$, or equivalently $[q_j,q_{\ell}] = [q_{\ell},q_j]^{-1}$; thus only the commutators $[q_j,q_{\ell}]$ need to be considered for $j < \ell$. In Line~17, the computation of the kernel matrix $\overline{K}$ can be carried out using the techniques in \cite[Chapter~5]{storjohann2000algorithms}. In terms of computational cost, the computation of $q_j^d$ and $[q_j,q_{\ell}]$ in Line~5 and Line~12 respectively involves $O(n)$ operations, while computation of $q$ is Line~20 involves $O(kn)$ operations.

\subsection{Minimal and near-minimal generating sets}
\label{ssec:minimal-gen-sets}

We are now in a position to answer the following question: given a multiset $S \subseteq \hwgrp$, can we find a non-empty generating set of $\langle S \rangle$ of the smallest size? Such a set is called a \textit{minimal generating set} of $\langle S \rangle$. The ``non-empty'' condition is only relevant for the case $\langle S \rangle = \{I\}$. This subsection is dedicated to providing a \textit{nearly complete} solution to this problem. We begin by stating a result from commutative algebra that will be needed below, whose proof can be found in \cite{span-equality-relation} (we thank Jeremy Rickard for the outline of the proof):

\begin{lemma}
\label{lem:span-same}
Let $A, B \in \mathbb{Z}_d^{k \times \ell}$. Then the following conditions are equivalent.
\begin{enumerate}[(i)]
    \item The submodules of $\mathbb{Z}_d^k$ generated by the columns of $A$ and $B$ are the same.
    \item There exists an invertible matrix $C \in \mathbb{Z}_d^{\ell \times \ell}$ such that $A = BC$.
\end{enumerate}
\end{lemma}

Also note the following simple corollary of this lemma, which we prove for completeness in Appendix~\ref{app:proofs_sec6}.

\begin{corollary}
\label{cor:equal-span-same-inv-factor}
Let $A \in \mathbb{Z}_d^{k \times p}$ and $B \in \mathbb{Z}_d^{k \times q}$. If the columns of $A$ and $B$ generate the same submodule of $\mathbb{Z}_d^k$, then the number of invariant factors of $A$ and $B$ are equal.
\end{corollary}

Now for the rest of this subsection, suppose $S := \{s_0,s_1,\dots,s_{k-1}\} \subseteq \hwgrp$ is an ordered mutltiset. The case when $\pi_2(S)$ has zero invariant factors is special. In this case, we have $\langle S \rangle = I_S$, and then a minimal generating set of $I_S$ can be obtained using Lemma~\ref{lem:identity-gen}(ii). Thus, through the remainder of this section we assume that $\pi_2(S)$ has at least one invariant factor. The following lemma is easy to establish, which gives us a lower and upper bound on the size of a minimal generating set of $\langle S \rangle$.

\edit{
\begin{lemma}[The size of a smallest generating set]\label{lem:min-gen-set-first-bound-simp}
For any ordered multiset $S\subseteq\hwgrp$, if $\pi_2(S) \in \mathbb{Z}_d^{2n \times k}$ has $r$ invariant factors in its Smith normal form, then the smallest generating set of $\langle S\rangle$ has either $r$ or $r+1$ elements.
\end{lemma}
\begin{proof}
This follows from Lemma~\ref{lem:min-gen-set-first-bound} in Appendix~\ref{app:proofs_sec6}. The crux is similar to the canonical generating set in \cite{gheorghiu2014standard}, but finds a canonical generating set for an arbitrary group rather than just a qudit stabilizer group, and does not change the Pauli basis using Clifford operators.
\end{proof}
}

Based on this result, we make the following definition.
\begin{definition}
\label{def:near-min-gen}
Given an ordered multiset $S := \{s_0,s_1,\dots,s_{k-1}\} \subseteq \hwgrp$, such that $r$ is the number of invariant factors of $\pi_2(S)$, we call a subset $T \subseteq \langle S \rangle$ a \textit{near-minimal generating set} of $\langle S \rangle$, if $T$ is of the form $T := T' \cup \{p\}$ and  satisfies (i) $\langle T \rangle = \langle S \rangle$, (ii) $\langle p \rangle = I_S$, and (iii) $|T'|=r$. Note that if $r=0$, then $T'$ is empty; so this definition also works for that case.
\end{definition}

Lemma~\ref{lem:min-gen-set-first-bound} gives us a way to compute a near-minimal generating set of $\langle S \rangle$. At this point, it begs the question of whether a near-minimal generating set of $\langle S \rangle$ is also a minimal generating set of $\langle S \rangle$ or not. While we do not completely resolve this question here, we prove some partial results below in the remainder of this subsection. Let us first show that indeed there are cases where a near-minimal generating set is a minimal generating set. An easy example is the case when the number of invariant factors of $\pi_2(S)$ is zero. In this case, a near-minimal generating set of $\langle S \rangle$ has size one, and hence it is a minimal generating set. We give another example below.
\begin{example}
Suppose $S = \{X, \omega I\} \subseteq \hwgrp$, for any $d > 2$ and number of qudits $n$. Then one checks easily that $\langle S \rangle = \{\omega^j X^a : j,a \in \mathbb{Z}_d\}$. Thus $|\langle S \rangle| = d^2$. Computing the Smith normal form of $\pi_2(S)$ shows that the number of its invariant factors is one. Thus a near-minimal generating set of $\langle S \rangle$ has size two. Suppose for contradiction that there exists a minimal generating set $T$ of $\langle S \rangle$ with $|T|=1$. Since the maximum possible order of any element of $\hwgrp$ is $2d$ (Lemma~\ref{lem:max-degree}), this implies that $|\langle T \rangle| \leq 2d < d^2$. So $T$ cannot generate $\langle S \rangle$. Thus all near-minimal generating sets of $\langle S \rangle$ are also minimal generating sets of $\langle S \rangle$ in this example.
\end{example}

Let us now give an example where a near-minimal generating set is not a minimal generating set, i.e. there exists generating sets of $\langle S \rangle$ of exactly size $r \geq 1$, where $r$ is the number of invariant factors of $\pi_2(S)$. \edit{One such example is when $r=k$ in Lemma~\ref{lem:min-gen-set-first-bound-simp} (then clearly $S$ is itself minimal).} Another one is given below in Example~\ref{ex:stabilizer}.

\begin{definition}[{\cite{gottesman1997stabilizer}}]
\label{def:stabilizer}
A stabilizer group is a subgroup $G \subseteq \hwgrp$ such that  $\{\omega^jI:j \in \mathbb{Z}_d\}\cap G=\{I\}$, \edit{or, equivalently, $I_G=\{I\}$.}
\end{definition}

\begin{example}
\label{ex:stabilizer}
Assume that $\langle S \rangle$ is a stabilizer group and $\langle S \rangle \neq \{I\}$. In this case, if $S$ is any generating set of $\langle S \rangle$, we know the size of a minimal generating set of $\langle S \rangle$ is exactly the number of invariant factors $r$ of $\pi_2(S)$. This is because a near-minimal generating set $T' \cup \{p\}$ of $G$ must satisfy $p = I$, and thus $\langle T' \cup \{p\} \rangle = \langle T' \rangle$. 
\end{example}

How far are we from computing a minimal generating set of $\langle S \rangle$, given we have a near-minimal generating set of $\langle S \rangle$? The next result gives a nice structure theorem to find minimal generating sets from near-minimal ones. 
\begin{theorem}
\label{thm:min-gen-set-simple-form}
Given $S \subseteq \hwgrp$, suppose that $T := T' \cup \{p\}$ is a near-minimal generating set of $\langle S \rangle$ with $\langle p \rangle = I_S$. Let the number of invariant factors of $\pi_2(S)$ be $r \geq 1$, and suppose $T' = \{q_0,q_1,\dots,q_{r-1}\}$. Then the following conditions are equivalent.
\begin{enumerate}[(i)]
    \item $T$ is not a minimal generating set of $\langle S \rangle$.
    \item There exist integers $\gamma_0, \gamma_1,\dots, \gamma_{r-1} \in \mathbb{Z}_d$, such that $T'' := \{p^{\gamma_0} q_0, p^{\gamma_1} q_1, \dots, p^{\gamma_{r-1}} q_{r-1}\}$ is a minimal generating set of $\langle S \rangle$.
    \item There exist integers $\gamma_0, \gamma_1,\dots, \gamma_{r-1} \in \mathbb{Z}_d$, such that $p \in \langle T'' \rangle$, with $T''$ defined in part (ii).
\end{enumerate}
\end{theorem}

\begin{proof}
Denote by $M$ the submodule of $\mathbb{Z}_d^{2n}$ generated by the columns of $\pi_2(S)$. We first prove that (i) implies (ii). Suppose that $U := \{u_0,u_1,\dots,u_{r-1}\}$ is a minimal generating set of $\langle S \rangle$. By Lemma~\ref{lem:min-gen-set-first-bound}(ii), (v), we know that the columns of both $\pi_2(T')$ and $\pi_2(U)$ generate $M$. Thus by Lemma~\ref{lem:span-same} we can conclude that there exists an invertible matrix $C \in \mathbb{Z}_d^{r \times r}$ such that $\pi_2(T') = \pi_2(U) C$. Now define a new set $T'' := \{t_0,t_1,\dots,t_{r-1}\}$ such that $t_i = \prod_{j=0}^{r-1} u_j^{C_{ji}}$, for every $i$. Lemma~\ref{thm:equiv-gen} then implies that $\langle T'' \rangle = \langle U \rangle = \langle S \rangle$. It also follows from the equality $\pi_2(T') = \pi_2(U) C$ that $\pi_2(t_i) = \pi_2(q_i)$ for every $i$. Thus each $t_i$ is equivalent to $q_i$ up to some phase factor, and since $\langle p \rangle = I_S$, we immediately conclude that $t_i = p^{\gamma_i} q_i$ for some $\gamma_i \in \mathbb{Z}_d$, for every $i$.

Now assume (ii) is true. As $\langle T'' \rangle = \langle S \rangle$, this implies that $p \in \langle T'' \rangle$, proving (iii).

Finally assume that (iii) is true, and we want to prove (i). Clearly  the definition of $T''$ in (iii) implies that $\langle T'' \rangle \subseteq \langle S \rangle$, because $p, q_i \in T$, and hence $p^{\gamma_i} q_i \in \langle S \rangle$ for every $i$. To prove the reverse containment, note that $p \in \langle T'' \rangle$ implies that $q_i \in \langle T'' \rangle$ for every $i$ (since $p_i^{\gamma_i}q_i \in T''$ by definition). This shows that $T \subseteq \langle T'' \rangle$ and thus  $\langle S \rangle \subseteq \langle T'' \rangle$. Thus $T''$ is a minimal generating set of $\langle S \rangle$ (since it has size $r$), and this proves (i).
\end{proof}
\begin{algorithm}
\caption{Given a near-minimal generating set $T = T' \cup \{p\}$, find a minimal generating set}
\begin{algorithmic}[1]
    \Procedure{Find\_Minimal\_Generating\_Set}{$T':=\{q_0,q_1,\dots,q_{r-1}\}, p, d$}
        \For{$(\gamma_0,\gamma_1,\dots, \gamma_{r-1}) \in \mathbb{Z}_d^r$}
            \State $T'' \leftarrow \{p^{\gamma_0} q_0, p^{\gamma_1} q_1, \dots, p^{\gamma_{r-1}} q_{r-1}\}$
            \State $\tilde{p} \leftarrow \textsc{Identity\_Generator} (T'',d)$
            \If{$p \in \langle \tilde{p} \rangle$}
            \Comment{Check membership of $p$ in $\langle \tilde{p} \rangle$}
                \State \Return $T''$
            \EndIf
        \EndFor
        \State \Return $T' \cup \{p\}$
    \EndProcedure
\end{algorithmic}
\label{alg:check-min-gen-set}
\end{algorithm}
Condition (iii) of Theorem~\ref{thm:min-gen-set-simple-form} can be used to obtain a simple (but inefficient) algorithm to test whether a near-minimal generating set of $\langle S \rangle$ is also a minimal generating set or not, and then output a minimal generating set of $\langle S \rangle$. This is given in Algorithm~\ref{alg:check-min-gen-set}. We assume that $\pi_2(S)$ has $r \geq 1$ invariant factors, so that a computed near-minimal generating set (using Lemma~\ref{lem:min-gen-set-first-bound}) of $\langle S \rangle$ has size $r+1$. Suppose here that $r < |S|$, so that we are in the setting of Lemma~\ref{lem:min-gen-set-first-bound}(v). Let $T = T' \cup \{p\}$ be one such near-minimal generating set with $\langle p \rangle = I_S$, and let $T' = \{q_0,q_1,\dots,q_{r-1}\}$. The algorithm proceeds by looping over each $r$-tuple $(\gamma_0,\gamma_1,\dots, \gamma_{r-1}) \in \mathbb{Z}_d^r$, and then constructing the set $T'' \subseteq \langle S \rangle$ of size $r$. If $T''$ is a minimal generating set of $\langle S \rangle$, then by Theorem~\ref{thm:min-gen-set-simple-form} we must have $p \in I_{T''}$. To check this condition, in Lines~3-5 we use Algorithm~\ref{alg:generator-IS} to compute $\tilde{p} := \omega^\beta I$ with the property that $\beta \in \mathbb{Z}_d$ is the smallest possible integer satisfying $\langle \tilde{p} \rangle = I_{T''}$. Then if $p = \omega^\delta I$ for some $\delta \in \mathbb{Z}_d$, it follows that $p \in I_{T''}$ if and only if $\beta$ divides $\delta$. If this condition check succeeds for any $(\gamma_0,\gamma_1,\dots, \gamma_{r-1}) \in \mathbb{Z}_d^r$, then we have found a minimal generating set $T''$ of $\langle S \rangle$ of size $r$, and otherwise we conclude that $T = T' \cup \{p\}$ is a minimal generating set of $\langle S \rangle$ of size $r+1$.

Another thing to note about Algorithm~\ref{alg:check-min-gen-set} is that the computational work in Line~4 can be significantly reduced by precomputing and storing a few quantities. Note that due to Lemma~\ref{lem:helper-lem-phase-mult}(iv), the results of Lines~2-17 (of Algorithm~\ref{alg:generator-IS}) in the execution of {\small{IDENTITY\_GENERATOR}} do not change irrespective of $T''$ (or equivalently irrespective of the choice of the $r$-tuple $(\gamma_0,\gamma_1,\dots, \gamma_{r-1}) \in \mathbb{Z}_d^r$ in Line~2 of Algorithm~\ref{alg:check-min-gen-set}) --- thus one can compute and store $\overline{K}$ and the $\mu$ value (let us call this $\mu_0$) up to this point (Line~17), with the choice $\gamma_0 = \gamma_1 = \dots = \gamma_{r-1}=0$ (i.e. $T'' = T'$). Subsequently, every time {\small{IDENTITY\_GENERATOR}} gets called in Line~4 of Algorithm~\ref{alg:check-min-gen-set}, we can initialize Algorithm~\ref{alg:generator-IS} at Line~19, using the precomputed $\overline{K}$ and setting $\mu = \mu_0$. In fact, if $\mu_0 = 1$, then there is nothing to compute -- we know that {\small{IDENTITY\_GENERATOR}} will return $\tilde{p} = \omega I$, which also means that the membership check in Line~5 of Algorithm~\ref{alg:check-min-gen-set} will succeed. For the case $\mu_0 \neq 1$, there are also ways to speed up the execution of Lines~19-26 of Algorithm~\ref{alg:generator-IS}. The main bottleneck is Line~20 which has a complexity of $O(rn)$ operations. But this can be reduced to $O(r)$ using another precomputation step: note that we need to compute $q \leftarrow \prod_{\ell=0}^{r-1} (p^{\gamma_{\ell}}q_{\ell})^{\overline{K}_{\ell j}}$ in Line~20, which can be simplified as $\prod_{\ell=0}^{r-1} (p^{\gamma_{\ell}}q_{\ell})^{\overline{K}_{\ell j}} = \left( \prod_{\ell=0}^{r-1} q_{\ell}^{\overline{K}_{\ell j}} \right) p^{\sum_{\ell=0}^{r-1} \gamma_{\ell} \overline{K}_{\ell j}}$, and thus the quantity $\left( \prod_{\ell=0}^{r-1} q_{\ell}^{\overline{K}_{\ell j}} \right)$ can be precomputed for each $j=0,1,\dots,s-1$ (here $s$ is the number of columns of $\overline{K}$). Putting everything together, and ignoring the precomputation step, we see that Algorithm~\ref{alg:check-min-gen-set} has a run time complexity of $O(sd^r (r + M(d)\log\log(d)))$, where $M(d)$ is the cost of integer multiplication of non-negative integers less than $d$. The computational complexity of the precomputation step is upper bounded by the complexity of Algorithm~\ref{alg:generator-IS}. The exponential factor of $d^r$ in the complexity of Algorithm~\ref{alg:check-min-gen-set} is undesirable, and coming up with a more efficient algorithm is left for future work.

We finish this subsection by stating a special case of Theorem~\ref{thm:min-gen-set-simple-form} below, when $d$ is prime. In this case, the theorem simplifies.

\begin{lemma}
\label{lem:d-prime-min-gen-sets}
Given $S \subseteq \hwgrp$, suppose that $T := T' \cup \{p\}$ is a near-minimal generating set of $\langle S \rangle$ with $\langle p \rangle = I_S$. Let $d$ be prime. Then the following conditions are equivalent.
\begin{enumerate}[(i)]
    \item $T$ is a minimal generating set of $\langle S \rangle$.
    \item $\langle T' \rangle$ is a stabilizer subgroup of $\hwgrp$, and $I_S \neq \{I\}$.
\end{enumerate}
\end{lemma}
\begin{proof}
\edit{See Appendix~\ref{app:proofs_sec6}.}
\end{proof}

\subsection{A Gram-Schmidt generating set of a subgroup}
\label{ssec:canonical-form-gen-set}
The symplectic Gram-Schmidt procedure (described in \cite{wilde2009logical} for qubits and easily extended to prime $d$) takes a generating set $S$ for a Pauli subgroup $G=\langle S\rangle$ and returns another generating set $S'=S_1\cup S_2\cup U$, where $S_1$, $S_2$ is a collection of non-commuting pairs (in the qubit case, anti-commuting pairs), $U$ is a subset of the center $\mathcal{Z}(G)$ of $G$,\footnote{Recall that, for a group $G$, its center $\mathcal{Z}(G)$ is the subgroup of elements in $G$ that commute with everything in $G$.} and $G=\langle S'\rangle$. Note, in particular, that for such a generating set $S_1$, $S_2$, and $U$ must be disjoint. Here we describe a procedure achieving the same result, a Gram-Schmidt generating set, for qudit Pauli groups.

\begin{lemma}[\edit{Gram-Schmidt generators}]\label{lem:gs_generators}
Let $S=\{s_0,s_1,\dots,s_{k-1}\}\subseteq\hwgrp$ and let $A_{ij}=\llbracket s_i,s_j\rrbracket_d$ define a matrix $A\in\mathbb{Z}_d^{k\times k}$. The group $G:=\langle S\rangle$ has a Gram-Schmidt generating set $S_1\cup S_2\cup U$ where $|S_1|=|S_2|=\Theta(M_A)/2$, $M_A$ is the submodule of $\mathbb{Z}_d^k$ generated by the columns of $A$, and moreover there does not exist a Gram-Schmidt generating set with smaller $S_1$ or $S_2$.
\end{lemma}

\begin{proof}
Constructing a Gram-Schmidt generating set makes use of the alternating Smith normal form from Lemma~\ref{lem:zdasnf}. Denote by $Q \in \mathbb{Z}_d^{k \times 2n}$ the matrix whose rows correspond to each Pauli in $S$, ignoring the phase factors. Then $A = Q \Lambda Q^T \in \mathbb{Z}_d^{k \times k}$. Now Lemma~\ref{lem:zdasnf} implies the existence of $L, B \in \mathbb{Z}_d^{k \times k}$ such that $A = LBL^T$, where $L$ is invertible and $B$ is of the form specified by the lemma, whose top-left block has the form $\bigoplus_{i=1}^r\left(\begin{smallmatrix}0&\beta_i\\-\beta_i&0\end{smallmatrix}\right)$ for non-zero $\beta_1,\beta_2,\dots,\beta_r$, and all other elements are zero. Here $r=\Theta(M_A)/2$. Let $S'=\{s'_0,\dots,s'_{k-1}\}$, where $s'_i := \prod_{j=0}^{k-1} s_j^{L^{-1}_{ij}}$ with ordering of the product irrelevant to our result. Because the commutation relations of $S$ are given by $A=Q\Lambda Q^T$, the commutation relations of $S'$ are given by $(L^{-1}Q)\Lambda(L^{-1}Q)^T=B$. Therefore, in particular it is a Gram-Schmidt generating set $S' = S_1 \cup S_2 \cup U$ with $|S_1| = |S_2| = r$. Moreover, the invertibility of $L$ gives $\langle S \rangle = \langle S' \rangle$ via Theorem~\ref{thm:equiv-gen}. Thus, this procedure gives the required Gram-Schmidt generating set.

Now we prove that there is no Gram-Schmidt generating set containing fewer non-commuting pairs. Suppose for contradiction, there exists sets $T, H \subseteq G$ such that  $\langle T, H \rangle = G$, $|T| < 2r$, and $H\subseteq\mathcal{Z}(G)$. Construct a multiset $E = T \cup J \cup H$, where $J$ is a multiset of $2r-|T|$ phaseless identity Paulis. Let $P, P' \in \mathbb{Z}_d^{(2r+|H|+|U|) \times 2n}$ be matrices such that the rows of $P$ (resp.~$P'$) represent the Paulis in the multiset $E\cup U$ (resp.~$S'\cup H$), modulo the phase factors. We first note that since $\langle E, U \rangle = \langle S', H \rangle = G$, this implies that the submodules generated by the columns of $P^T$ and $P'^T$ are equal. Thus, by Lemma~\ref{lem:span-same} there exists an invertible matrix $V$ such that $P = VP'$. The number of non-zero entries in the ASNF of $B'=\left(\begin{smallmatrix}B&0\\0&0\end{smallmatrix}\right)=P'\Lambda P'^T$ is $2r$, and by the uniqueness part of the ASNF in Lemma~\ref{lem:zdasnf}, this should equal the number of non-zero entries in the ASNF of $P\Lambda P^T=VB'V^T$. However, by the assumptions on the sets $T$ and $H$,
\begin{equation}
C = P\Lambda P^T = \left(\begin{matrix} D & 0\\0 & 0\end{matrix}\right),
\end{equation}
where $D \in \mathbb{Z}_d^{|T| \times |T|}$, which in turn implies that $C$ cannot have more than $|T|<2r$ non-zero entries in its ASNF, obtaining our contradiction.
\end{proof}

With Lemma~\ref{lem:gs_generators}, we can make a Gram-Schmidt generating set $S_1\cup S_2\cup U$ with minimally sized $S_1$ and $S_2$. However, there is no such minimality guarantee on $U$. If $d$ is prime, the only elements of $\mathcal{Z}(G)$ that can be generated by $\langle S_1,S_2\rangle$ are exactly those in $K=\{\omega^jI:j\in\mathbb{Z}_d\}$. If $d$ is not prime, it is possible for non-trivial elements of $\mathcal{Z}(G)$ to be generated by $\langle S_1,S_2\rangle$, including elements of $U$. The ways this can happen is limited however.

\begin{lemma}
\label{lem:center_of_noncomm_pairs}
Suppose $S=\{s_0,\dots,s_{k-1}\},T=\{t_0,\dots t_{k-1}\}$ is a collection of non-commuting pairs, $f_i=\llbracket s_i,t_i\rrbracket$ for every $i$, and let $H=\langle S,T\rangle$. Then 
\begin{equation}
\mathcal{Z}(H)=\langle \omega^{f_i}I,s_i^{a_i},t_i^{a_i} : i=0,\dots,k-1,a_i\in\mathbb{Z}\text{\space s.t.~}a_if_i \equiv 0\pmod{d} \rangle.
\end{equation}
\end{lemma}

\begin{proof}
Let $A=\langle \omega^{f_i}I,s_i^{a_i},t_i^{a_i} : i=0,\dots,k-1,a_i\in\mathbb{Z}\text{\space s.t.~}a_if_i \equiv 0 \pmod{d} \rangle$ and note that $J_H\subseteq A$. We see that $\omega^{f_i}I=[s_i,t_i]$, $s_i^{a_i}$ and $t_i^{a_i}$ are all in $\mathcal{Z}(H)$ because they are products of generators of $H$ and they commute with all generators of $H$. Thus, $A\subseteq\mathcal{Z}(H)$.

To show containment in the other direction, let $p\in\mathcal{Z}(H)$. Therefore, $p$ can be written in terms of the generating set of $H$ as
\begin{equation}
p=\omega^c\prod_{i=0}^{k-1}s_i^{a_i}t_i^{b_i}
\end{equation}
for some $a_i,b_i\in\mathbb{Z}_d$ and $\omega^cI\in J_H\subseteq A$. Evaluate the commutator: $\llbracket p,t_j\rrbracket=\llbracket s_j^{a_j},t_j\rrbracket=a_j\llbracket s_j,t_j\rrbracket=a_jf_j$. Since $p\in\mathcal{Z}(H)$ the value of this commutator modulo $d$ must be 0, implying $a_jf_j\equiv 0 \pmod{d}$. This argument was independent of $j$, and an analogous argument works for commutators with $s_j$, implying $b_jf_j\equiv0 \pmod{d}$. Therefore, $p\in A$, completing the proof.
\end{proof}

Let $U=\{u_0,\dots,u_{k-1}\}$. One can use the Howell normal form and the approach outlined in Section~\ref{sssec:howell-normal-forms} to check membership of $\pi_2(u_0)$ in $\mathcal{Z}(\langle S_1,S_2\rangle)$ and, if it is, check that the phase can be corrected by an element of $I_{\langle S_1,S_2\rangle}$. If that is also the case, then $u_0$ is redundant and can be removed from $U$. Otherwise, continue by checking whether $u_1$ is in $\mathcal{Z}(\langle S_1,S_2,u_0\rangle)$ and so forth.

\subsection{\texorpdfstring{Sizes of subgroups of $\hwgrp$}{}}
\label{ssec:subgroups-noncomm-pairs}

In this section, we investigate the sizes of subgroups of $\hwgrp$ given their generating sets. First, let's assume nothing about the generating set. Using the Smith normal form, we can prove the following.

\begin{lemma}[\edit{Pauli subgroup size}]
\label{lem:subgroup-size}
Given an ordered multiset $S := \{q_0,q_1,\dots,q_{k-1}\} \subseteq \hwgrp$, suppose that the invariant factors of $\pi_2(S)$ are $d_0,d_1,\dots,d_{r-1}$. Then $|\langle S \rangle| = |I_S| \; \prod_{j=0}^{r-1} |d_i \mathbb{Z}_d| = |I_S| \; \prod_{j=0}^{r-1} \left( d / \gcd(d_i,d) \right)$.
\end{lemma}

\begin{proof}
Firstly, note that since invariant factors are unique up to units in $\mathbb{Z}_d$, the sizes of the ideals $|d_i \mathbb{Z}_d|$ do not depend on this choice; so the expression for $|\langle S \rangle|$ is well-defined. Thus, let $\pi_2(S)$ have a Smith normal form $D \in \mathbb{Z}_d^{2n \times k}$, a diagonal matrix satisfying $\pi_2(S) P = QD$, where $P \in \mathbb{Z}_d^{k \times k}$ and $Q \in \mathbb{Z}_d^{2n \times 2n}$ are invertible, and the first $r$ (and only) non-zero diagonal entries of $D$ are $d_0,d_1,\dots,d_{r-1}$. Let the first $r$ columns of $Q$ be $Q_0,Q_1,\dots,Q_{r-1}$. Then the submodule of $\mathbb{Z}_d^{2n}$ generated by the columns of $\pi_2(S)$ equals that generated by the set $\{d_i Q_i : i=0,1,\dots,r-1\}$; call this submodule $M$. It is then clear that $|\langle S \rangle| = |I_S| \; |\langle S \rangle / I_S| = |I_S| \; |M|$. 

It follows from the structure theorem of finitely generated modules over a principal ideal ring \cite[Chapter~15]{brown1993matrices} that $|M| = \prod_{j=0}^{r-1} |d_j \mathbb{Z}_d|$. But it is possible to prove the same in an elementary fashion. Since $Q$ is invertible, any non-empty subset of the columns of $Q$ generate a free submodule. Thus $Q_0,Q_1,\dots,Q_{r-1}$ form a basis, and it follows that $|M| = \prod_{j=0}^{r-1} |\{x d_j Q_j : x \in \mathbb{Z}_d\}| = \prod_{j=0}^{r-1} |\{x d_j : x \in \mathbb{Z}_d\}| = \prod_{j=0}^{r-1} \left( d / \gcd(d_j,d) \right)$.
\end{proof}
\edit{In the special case that $S$ is a qudit stabilizer group, i.e.~an abelian subgroup of the $n$-qudit Pauli group, there is a codespace consisting of qudit states that are $+1$ eigenvectors of all the elements of $S$. This codespace has dimension $d^n/|\langle S\rangle|$ \cite{gheorghiu2014standard}, which we can efficiently calculate using Lemma~\ref{lem:subgroup-size}.}

Next, consider subgroups of $\hwgrp$ that are generated by non-commuting pairs of qudit Paulis. Such subgroups enjoy some surprising properties. For example, if $S = \{s_0,s_1,\dots,s_{k-1}\}$ and $T = \{t_0,t_1,\dots,t_{k-1}\}$ is a set of non-commuting pairs, then $S \cup T$ is an independent set of generators, in the sense that the group generated by removing any generator from $S \cup T$ leads to a proper subgroup of $\langle S, T \rangle$. To see this note that if $s_0 \in \langle S \setminus \{s_0\}, T \rangle$, then it must imply that $\llbracket s_0,t_0\rrbracket_d = 0$, which is a contradiction. We could have argued the same with any other generator of $S \cup T$.

A particularly interesting result is Theorem~\ref{thm:square-free-max-pairs}, where it is shown that a maximum collection of non-commuting pairs generate the entire qudit Pauli group $\overline{\mathcal{P}}_n$ when $d$ is square free. The following lemma proves useful.

\begin{lemma}
\label{lem:group_counting}
Let $p,q\in\overline{\mathcal{P}}_n$ and $G\subseteq\overline{\mathcal{P}}_n$. If every element in $G$ commutes with both $p$ and $q$ and $\llbracket p,q\rrbracket_d=c$ is a unit in $\mathbb{Z}_d$, then $|\langle p,q,G\rangle/K|=d^2|\langle G\rangle/I_G|$ and $|\langle p,q,G\rangle|=d^3|\langle G\rangle/I_G|$.
\end{lemma}

\begin{proof}
The second equality differs from the first in that it includes arbitrary phases, which one can create from products of $p$ and $q$: $\{(pqp^\dag q^\dag)^j:j=0,1,\dots,d-1\}=\{\omega^{jc}I:j=0,1,\dots,d-1\}=K$, as $c$ is a unit in $\mathbb{Z}_d$. 

To prove the first equality, note $|\langle p\rangle/ I_{\{p\}}|=|\langle q\rangle/ I_{\{q\}}|=d$. If $|\langle p\rangle/ I_{\{p\}}|$ were not $d$, then $p^a\in K$ for $0<a<d$ and thus $0=\llbracket p^a,q\rrbracket_d \equiv ac \pmod{d}$, contradicting the fact that $c$ is a unit. Next, we see that for any $j \not \equiv 0 \pmod d$, $p^j$ is not in $\langle q,G\rangle$ even up to a phase. If it were, then $p^j$ would commute with $q$, but $\llbracket p^j,q\rrbracket_d=jc \bmod d$, which is not zero because $c$ is a unit and $j \not \equiv 0 \pmod d$. Likewise, $q^j$ is not in $\langle G\rangle$ up to a phase for any $j \not \equiv 0 \pmod{d}$. Thus, $|\langle p,q,G\rangle/K|=|\langle p\rangle/I_{\{p\}}| \; |\langle q,G\rangle / I_{\{q\} \cup G}| = |\langle p\rangle/ I_{\{p\}}| \; |\langle q\rangle/I_{\{q\}}| \; |\langle G\rangle/I_G| = d^2|\langle G\rangle/I_G|$.
\end{proof}

\begin{theorem}[\edit{The square-free theorem}]
\label{thm:square-free-max-pairs}
Suppose $d=p_0p_1\dots p_{m-1}$ is square-free. If $S,T$ is a collection of $nm$ non-commuting pairs, then $\langle S,T\rangle=\overline{\mathcal{P}}_n$.
\end{theorem}
\begin{proof}
A collection of $nm$ non-commuting pairs is a maximum size such collection, characterized by Theorem~\ref{thm:nm-case}. This implies that, for the appropriate order of the set $S\cup T$, the commutator matrix $C$ is block-diagonal with $2\times2$ blocks: $C=\bigoplus_{j=0}^{m-1}C_j=\bigoplus_{j=0}^{m-1}\bigoplus_{i=0}^{n-1}C_{ij}$, $C_{ij}=\left(\begin{smallmatrix}0&f_{ij}\\-f_{ij}&0\end{smallmatrix}\right)$, and $\gcd(f_{ij},d)=\prod_{h\neq j}p_h$ for all $i,j$.

Put $C$ into alternating Smith normal form $C=LBL^T$ using Lemma~\ref{lem:zdasnf}. Due to the simple form of $C$, it is easy to evaluate $d_{2i}$, the greatest common divisors of all $2i\times 2i$ matrix minors of $C$:
\begin{equation}
d_{2i}\equiv\bigg\{\begin{array}{ll}1,&i\le n\\d^{2(i-n)},&i>n\end{array}.
\end{equation}
Therefore, for appropriate invertible $L$, we have $B=\bigoplus_{i=1}^n\left(\begin{smallmatrix}0&1\\-1&0\end{smallmatrix}\right)$. 

The Smith normal form implies the existence of a collection of non-commuting pairs, $S'=\{s'_0,\dots,s'_{n-1}\}$ and $T'=\{t'_0,\dots,t'_{n-1}\}$ where $\llbracket s'_i,t'_i\rrbracket_d=1$, and $G\subseteq\overline{\mathcal{P}}_n$ that commutes with every element of $S'$ and $T'$. Moreover $\langle S',T',G\rangle=\langle S,T\rangle$. More specifically, $S',T',$ and $G$ can be found from $S$ and $T$ by taking appropriate products specified by $L$.

By Lemma~\ref{lem:group_counting}, $|\langle S',T'\rangle|=d^{2n+1}$. This is the size of $\overline{\mathcal{P}}_n$. Therefore, 
$G=\{I\}$ and $\overline{\mathcal{P}}_n=\langle S',T'\rangle=\langle S,T\rangle$.
\end{proof}

A collection of non-commuting pairs $S,T$ that generates the full Heisenberg-Weyl Pauli group $\hwgrp$ is a convenient generating set $B=S\cup T=\{s_0,t_0,\dots,s_{k-1},t_{k-1}\}$ for the entire group. If we have $p\in\hwgrp$, it has some decomposition in terms of the generators in $B$ which can be rearranged to
\begin{equation}
p=\omega^c\prod_{i=0}^{k-1}s_i^{a_i}t_i^{b_i},
\end{equation}
where $\omega^cI\in J_B$ and $a_i,b_i\in\mathbb{Z}_d$. We can write $\omega^cI$ in terms of a product of the generators by using a modified form of Algorithm~\ref{alg:generator-IS} that records which product of generators gives $\omega I$. Also, using the fact that $S,T$ are non-commuting pairs, we have $\llbracket p,t_i\rrbracket=a_i\llbracket s_i,t_i\rrbracket$, for every $i=0,1,\dots,k-1$. In particular, $a_i=\llbracket p,t_i\rrbracket/\llbracket s_i,t_i\rrbracket\mod{d}$ must be an integer. Likewise, we have $b_i=\llbracket s_i,p\rrbracket/\llbracket s_i,t_i\rrbracket\mod{d}$, for every $i$. This makes the decomposition of $p$ into \edit{generators from} $B$ rather simple, just amounting to calculation of $\llbracket s_i,p\rrbracket$ and $\llbracket p,t_i\rrbracket$, and the inclusion of an overall phase.


\edit{To conclude this section, we note that for more general non-square-free $d$, we can get lower bounds on the group size generated by non-commuting pairs. We state this as a corollary of Lemma~\ref{lem:group_counting_nonunit}, which we state and prove in Appendix~\ref{app:proofs_sec6}.}

\begin{corollary}
\label{cor:group-size-noncomm-pairs-lb}
Suppose $S=\{s_0,\dots,s_{k-1}\},T=\{t_0,\dots t_{k-1}\}$ is a collection of non-commuting pairs. For every $i$, let $f_i=\llbracket s_i,t_i\rrbracket_d$, and denote the order of $f_i$ in $\mathbb{Z}_d$ as $a_i := d/\gcd(f_i,d)$. Then we have the lower bound $|\langle S, T \rangle / I_{S\cup T} | \geq \prod_{i=0}^{k-1} a_i^2$.
\end{corollary}

\editt{\section{Discussion}

We have applied the theory of modules over commutative rings to study properties of the qudit Pauli group. One problem we have left open is the maximum size of a non-commuting set on $n>1$ qudits. Though we have provided some lower bounds on the size, we are also lacking a non-trivial upper bound. A second problem that we left open is the efficient construction of minimal generating sets of qudit Pauli subgroups in all cases. The missing step to get from a near-minimal generating set to a minimal one involves determining whether appropriate phases can be added to each generator as in Theorem~\ref{thm:min-gen-set-simple-form}.

Further work on qudit group theory might involve applying techniques similar to those we developed here to selecting an element of the qudit Clifford group uniformly at random. If this line of work parallels similar developments in the qubit case \cite{koenig2014efficiently,bravyi2021hadamard}, interesting structure of the qudit Clifford group could be discovered.}
\section*{Acknowledgments}
\label{sec:ack}

The authors would like to acknowledge Sergey Bravyi, Andrew Cross, Robert K\"{o}nig, Seth Merkel, and John Smolin for motivating discussions. We are also extremely grateful to Prof.~Brian Conrad for bringing to our attention the existence of the ASNF for bilinear alternating forms (Theorem~\ref{thm:lang-bilinear-canonical-form}), which was key to this paper, and for clarifying several aspects of commutative algebra that arose in this work.

\appendix
\appendixpage
\section{An Alternating Smith Normal Form}
\label{app:alternating-smith-normal-form}

\subsection{Bilinear alternating forms}
\label{ssec:alternating-forms}

For a commutative ring $R$ with multiplicative identity, consider a $R$-module $M$. Let $g: M \times M \rightarrow R$ be a \textit{bilinear form}. This means that for all $x,y,z \in M$ and $a \in R$, we have (i) $g(x+y,z) = g(x,z) + g(y,z)$, (ii) $g(x,y+z)=g(x,y)+g(x,z)$, and (iii) $g(a \cdot x,y) = g(x,a \cdot y) = a g(x,y)$. We say that $g$ is \textit{alternating} if $g(x,x)=0$ for all $x \in M$, which implies $g(x,y)+g(y,x)=0$ for all $x,y \in M$. A key fact about alternating bilinear forms is the following theorem:

\begin{theorem}[{\cite[Ch~ XV, Exercise~17]{lang2012algebra}}]
\label{thm:lang-bilinear-canonical-form}
Let $M$ be a finitely generated free module of rank $k$ over a commutative ring $R$ with multiplicative identity, and let $g: M \times M \rightarrow R$ be a bilinear alternating form. If $R$ is a PIR, then there exists a basis $\{e_i : i = 0,1,\dots,k-1\}$ of $M$, and a non-negative integer $r$ with $2r \leq k$, such that 
\begin{enumerate}[(i)]
    \item $g(e_{2i}, e_{2i+1}) = - g(e_{2i+1}, e_{2i}) = a_i \neq 0$, for $i=0,1,\dots,r-1$, where $a_0,a_1,\dots,a_{r-1} \in R$, while $g(e_i,e_j)=0$ for all other pairs of indices $i, j$,
    \item $a_0 \mid a_1 \mid \dots \mid a_{r-1}$.
\end{enumerate}
Moreover, the ideals $\{b a_i : b \in R\}$ are uniquely determined, for every $i=0,1,\dots,r-1$, irrespective of the choice of the basis.
\end{theorem}

Now suppose that $R$ is a principal ideal ring, and $A \in R^{k \times k}$ is an alternating matrix over $R$ (see beginning of Section~\ref{ssec:comm-matrix-relations} for definition of an alternating matrix). Then the module $R^k$ is a finitely generated free $R$-module of rank $k$, as mentioned previously. The matrix $A$ defines a bilinear alternating form $g: R^k \times R^k \rightarrow R$ as follows:
\begin{equation}
    g(x,y) := x^T A y, \;\; \text{for all } x,y \in R^k.
\end{equation}
We can then apply Theorem~\ref{thm:lang-bilinear-canonical-form} to this setting to extract a basis $\{e_i : i = 0,1,\dots,k-1\}$ of $R^k$, a non-negative integer $r \leq \lfloor k/2 \rfloor$, and $a_0,a_1,\dots,a_{r-1} \in R$ as in the theorem, such that $B \in R^{k \times k}$ defined by the equation $L B L^T = A$ is an alternating matrix, where the columns of $(L^{-1})^T \in R^{k \times k}$ are $e_0,e_1,\dots,e_{k-1}$, and 
\begin{equation}
    B_{jj'} = 
    \begin{cases}
        a_{j/2} \;\; & \text{if } (j,j') \in \{(2i, 2i+1): i=0,1,\dots,r-1\}, \\
        -a_{j'/2} \;\; & \text{if } (j',j) \in \{(2i, 2i+1): i=0,1,\dots,r-1\}, \\
        0 \;\; & \text{otherwise}.
    \end{cases}
\end{equation}
Furthermore, by applying \cite[Corollary~5.16]{brown1993matrices} we can conclude that $L$ is an invertible matrix over $R$, and then by Lemma~\ref{lem:min-genarators-modules} and Lemma~\ref{lem:min-number-gens-special-matrix} we also know that $\Theta(M_C) = \Theta(M_A) = 2r$, where $M_A$ and $M_C$ are the submodules of $R^k$ generated by the columns of $A$ and $C$ respectively. This decomposition $A = L B L^T$ is known as the \textit{alternating Smith Normal form} (ASNF), for example see \cite[Theorem~18]{kuperberg2002kasteleyn} where it is done assuming $R$ is a PID. In this paper, we are interested in the specific cases when $R$ is either $\mathbb{Z}$ or $\mathbb{Z}_d$. In these cases, one even has an explicit formula for the non-zero entries of $B$ in the ASNF, in terms of the minors of the matrix $C$. In the next subsection, Lemma~\ref{lem:asnf} summarizes all these facts when $R = \mathbb{Z}$, and provides a complete proof that has the added advantage of specifying an explicit algorithm to compute the ASNF. We also discuss how Lemma~\ref{lem:asnf} easily generalizes to the case $R = \mathbb{Z}_d$, leading to Lemma~\ref{lem:zdasnf}, which we already encountered before.

\subsection{An algorithm for the ASNF}

For the following lemma, for $a,b \in \mathbb{Z}$ we define the greatest common divisor of $a$ and $b$ to be the largest positive integer that divides both $a$ and $b$ over $\mathbb{Z}$, and denote it $\gcd (a,b)$. We also define $\gcd (0,0)=0$, and $\gcd(a)=a$ for any $a \in \mathbb{Z}$.


\begin{lemma}
\label{lem:asnf}
Suppose $A\in\mathbb{Z}^{k\times k}$ is an alternating matrix. Then, there are matrices $L,B\in\mathbb{Z}^{k\times k}$, where $B$ is alternating and has at most one non-zero entry per row and column, and $L$ is invertible, such that $A=LBL^T$. Moreover, we may further arrange $B$ so that it is non-zero only in the top-left $2r \times 2r$ block which has the form $\bigoplus_{i=1}^r\left(\begin{smallmatrix}0&\beta_i\\-\beta_i&0\end{smallmatrix}\right)$ for integer $r=\Theta(M_A)/2$, where $M_A$ is the $\mathbb{Z}$-submodule generated by the columns of $A$, and each $\beta_i \in \mathbb{Z}$ non-zero, satisfying $\beta_i\mid\beta_{i+1}$ for all $i < r$. Moreover, for all $i=1,2,\dots,r$, we have $|\beta_i|=d_{2i}/d_{2i-1}=d_{2i-1}/d_{2i-2}$, $\beta_i^2=d_{2i}/d_{2i-2}$, and for all $i > 2r$ we have $d_{i}= 0$, where $d_j$ is the greatest common divisor of all $j\times j$ minors of $A$ (with $d_0:=1$), which implies that the $\beta_i$ are unique up to choice of $\pm$ sign.
\end{lemma}

\begin{remark}
In the context of the lemma above, we have a chain of divisibilities $d_0 \mid d_1 \mid \dots \mid d_k$. This is because for any $j \geq 0$, any $(j+1) \times (j+1)$ minor of $A$ is divisible by $d_j$, and thus $d_j$ must also divide $d_{j+1}$. This means that if there is any $j' \leq k$ such that $d_{j'}=0$, then it automatically implies that $d_j = 0$ for all $j > j'$. Suppose such a $j'$ exists and moreover assume that $d_j \neq 0$ for all $j < j'$ (for example if $k$ is odd, then $d_k= \det(A)=0$ by property of alternating matrices). In other words, suppose $j'$ is the smallest integer such that all $j' \times j'$ minors of $A$ are zero. Then we can see from the above lemma that $j' = 2r+1$, and hence must be odd. Thus we also conclude that $j' = \Theta(M_A) + 1$.
\end{remark}

This lemma appears in more or less general forms elsewhere. As alluded in the previous section, it is a consequence of Theorem~\ref{thm:lang-bilinear-canonical-form} from \cite{lang2012algebra}. Part of this lemma is also Theorem 18 in \cite{kuperberg2002kasteleyn} applied to the ring of integers. However, we go further and include the formulas for $\beta_i$. Finally, the analogous lemma for matrices over $\mathbb{F}_2$ appears as Lemma~B.1 in \cite{sarkar2021graph}. After the proof, we discuss two variations on the theme: Lemma~\ref{lem:zdasnf}, which replaces the ring $\mathbb{Z}$ with $\mathbb{Z}_d$, and a generalization to matrices over any principal ideal ring (of which $\mathbb{Z}_d$ is just one).


\begin{figure}
\centering
\includegraphics[width=0.4\textwidth]{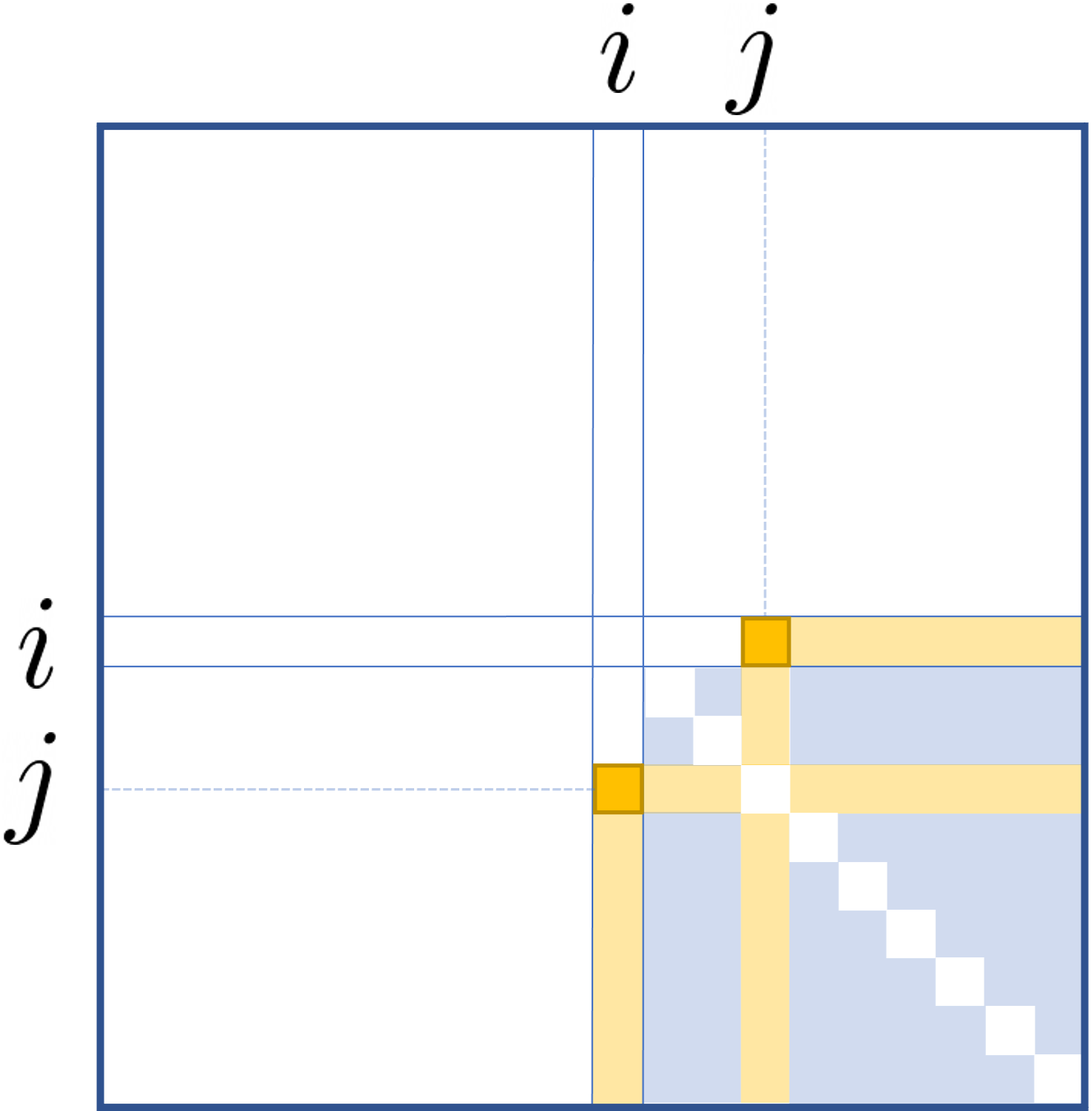}
\caption{A schematic of the matrix $B_{i-1}$ at the start of iteration $i$ of the alternating Smith normal form construction. Non-shaded regions are guaranteed to be zero except for at most $2(i-1)$ existing pivotal entries. Dark yellow squares indicate the entries that iteration $i$ is making into pivotal entries. Yellow-shaded regions contain entries that must be zeroed-out to complete iteration $i$.}
\label{fig:asnf}
\end{figure}

\begin{proof}[Proof of Lemma~\ref{lem:asnf}]
The construction of $B$ and $L$ is iterative. We let $B_0=A$.  Entries of a matrix that are the only non-zero entry in both their row and column we call ``pivotal" entries. Rows and columns containing those entries are also called pivotal. Suppose $B_{i-1}$ is an antisymmetric matrix with zero diagonal and each of its first $i-1$ rows and columns are either pivotal or all zero. Because it is antisymmetric with zero diagonal, the total number of pivotal rows (counting also those outside the first $i-1$ rows, if any) is even and equals the number of pivotal columns. See Fig.~\ref{fig:asnf} for a depiction of the structure of $B_{i-1}$.

Starting with $B_{i-1}$, we will construct a matrix $B_{i}=L_{i}B_{i-1}L_{i}^T$ with the same properties for one additional row and column, and $L_i$ is invertible. The sets of pivotal rows and columns of $B_i$ will contain the sets of pivotal rows and columns of $B_{i-1}$. For notational simplicity we denote $B_{i-1}=B'$ below. 

If row $i$ of $B'$ is already pivotal or all zero, then we can just set $B_i=B'$ and $L_i$ to the identity matrix. Otherwise, there exists a smallest $j$ such that $B'_{ji}=-B'_{ij}\neq0$. Note that $j>i$; otherwise, row $i$ must already be pivotal.

We now alternate two subroutines until a termination condition is met. The first subroutine zeroes out entries in the column $i$ other than $B'_{ji}$, and the second zeroes out entries in the column $j$ other than $B'_{ij}$. Updates are made to $B'$ throughout each subroutine, but it always remains antisymmetric with zero diagonal. Thus, the subroutines also simultaneously zero out entries in row $i$ and row $j$, respectively.

The fundamental task in both subroutines is the same --- we have two non-zero entries $a$ and $b$ in the same column, and wish to take linear combinations of their rows so that $a$ is replaced by $\gcd(a,b)>0$ and $b$ is replaced by $0$. By Bezout's identity, there exist integers $x,y$ such that $ax+by=\gcd(a,b)$ (and these can be found by performing the extended version of Euclid's GCD algorithm). Let $z=a/\gcd(a,b)$ and $w=b/\gcd(a,b)$. Then,
\begin{equation}\label{eq:Bezout_matrices}
\left(\begin{array}{cc}x&y\\-w&z\end{array}\right)\left(\begin{array}{c}a\\b\end{array}\right)=\left(\begin{array}{c}\gcd(a,b)\\0\end{array}\right),\quad\text{and}\quad \left(\begin{array}{cc}x&y\\-w&z\end{array}\right)^{-1}=\left(\begin{array}{cc}z&-y\\w&x\end{array}\right).
\end{equation}
Let $K\in\mathbb{Z}^{k\times k}$ be the identity matrix except for the entries $K_{\alpha\alpha}=x$, $K_{\alpha\beta}=y$, $K_{\beta\alpha}=-w$, $K_{\beta\beta}=z$, where $\alpha$ is the row index of $a$ and $\beta$ the row index of $b$. Then $B'\leftarrow KB'K^T$ has replaced entry $a$ with $\gcd(a,b)$ and $b$ with $0$, while also preserving the antisymmetry of $B'$ (since we likewise took linear combinations of columns by acting via $K^T$ on the right).

The first subroutine sequentially zeros out all entries $B'_{j'i}\neq0$ for $j'\neq j$, leaving a non-zero entry at just position $(j,i)$ in column $i$. The second subroutine sequentially zeros out all entries $B'_{i'j}$ for $i'\neq i$, leaving a non-zero entry at just position $(i,j)$ in column $j$. Note, the second subroutine may cause entries $B'_{j'i}$ to become non-zero again, undoing some of the work of the first subroutine, and likewise the first subroutine may cause entries $B_{i'j}$ to become non-zero. For this reason we must alternate application of subroutines until rows $i$ and $j$ are both pivotal. This termination is guaranteed -- the absolute values of entries at $(i,j)$ and $(j,i)$ are non-increasing because $\gcd(a,b)$ divides $|a|$. Moreover, if $\gcd(a,b)=|a|$ (and this holds for all entries $b$ we are trying to zero out in this iteration), then we may choose $y=0$ in Bezout's identity and the row containing $a$ remains unchanged (in absolute value) -- that is, the work of the previously applied subroutine is not undone, the rows $i$ and $j$ have been made pivotal, and we have constructed $B_i$ and $L_i$ (the latter a product of the appropriate sequence of matrices $K$ above).

After $k$ iterations, we have obtained invertible matrices $L_1,\dots,L_k$ such that
\begin{equation}
(L_kL_{k-1}\dots L_1)A(L_kL_{k-1}\dots L_1)^T=B_k
\end{equation}
where $B_k$ is antisymmetric with zero diagonal and at most one non-zero entry per row and column. There is a permutation matrix $P$ such that $B_{k+1}=PB_kP^T=\bigoplus_{i=1}^r\left(\begin{smallmatrix}0&\gamma_i\\-\gamma_i&0\end{smallmatrix}\right)$ for some integers $r$ and $\gamma_i$. We are still lacking the divisibility condition on entries of $B$, which we rectify next.

Consider the matrix of integers
\begin{equation}
T=\left(\begin{array}{cccc}0&a&0&0\\-a&0&0&0\\0&0&0&b\\0&0&-b&0\end{array}\right).
\end{equation}
Suppose we have Bezout's identity $ax+by=\gcd(a,b)$. We claim there exists invertible matrix $H$ such that 
\begin{equation}
HTH^T=\left(\begin{array}{cccc}0&\gcd(a,b)&0&0\\-\gcd(a,b)&0&0&0\\0&0&0&\lcm(a,b)\\0&0&-\lcm(a,b)&0\end{array}\right).
\end{equation}
This is exemplified by the following sequence of row and column operations:
\begin{align}\label{eq:T_to_HTH}
T\rightarrow\left(\begin{array}{cccc}0&a&0&b\\-a&0&0&0\\0&0&0&b\\-b&0&-b&0\end{array}\right)\rightarrow\left(\begin{array}{cccc}0&\gcd(a,b)&0&0\\-\gcd(a,b)&0&-by&0\\0&by&0&\lcm(a,b)\\0&0&-\lcm(a,b)&0\end{array}\right)\rightarrow HTH^T.
\end{align}
The first step adds row $3$ to row $1$ (and likewise on columns). The second uses Eq.~\eqref{eq:Bezout_matrices} to zero out entries $(1,4)$ and $(4,1)$. The third step zeros entries at $(2,3)$ and $(3,2)$ by subtracting $by/\gcd(a,b)$ times row $1$ from row $3$ (and likewise on columns).

Apply this procedure iteratively to pairs of $2\times2$ blocks of $B_{k+1}$. In iteration $i$, apply it sequentially $r-i$ times to block $i$ paired with each of blocks $j=i+1,\dots,r$. After iteration $i$, the non-zero entry of the $i^{\text{th}}$ block divides all entries in blocks $i+1,\dots,r$. After iteration $r-1$, we have an invertible $H$ so that $B_{k+2}=HB_{k+1}H^T=\bigoplus_{i=1}^r\left(\begin{smallmatrix}0&\beta_i\\-\beta_i&0\end{smallmatrix}\right)$ where $\beta_i\mid\beta_{i+1}$.

Set $L=(HPL_kL_{k-1}\dots L_1)^{-1}=L_1^{-1}L_2^{-1}\dots L_k^{-1}P^{-1}H^{-1}$ and $B=B_{k+2}$, so that $A=LBL^T$. Note that by combination of Lemmas~\ref{lem:min-genarators-modules} and \ref{lem:min-number-gens-special-matrix}, we deduce that $r=\Theta(M_A)/2$.

To show that $|\beta_i|=d_{2i}/d_{2i-1}=d_{2i-1}/d_{2i-2}$ and $\beta_i^2=d_{2i}/d_{2i-2}$, we note that invertible row and column operations do not change the greatest common divisor of matrix minors (see Corollary 4.8 in \cite{brown1993matrices} or the proof of Proposition 8.1 in \cite{miller2009differential}). Thus, since the formulas for $\beta_i$ are easily verified using matrix minors of $B=\bigoplus_{i=1}^r\left(\begin{smallmatrix}0&\beta_i\\-\beta_i&0\end{smallmatrix}\right)$, they hold also for $A$. 

We summarize the construction in this proof in Algorithm~\ref{alg:asnf}.
\end{proof}

\begin{algorithm}
\edit{
\caption{\edit{Given alternating matrix $A\in\mathbb{Z}^{k\times k}$, find $L,B\in\mathbb{Z}^{k\times k}$ such that $A=LBL^T$ and $B,L$ have the properties stated in Lemma~\ref{lem:asnf}.}\label{alg:asnf}}
\begin{algorithmic}[1]
    \Procedure{Alternating\_Smith\_Normal\_Form}{$A$}
        \State $B \leftarrow A$
        \State $L \leftarrow \mathsf{IdentityMatrix}(k\times k)$
        \For{$1 \leq i \leq k$}
            \If{Row $i$ of $B$ is all 0}
                \State Continue
            \EndIf
            \State $j \leftarrow $ the smallest index such that $B_{ij}\neq0$
            \While{$S=\{h>j:B_{ih}\neq0\}$ is not empty or $T=\{l>i:B_{lj}\neq0\}$ is not empty }
                \For{$h\in S$}
                    \State $x,y,g=\mathsf{ExtendedGCD}(B_{ij},B_{ih})$
                    \State $K \leftarrow \mathsf{IdentityMatrix}(k\times k)$
                    \State $K_{jj}\leftarrow x$, $K_{jh}\leftarrow y$, $K_{hj}\leftarrow -B_{ih}/g$, $K_{hh}\leftarrow B_{ij}/g$
                    \State $B\leftarrow KBK^T$
                    \State $L\leftarrow LK^{-1}$
                \EndFor
                \For{$l\in T$}
                    \State $x,y,g=\mathsf{ExtendedGCD}(B_{ij},B_{lj})$
                    \State $K \leftarrow \mathsf{IdentityMatrix}(k\times k)$
                    \State $K_{ii}\leftarrow x$, $K_{il}\leftarrow y$, $K_{li}\leftarrow -B_{lj}/g$, $K_{ll}\leftarrow B_{ij}/g$
                    \State $B\leftarrow KBK^T$
                    \State $L\leftarrow LK^{-1}$
                \EndFor
            \EndWhile
            \State $P\leftarrow $ permutation matrix so that $PBP^T$ is block diagonal with $r$ non-zero $2\times 2$ blocks 
            \State $B=PBP^T$
            \State $L=LP^{-1}$
            \For{$1\le i\le r-1$}
                \For{$i+1\le j\le r$}
                    \State $H\leftarrow $ the matrix enacting the transformation in Eq.~\eqref{eq:T_to_HTH} on blocks $i$ and $j$
                    \State $B=HBH^T$
                    \State $L=LH^{-1}$
                \EndFor
            \EndFor
        \EndFor
        \State Return $B,L$
    \EndProcedure
    
\# Note: $\mathsf{ExtendedGCD}(a,b)$ finds $x,y\in\mathbb{Z}$ and $0<g\in\mathbb{Z}$ such that $xa+yb=g=\text{gcd}(a,b)$. Also, if $\text{gcd}(a,b)=|a|$, $y=0$.
\end{algorithmic}}
\end{algorithm}

To start the commentary on Lemma~\ref{lem:asnf}, let us apply it more specifically to the ring $\mathbb{Z}_d$ that is relevant for $d$-dimensional qudits. 
Suppose we have an alternating matrix $\bar{A} \in\mathbb{Z}_d^{k\times k}$. Treating $\bar{A}$ as an element of $\mathbb{Z}^{k\times k}$, we use Lemma~\ref{lem:asnf} to find matrices $\bar{L},\bar{B}\in\mathbb{Z}^{k\times k}$ so that $\bar{A}=\bar{L} \bar{B} \bar{L}^T$. We note that invertibility of $\bar{L}$ implies that $\det (\bar{L}) = \pm 1$ over integers, since $\pm 1$ are the only units in $\mathbb{Z}$. We may reduce this expression modulo $d$ to obtain the equality $A=L B L^T$, where $A,B$, and $L$ are obtained by reducing $\bar{A}, \bar{B}$, and $\bar{L}$ modulo $d$, respectively. Note that $\bar{A} = A$, and $L$ is an invertible matrix in $\mathbb{Z}_d^{k\times k}$ since $\det (L) = \det (\bar{L}) \mod d = \pm 1$. The matrix $B$ is alternating in $\mathbb{Z}_d^{k\times k}$ since $\bar{B}$ was alternating in $\mathbb{Z}^{k\times k}$. Moreover, any zero entry in $\bar{B}$ remains zero in $B$, while some non-zero entries of $\bar{B}$ may become zero in $B$, due to the modulo $d$ operation. The divisibility and uniqueness parts of Lemma~\ref{lem:asnf} are also modified in accordance with the following observations on divisibility in $\mathbb{Z}_d$.


A greatest common divisor $\gcd(a,b)$ for $a,b\in\mathbb{Z}_d$ is any element of $\mathbb{Z}_d$ that divides $a$ and $b$ and has the property that if $c\in\mathbb{Z}_d$ divides both $a$ and $b$, then $c$ divides $\gcd(a,b)$. The group of units $U_d$ is the multiplicative group consisting of elements of $\mathbb{Z}_d$ that have multiplicative inverses in $\mathbb{Z}_d$. Elements of $U_d$ are the integers less than $d$ that are relatively prime to $d$. Multiplication by units does not affect divisibility, i.e.~if $a\mid b$, then $ua\mid b$ and $a\mid ub$ for all $u\in U_d$. Therefore, the value of $\gcd(a,b)$ is not unique, and if $g$ and $h$ are both greatest common divisors of $a$ and $b$, then there is a unit $u\in U_d$ such that $g=uh$, a fact that holds over commutative principal ideal rings \cite{kaplansky1949elementary,anderson2004associates}. These observations lead to the ASNF over the ring $\mathbb{Z}_d$ as stated in Lemma~\ref{lem:zdasnf} in the main text.

Finally, we note that Lemma~\ref{lem:asnf} can be generalized to matrices over commutative principal ideal rings. A commutative principal ideal ring $R$ is defined by two properties
\begin{enumerate}
\item Noetherian property (e.g.~Chapter 10 of \cite{lang2012algebra}): all ideals are finitely generated, or, equivalently, any ascending chain of ideals $I_0\subseteq I_1\subseteq I_2\subseteq\dots$ terminates with $I_n=I_m$ for all $m\ge n$.
\item Hermite property \cite{kaplansky1949elementary}: for any two elements $a,b\in R$, there exists $c\in R$ and an invertible matrix $T\in R^{2\times 2}$ such that $(a\text{\space}b)T=(c\text{\space}0)$.
\end{enumerate}

These properties generalize key steps in the proof of Lemma~\ref{lem:asnf}. The Hermite property replaces Eq.~\eqref{eq:Bezout_matrices}, which is used to zero out matrix entries. Note that it also implies $(a\text{\space}b)=(c\text{\space}0)T^{-1}$ and therefore $c\mid a$ and $c\mid b$. The Noetherian property is used to argue that the repeated application of the zeroing subroutines eventually terminates with a new pivotal row and column. Because $c\mid a$, the ideal generated by $c$ contains the ideal generated by $a$. Since the Noetherian property means the chain of ideals must terminate, the algorithm must also terminate. These modifications to the proof of Lemma~\ref{lem:asnf} lead to a more general theorem.

\begin{theorem}
\label{thm:asnf}
Let $R$ be a commutative principal ideal ring, and suppose $A\in R^{k\times k}$ is an alternating matrix. Then, there are matrices $L,B\in R^{k\times k}$, where $B$ is alternating and has at most one non-zero entry per row and column, and $L$ is invertible, such that $A=LBL^T$. We may further arrange $B$ so that it is non-zero only in the top-left $2r \times 2r$ block which has the form $\bigoplus_{i=1}^r\left(\begin{smallmatrix}0&\beta_i\\-\beta_i&0\end{smallmatrix}\right)$ for integers $r=\Theta(M_A)/2$, where $M_A$ is the $R$-submodule generated by the columns of $A$, and each $\beta_i \in R$ non-zero, satisfying $\beta_i\mid\beta_{i+1}$ for all $i < r$. Also, $\beta_i$ is uniquely determined up to multiplication by a unit and satisfies the formula $\beta_i d_{2i-1}=d_{2i}$ (or, alternatively, the formula $\beta_id_{2i-2}=d_{2i-1}$), where $d_j$ is a greatest common divisor in $R$ of all $j\times j$ minors of $A$ (and $d_0:=1$).
\end{theorem}

\edit{

\section{Technical lemmas and proofs}
\label{app:technical-proofs}

\subsection{Proofs for Section~\ref{sec:prelim}}\label{app:proofs_sec2}

In this section, we give proofs for some of the preliminary lemmas involving modules over commutative rings.

\begin{replemma}{lem:min-genarators-modules}
Suppose $C \in R^{k \times t}$ and $A \in R^{k \times k}$, $B \in R^{t \times t}$ are invertible matrices. Let $\bar{C} := ACB$. The minimal number of generators of the submodules $M_{C}$ and $M_{\bar{C}}$ of $R^k$ are equal.
\end{replemma}

\begin{proof}
Let $\hat{C} := CB$, let its columns be $\hat{c}_0,\hat{c}_1,\dots,\hat{c}_{t-1}$, and let $M_{\hat{C}}$ be the submodule of $R^k$ generated by the columns of $\hat{C}$. For every $i$, each column $\hat{c}_i \in M_{C}$ is a linear combination of the columns of $C$, from which we conclude that $M_{\hat{C}} \subseteq M_{C}$. By invertibility of $B$, we also have $C = \hat{C}B^{-1}$, and by the same argument we conclude that $M_{C} \subseteq M_{\hat{C}}$. Thus $M_{C} = M_{\hat{C}}$, and so we have $\Theta(M_{C}) = \Theta (M_{\hat{C}})$. It remains to show that $\Theta(M_{\bar{C}}) = \Theta(M_{\hat{C}})$, where $\bar{C} = A \hat{C}$.

Let $D \in R^{k \times t'}$, with $t' \leq t$, be a matrix whose columns form a minimal generating set of $M_{\hat{C}}$. Then each $\hat{c}_i$ can be expressed as a linear combination of the columns of $D$; so we can write $\hat{C} = D E$, for some $E \in R^{t' \times t}$. Letting $\bar{A} := AD \in R^{k \times t'}$, we thus have $\bar{C} = \bar{A} E$. This shows that $\Theta(M_{\bar{C}}) \leq t = \Theta(M_{\hat{C}})$. Again by invertibility of $A$, we also have $\hat{C} = A^{-1} \bar{C}$, and now repeating the argument gives $\Theta(M_{\hat{C}}) \leq \Theta(M_{\bar{C}})$.
\end{proof}

\begin{replemma}{lem:min-number-gens-special-matrix}
Let $C \in R^{k \times t}$ be a matrix such that each row and column has at most one non-zero element, and suppose one of those non-zero elements is divisible by all the others. Then the minimal number of generators of the $R$-module $M_{C}$, generated by the columns of $C$, is equal to the number of non-zero elements of the matrix $C$.
\end{replemma}

\begin{proof}
The lemma is clearly true if $C = 0$; so assume that $C$ has at least one non-zero element. Without loss of generality, we may assume that $C$ is a diagonal matrix so that $C_{00}, C_{11}, \dots, C_{r-1,r-1}$ all divide $C_{rr}\neq0$ (which implies $C_{ii} \neq 0$ for all $i=0,1,\dots,r$), where $r+1$ is the number of non-zero elements of $C$. If $C$ is not of this form, then one can permute the rows and columns of $C$ to bring it to this form, and Lemma~\ref{lem:min-genarators-modules} ensures that the minimal number of generators remain the same. Notice that if $x \in M_C$, then for $i=0,1,\dots,r$ we have $x_i = a_i C_{ii}$ for some $a_i \in R$, and $x_i = 0$ for all $i > r$.

Define the ideals $N_i := \{x \in R : x C_{ii} = 0\}$ for each $i=0,1,\dots,r$, and note that by the divisibility condition we have $N_0, N_1,\dots, N_{r-1} \subseteq N_r\neq R$. Let $N$ be a maximal ideal of $R$ containing $N_r$. Then $R/N$ is a field \cite[Chapter~2]{lang2012algebra}, and let $f: R \rightarrow R/N$ be the corresponding quotient map (which is a ring homomorphism) taking an element $y \in R$ to its coset in $R/N$. Treating $(R/N)^{r+1}$ as a $R$-module, we now define a $R$-module homomorphism $\kappa: M_C \rightarrow (R/N)^{r+1}$ as follows: if $x \in M_C$, then $(\kappa(x))_i = f(a_i)$ for every $i=0,1,\dots,r$, where $x_i = a_i C_{ii}$ (one can check that $\kappa$ is well-defined, i.e. it does not depend on the choice of $a_i$ as $N_i \subseteq N$, and that it is a module homomorphism). The map $\kappa$ is also clearly surjective. Thus if $M_C$ has a generating set $S$ of size $|S| < r+1$, then surjectivity of $\kappa$ implies that a generating set for $(R/N)^{r+1}$ is $\{\kappa(x) : x \in S\}$. But then $\{\kappa(x) : x \in S\}$ is also a generating set of the $R/N$-module $(R/N)^{r+1}$, which is a vector space of dimension $r+1$ as $R/N$ is a field, and thus it cannot have less than $r+1$ generators.
\end{proof}

\subsection{Proofs for Section~\ref{ssec:noncomm-set-max-size}}\label{app:proofs_sec4}

Recall that in Section~\ref{ssec:noncomm-set-max-size}, we studied a graph $G$ with a vertex for each single-qudit Pauli in $\mathcal{P}_1$ and edges between pairs that do not commute. A subset of vertices $W$ contains those corresponding to Paulis $X^aZ^b$ with $\gcd(a,b,d)=1$.
\begin{replemma}{lem:G_properties}
The graph $G=(V,E)$ has the following properties.
\begin{enumerate}[(i)]
\item If $C\subseteq V$ is a clique in $G$, then there is a clique $C'\subseteq W$ of the same size, $|C|=|C'|$.  
\item If $v_0,v_2\in V$, $v_1\in W$, $(v_0,v_1)\not\in E$, and $(v_1,v_2)\not\in E$, then $(v_0, 
v_2)\not\in E$.
\item If $(a,b)$ and $(s,t)$ are in $W$, then $((a,b),(s,t))\not\in E$ if and only if there exists $u\in\mathbb{Z}$ with $\gcd(u,d)=1$ such that $(a,b)=(us,ut)$, the right hand side evaluated modulo $d$.
\end{enumerate}
\end{replemma}
\begin{proof}
To prove (i), we show that if $C\subseteq V$ is a clique, then for any $v=(a,b)\in C$, $C'=(C\setminus\{v\})\cup\{v'\}$, where $v'=v/\gcd(a,b,d)\in W$, is also a clique. Suppose for contradiction, there is some $(s,t)\in C$ that is connected by an edge to $v$ but not to $v'$. Then, for $g=\gcd(a,b,d)$, we have $\det\left(\begin{smallmatrix}a/g&b/g\\s&t\end{smallmatrix}\right) \equiv 0 \pmod{d}$ which implies $\det\left(\begin{smallmatrix}a&b\\s&t\end{smallmatrix}\right) \equiv 0 \pmod{d}$, contradicting that $(s,t)$ is connected to $v$.

To prove (ii), let $v_i :=(a_i,b_i)$ for $i=0,1,2$. Note that by Bezout's identity, there are integers $x,y,z$ such that $xa_1+yb_1+zd=\gcd(a_1,b_1,d)=1$. Let
\begin{equation}
A :=\left(\begin{array}{ccc}a_0&b_0&xa_0+yb_0+zd\\a_1&b_1&1\\a_2&b_2&xa_2+yb_2+zd\end{array}\right) \in \mathbb{Z}^{3 \times 3}.
\end{equation}
Due to the last column of $A$ being a linear combination of the first two and $(d,d,d)^T$, we have $\det(A) \equiv 0 \pmod{d}$. Also, for $c_i=xa_i+yb_i+zd$, the cofactor expansion gives $\det(A)=c_0\det\left(\begin{smallmatrix}a_1&b_1\\a_2&b_2\end{smallmatrix}\right)-\det\left(\begin{smallmatrix}a_0&b_0\\a_2&b_2\end{smallmatrix}\right)+c_2\det\left(\begin{smallmatrix}a_0&b_0\\a_1&b_1\end{smallmatrix}\right) \equiv -\det\left(\begin{smallmatrix}a_0&b_0\\a_2&b_2\end{smallmatrix}\right) \pmod{d}$. Thus we have $\det\left(\begin{smallmatrix}a_0&b_0\\a_2&b_2\end{smallmatrix}\right) \equiv 0 \pmod{d}$, proving $(v_0,v_2)\not\in E$.

The reverse direction of part (iii) is a simple calculation. The forward direction is more involved. Apply the Smith normal form (Theorem~\ref{thm:smith-normal-form}) to obtain $A :=\left(\begin{smallmatrix}a&b\\s&t\end{smallmatrix}\right)=S^{-1}BT$ for invertible matrices $S,T\in\mathbb{Z}_d^{2\times2}$ and diagonal $B\in\mathbb{Z}_d^{2\times2}$ with $B_{00}=\gcd(a,b,s,t) \bmod{d}=\gcd(a,b,s,t)$ and $B_{11}=(\det(A)/B_{00}) \bmod{d}$. The formulas for $B_{00}$ and $B_{11}$ follow from \cite[Theorem~2.4]{stanley2016smith}, even though $\mathbb{Z}_d$ is not a unique factorization domain (as required to apply the theorem), but one may simply compute the Smith normal form of $A$ over $\mathbb{Z}$, which is a unique factorization domain, and then reduce the resulting matrices modulo $d$, leading to the formulas for $B_{00}$ and $B_{11}$. Since $((a,b),(s,t))\not\in E$, we know $\det(A) \equiv 0\pmod{d}$, and because we also have $\gcd(a,b,d)=1$, we get $B_{11}=0$. Rearrange to obtain $SA=BT$ and note that the last row of $BT$ is $(0,0)$. This implies
\begin{equation}\label{eq:linear_combo}
S_{10}(a,b)+S_{11}(s,t)=(0,0),
\end{equation}
evaluated over $\mathbb{Z}_d$. This equation means $\gcd(S_{10},d)$ divides $S_{11}s$ and $S_{11}t$ over integers. Because $\gcd(s,t,d)=1$, it must then be that $\gcd(S_{10},d)$ also divides $S_{11}$ over integers. Therefore, $\gcd(S_{10},d)$ divides $\det(S)$ over integers, and hence also over $\mathbb{Z}_d$. Since $S$ is invertible, $\det(S)$ (modulo $d$) is a unit of the ring $\mathbb{Z}_d$, so its only divisors are other units. This means $\gcd(S_{10},d)=1$. Likewise, we can argue $\gcd(S_{11},d)=1$. So $S_{10}$, $S_{11}$, and $S_{11} (S_{10})^{-1}$ are all units in $\mathbb{Z}_d$. Multiply Eq.~\eqref{eq:linear_combo} by $S_{10}^{-1}$ and rearrange to finish the proof.
\end{proof}

The other lemma we prove here involves the arithmetic function $H(d)=|W|/d$. We argued in Section~\ref{ssec:noncomm-set-max-size} that it is equal to
\begin{equation}
H(d)=\sum_{a\in\mathbb{Z}_d}\prod_{p\mid\gcd(a,d)}\left(1-\frac1p\right).
\end{equation}

\begin{replemma}{lem:H_is_multiplicative}
If $d_0,d_1\in\mathbb{Z}$ are relatively prime, then $H(d_0d_1)=H(d_0)H(d_1)$.
\end{replemma}
\begin{proof}
By Bezout's identity, there are integers $x_0,x_1$ such that $x_0d_0+x_1d_1=1$. Note that $x_0$ is an inverse of $d_0$ modulo $d_1$, so $\gcd(x_0,d_1)=1$. Likewise $\gcd(x_1,d_0)=1$. 

The Chinese remainder theorem implies there is an isomorphism between $\mathbb{Z}_{d_0d_1}$ and $\mathbb{Z}_{d_0}\times\mathbb{Z}_{d_1}$. Explicitly, it says that $a \equiv a_0\pmod{d_0}$ and $a \equiv a_1\pmod{d_1}$, if and only if $a \equiv a_0x_1d_1+a_1x_0d_0\pmod{d_0d_1}$. This means that 
$\gcd(a,d_0 d_1)= \gcd(a_0x_1d_1+a_1x_0d_0,d_0 d_1)$. We also have $\gcd(x,d_0d_1)=\gcd(x,d_0)\gcd(x,d_1)$ for any integer $x$, as $d_0,d_1$ are relatively prime.

Putting these facts together, we complete the proof:
\begin{align}
H(d_0d_1)&=\sum_{a\in\mathbb{Z}_{d_0d_1}}\prod_{p\mid\gcd(a,d_0d_1)}(1-1/p)\\&=\sum_{a_0\in\mathbb{Z}_{d_0}}\sum_{a_1\in\mathbb{Z}_{d_1}}\prod_{p|\gcd(a_0x_1d_1+a_1x_0d_0,d_0d_1)}(1-1/p)\\&= \left( \sum_{a_0\in\mathbb{Z}_{d_0}}\prod_{p|\gcd(a_0,d_0)}(1-1/p) \right) \left( \sum_{a_1\in\mathbb{Z}_{d_1}}\prod_{p\mid\gcd(a_1,d_1)}(1-1/p) \right) = H(d_0)H(d_1).
\end{align}
\end{proof}

\subsection{Proofs for Section~\ref{sec:achievable-patterns}}\label{app:proofs_sec5}

\begin{lemma}
\label{lem:helper-lem-2}
Let $d = p_0^{\alpha_0}p_1^{\alpha_1}\dots p_{m-1}^{\alpha_{m-1}}$ be the prime factorization of an integer $d \geq 2$, where $p_0,\dots,p_{m-1}$ are primes. Let $d', d'' \geq 2$ be two integers. Define $J := \{j \in \{0,\dots,m-1\} : d' \not \equiv 0 \pmod{p_j^{\alpha_j}}\}$, and suppose that $d'' \not \equiv 0 \pmod{p_j}$ for every $j \in J$. Then we have the set equality
\begin{equation*}
    \{x \bmod d: x \in d' \mathbb{Z}\} = \{x \bmod d: x \in d' d'' \mathbb{Z}\}.
\end{equation*}
\end{lemma}
\begin{proof}
If $J$ is empty, then $d'$ is divisible by $d$, so the statement is trivial. Thus assume $J \neq \emptyset$, so $d'$ is not divisible by $d$. It is clear that $\{x \bmod d: x \in d' d'' \mathbb{Z}\} \subseteq \{x \bmod d: x \in d' \mathbb{Z}\}$. Let $\ell$ be the smallest non-negative integer such that $\ell d' \equiv 0 \pmod{d}$. Then $\ell$ must be of the form $\ell = \prod_{j \in J} p_j^{\beta_j}$, where $1 \leq \beta_j \leq \alpha_j$ for each $j$. Note that $\ell$ is an upper bound on the size of the set $\{x \bmod d: x \in d' \mathbb{Z}\}$. Now consider the multiset of elements $K := \{t d' d'' \mod{d}: t = 0,\dots,\ell - 1\}$. We will show that all elements of $K$ are distinct, which will prove the lemma as it implies that $\ell$ is a lower bound on the size of the set $\{x \bmod d: x \in d' d'' \mathbb{Z}\}$.

Suppose this is not the case, and there exist distinct $t,t'  \in \{0,\dots,\ell-1\}$, with $t > t'$, such that $td' d'' \equiv t' d' d'' \pmod{d}$, or equivalently $(t-t') d' d'' \equiv 0 \pmod{d}$. By assumptions on $J$ and $d''$, we then conclude that $(t-t') d' \equiv 0 \pmod{d}$. Since $0 < t-t' \leq \ell - 1$, this implies that $\ell$ is not the smallest non-negative integer satisfying $\ell d' \equiv 0 \pmod{d}$, which is a contradiction.
\end{proof}

The following corollary now follows immediately from the above lemma by setting $d'=l_J$ and $d''=l_J^{r-1}$, for $r \geq 1$.
\begin{corollary}
\label{cor:helper-cor-2-1}
Let $d = p_0^{\alpha_0}p_1^{\alpha_1}\dots p_{m-1}^{\alpha_{m-1}}$ be the prime factorization of an integer $d \geq 2$, where $p_0,\dots,p_{m-1}$ are primes. Let $J \subseteq \{0,\dots,m-1\}$, and define $\ell_J := \prod_{j \in J} p_j^{\alpha_j}$. Then for every integer $r \geq 1$, we have  $\{x \bmod d: x \in \ell^r_J \mathbb{Z}\} = \{x \bmod d: x \in \ell_J \mathbb{Z}\} = \{k \ell_J : 0 \leq k < d/\ell_J \}$.
\end{corollary}

\subsection{Proofs for Section~\ref{sec:group_theory}}\label{app:proofs_sec6}

We start by giving a full proof of the following theorem on transforming generating sets of Pauli groups.
\begin{reptheorem}{thm:equiv-gen}\edit{(Equivalent generating sets)}
Suppose $S := \{q_0,q_1,\dots q_{k-1}\}$ is an ordered multiset of elements of $\hwgrp$. Let $A \in \mathbb{Z}_d^{k \times k}$ be an invertible matrix, and consider the ordered multiset $T := \{q'_0,q'_1,\dots,q'_{k-1}\} \subseteq \hwgrp$, where we define $q'_i := \prod_{j=0}^{k-1} q_j^{A_{ij}}$. Then $\langle S \rangle = \langle T \rangle$.
\end{reptheorem}
\begin{proof}
It is clear that $\langle T \rangle \subseteq \langle S \rangle$, so we only need to show that $\langle S \rangle \subseteq \langle T \rangle$. Since $A$ is invertible, there exists a matrix $B \in \mathbb{Z}_d^{k \times k }$ such that $BA = I$. Now define the multiset $U := \{r_0,r_1,\dots,r_{k-1}\} \subseteq \hwgrp$, where $r_{\ell} := \prod_{j=0}^{k-1} (q'_j)^{B_{\ell j}}$ for each $\ell = 0,1,\dots,k-1$. Then 
 clearly we have $\langle U \rangle \subseteq \langle T \rangle$, and we will show that $\langle S \rangle \subseteq \langle U \rangle$. Corresponding to the sets $S$ and $U$, we also define the sets $J_S$ and $J_U$, according to the notation above.

Now fix any $\ell$, and note that $r_\ell = \prod_{i=0}^{k-1} \left( \prod_{j=0}^{k-1} q_j^{A_{ij}} \right)^{B_{\ell i}}$, and this product can be rearranged using group commutators as
\begin{equation}
\label{eq:equiv-gen-proof-2}
     r_\ell = \lambda_{\ell} \prod_{j=0}^{k-1} q_j^{\sum_{i=0}^{k-1} B_{\ell i} A_{ij}}, \;\; \lambda_{\ell} \in \langle J_S \rangle.
\end{equation}
Notice that the fact $BA =I$ over the ring $\mathbb{Z}_d$ implies that $\sum_{i=0}^{k-1} B_{\ell i} A_{ij} \equiv 1  \pmod{d}$ if $\ell = j$, and otherwise $\sum_{i=0}^{k-1} B_{\ell i} A_{ij} \equiv 0  \pmod{d}$. Thus using Eq.~\eqref{eq:equiv-gen-proof-2} and the definition of $J_S$, we can conclude that $r_{\ell} = \overline{\lambda}_{\ell} q_{\ell}$ for some $\overline{\lambda}_{\ell} \in \langle J_S \rangle$, for every $\ell$. We can now use Lemma~\ref{lem:helper-lem-phase-mult}(i),(iv) to deduce that $\langle J_S \rangle = \langle J_U \rangle \subseteq \langle U \rangle$. It follows that for every $\ell$, we have $\left(\overline{\lambda}_{\ell}\right)^{-1} \in \langle U \rangle$, and this implies $q_{\ell} \in \langle U \rangle$. We thus conclude that $\langle S \rangle \subseteq \langle U \rangle$, completing the proof.
\end{proof}

The next lemma characterizes, for generating set $S\subseteq \hwgrp$, the group of elements proportional to identity that $S$ generates.

\begin{replemma}{lem:KS-generators}
Suppose $S := \{q_0,q_1,\dots,q_{k-1}\} \subseteq \hwgrp$ is an ordered multiset. Consider the matrix $\pi_2(S) \in \mathbb{Z}_d^{2n \times k}$, and let $\overline{K} \in \mathbb{Z}_d^{k \times \ell}$ be such that the columns of $\overline{K}$ is a generating set for $\ker(\pi_2(S))$. Then we have the following:
\begin{enumerate}[(i)]
    \item Let $v \in \ker(\pi_2(S))$, and define $q' := \prod_{j=0}^{k-1}q_j^{v_j}$. Then $\pi_2(q')=0$.
    \item Define the set $N_S := \left \{ \prod_{j=0}^{k-1}q_j^{v_j} : v \in \ker(\pi_2(S)) \right\}$. Then $I_S = \langle N_S, J_S \rangle$.
    \item Define the multiset $\overline{K}_S := \left \{\prod_{j=0}^{k-1} q_j^{\overline{K}_{ji}} : i=0,1,\dots,\ell-1 \right\}$. Then $I_S = \langle \overline{K}_S, J_S \rangle$.
\end{enumerate}
\end{replemma}

\begin{proof}
For (i) note that since $v \in \ker(\pi_2(S))$, we have $\pi_2(q') = \sum_{j=0}^{k-1} v_j \pi_2(q_j) = \pi_2(S) v = 0$, where all expressions are evaluated modulo $d$.

For part (ii) we only need to prove that $I_S \subseteq \langle N_S,J_S \rangle$, as the reverse containment is obvious by part (i) and definition of $J_S$. Take any $q \in I_S$. Then one can write $q = \prod_{j=0}^{\ell} h_{j}$, where each $h_j \in S$. Using group commutators we can rearrange this product to obtain $q = \lambda \prod_{j=0}^{k-1} q_j^{r_j}$, for $\lambda \in \langle J_S \rangle$, and non-negative integers $r_j$ for each $j$. Using Lemma~\ref{lem:max-degree}, we can further simplify this expression by reducing the powers $r_j$ modulo $d$, to obtain $q = \lambda \lambda' \prod_{j=0}^{k-1} q_j^{s_j}$, for some $\lambda' \in \langle J_S \rangle$, and each $s_j \in \mathbb{Z}_d$. Define $s:=(s_0,s_1,\dots,s_{k-1}) \in \mathbb{Z}_d^{k}$. Now we know that $\pi_2(q) = 0$. This then implies that $\sum_{j=0}^{k-1} s_j \pi_2(q_j) = 0$, evaluated over $\mathbb{Z}_d$, or equivalently $s \in \ker(\pi_2(S))$. We can then conclude that $\prod_{j=0}^{k-1} q_j^{s_j} \in N_S$, and thus $q \in \langle N_S,J_S \rangle$.

We now prove part (iii). Each column of $\overline{K}_S$ is an element of $\ker(\pi_2(S))$. Thus it is clear that $\overline{K}_S \subseteq N_S$, and so $\langle \overline{K}_S,J_S \rangle \subseteq \langle N_S, J_S \rangle$. We will now show that $N_S \subseteq \langle \overline{K}_S,J_S \rangle$, which will prove $\langle N_S,J_S \rangle = \langle \overline{K}_S,J_S \rangle$, and then by (ii) we will obtain $\langle \overline{K}_S,J_S \rangle = I_S$. For this, take any $v \in \ker(\pi_2(S))$. Since the columns of $\overline{K}$ is a generating set for $\ker(\pi_2(S))$, we have $v = \overline{K}w$, for some $w \in \mathbb{Z}_d^{\ell}$. Then using the group commutators one obtains for $\lambda, \lambda', \lambda'' \in \langle J_S \rangle$,
\begin{equation}
\begin{split}
\prod_{j=0}^{k-1}q_j^{v_j} &= \prod_{j=0}^{k-1}q_j^{(\sum_{i=0}^{\ell-1} \overline{K}_{ji}w_i) \bmod{d}} = \lambda \prod_{j=0}^{k-1}q_j^{\sum_{i=0}^{\ell-1} \overline{K}_{ji}w_i} \\
&= \lambda \lambda' \prod_{i=0}^{\ell -1} \left( \prod_{j=0}^{k-1}q_j^{\overline{K}_{ji}w_i} \right) = \lambda \lambda' \lambda'' \prod_{i=0}^{\ell -1} \left( \prod_{j=0}^{k-1}q_j^{\overline{K}_{ji}} \right)^{w_i}.
\end{split}
\end{equation}
Now since $\prod_{j=0}^{k-1}q_j^{\overline{K}_{ji}} \in \langle \overline{K}_S \rangle$ for every $i$, we deduce that $\prod_{j=0}^{k-1}q_j^{v_j} \in \langle \overline{K}_S,J_S \rangle$, and since $v$ is arbitrary, we conclude that $N_S \subseteq \langle \overline{K}_S,J_S \rangle$.
\end{proof}

Next, we prove the corollary of Lemma~\ref{lem:span-same}.

\begin{repcorollary}{cor:equal-span-same-inv-factor}
Let $A \in \mathbb{Z}_d^{k \times p}$ and $B \in \mathbb{Z}_d^{k \times q}$. If the columns of $A$ and $B$ generate the same submodule of $\mathbb{Z}_d^k$, then the number of invariant factors of $A$ and $B$ are equal.
\end{repcorollary}
\begin{proof}
Without loss of generality, assume that $p \leq q$. Define $\overline{A} = \begin{pmatrix}
    A & 0
\end{pmatrix} \in \mathbb{Z}_d^{k \times q}$,
where we have added $q-p$ zero columns to $A$. Then the columns of $\overline{A}$ and $B$ still generate the same submodule of $\mathbb{Z}_d^k$. Thus by Lemma~\ref{lem:span-same} we may conclude that there exists an invertible matrix $C \in \mathbb{Z}_d^{q \times q}$ such that $\overline{A} = BC$. Next, suppose that $A$ has the Smith normal form $A P = QD$ where $P,Q$ are invertible matrices, and $D$ is a diagonal matrix (all matrices are over the ring $\mathbb{Z}_d$). Let $D$ have $r$ non-zero elements. Then we can note that $\overline{A} \left( \begin{smallmatrix} P & 0 \\ 0 & I \end{smallmatrix} \right) = Q \begin{pmatrix} D & 0 \end{pmatrix}$, or equivalently $B C \left( \begin{smallmatrix} P & 0 \\ 0 & I \end{smallmatrix} \right) = Q \begin{pmatrix} D & 0 \end{pmatrix}$, and since $\left( \begin{smallmatrix} P & 0 \\ 0 & I \end{smallmatrix} \right)$ is invertible, we can conclude by the uniqueness part of the Smith normal form (Theorem~\ref{thm:smith-normal-form}) that $\begin{pmatrix} D & 0 \end{pmatrix}$ is the Smith normal form of $B$, and hence $B$ also has $r$ invariant factors.
\end{proof}

There is a generalization of Lemma~\ref{lem:group_counting}, when the commutator $\llbracket p,q\rrbracket_d=c$, as appearing in the lemma, is not a unit in $\mathbb{Z}_d$. In this case, instead of exact equality, we get lower bounds:
\begin{lemma}
\label{lem:group_counting_nonunit}
Let $p,q\in\overline{\mathcal{P}}_n$, $G\subseteq\overline{\mathcal{P}}_n$, $\llbracket p,q\rrbracket_d=c \in \mathbb{Z}_d \setminus \{0\}$, and assume the following:
\begin{enumerate}[(a)]
    \item Each element in $G$ commutes with both $p$ and $q$.
    \item $a, b$ are the smallest positive integers such that $p^a$ and $q^b$ are equivalent to $I$ up to phase. 
    \item $a_1, b_1$ are the smallest positive integers such that $\llbracket p^{a_1},q\rrbracket_d = \llbracket p,q^{b_1}\rrbracket_d = 0$.
\end{enumerate}
Then the following statements are true:
\begin{enumerate}[(i)]
    \item $a_1 = b_1$, $a_1$ divides both $a$ and $b$, and both $a, b$ divide $d$, where division is over integers. Moreover, the order of $c$ in $\mathbb{Z}_d$ is $a_1$.
    \item We have the lower bounds $\left \lvert \langle G,p \rangle / I_{G \cup \{p\}} \right \rvert \geq a_1 \; |\langle G \rangle / I_G|$, $\left \lvert \langle G,q \rangle / I_{G \cup \{q\}} \right \rvert \geq a_1 \; |\langle G \rangle / I_G|$, and $\left \lvert \langle G,p,q \rangle / I_{G \cup \{p,q\}} \right \rvert \geq a_1^2 \; |\langle G \rangle / I_G|.$
\end{enumerate}
\end{lemma}

\begin{proof}
We first prove (i). It is clear that $a_1 \leq a$: if not, since $p^a$ is equivalent to $I$ up to phase factors, we have $\llbracket p^{a},q\rrbracket_d = 0$, which contradicts (c). Moreover, if $a_1$ does not divide $a$, i.e. $a = ta_1 + s$ for some integers $t \geq 0$ and $0 < s < a_1$, then we also obtain a contradiction to (c) as $0 = \llbracket p^{a},q\rrbracket_d = \llbracket p^{s},q\rrbracket_d$. Similarly, we can argue that $b_1$ divides $b$. Also it is easy to see that $a$ divides $d$: if this is not true, then $d = t a + s$ for some integers $t \geq 0$ and $0 < s < a$, which implies $p^d = (p^{a})^t p^s$, and since both $p^d$ and $p^a$ are equivalent to $I$ up to phase factors, we conclude that the same is true for $p^s$, which contradicts (b). Similarly we may also conclude that $b$ divides $d$.

Next, suppose that the order of $c$ in $\mathbb{Z}_d$ is $0 < \gamma < a_1$. Thus $c\gamma \bmod d = 0$, which implies $\llbracket p^{\gamma},q\rrbracket_d = c\gamma \bmod d = 0$, contradicting (c). Thus the order of $c$ in $\mathbb{Z}_d$ is at least $a_1$. On the other hand, $ca_1 \bmod{d} = \llbracket p^{a_1},q\rrbracket_d = 0$, implying the order of $c$ is at most $a_1$. From this we conclude that the order of $c$ is $a_1$. By repeating the same argument using the fact that $b_1$ be the smallest positive integer such that $\llbracket p,q^{b_1}\rrbracket_d = 0$, we also get that the order of $c$ equals $b_1$, and thus $a_1 = b_1$.

Now we prove (ii). All quotient groups appearing in the following argument will be treated as subgroups of $\hwgrp / K$. Define the sets $P = \{p^j: j=0,1,\dots,a-1\}$, and $Q = \{q^j : j=0,1,\dots,b-1\}$. By (b) we know that the elements in the set $P$ (or $Q$) are all distinct, even up to phase factors. This implies that the quotient groups $\langle p \rangle / I_{\{p\}}$ and  $\langle q \rangle / I_{\{q\}}$ are in one-to-one correspondence with the sets $P$ and $Q$ respectively, where the identification is made by taking equivalence up to phase factors. Next, we make two claims, which we prove at the end:

\textbf{Claim~A:} Let $H \subseteq \hwgrp$ is a subgroup such that $p$ commutes with every element of $H$. Define $H_Q := \{q^j \in Q: \omega^t q^j \in H, \text{ for some } t \in \mathbb{Z}_d\}$. Then $H_Q \subseteq \{q^j: 0 \leq j \leq b-1, j \equiv 0 \pmod{a_1}\}$. 

\textbf{Claim~B:} Let $H \subseteq \hwgrp$ is a subgroup such that $q$ commutes with every element of $H$. Define $H_P := \{p^j \in P: \omega^t p^j \in H, \text{ for some } t \in \mathbb{Z}_d\}$. Then $H_P \subseteq \{p^j: 0 \leq j \leq a-1, j \equiv 0 \pmod{a_1}\}$.

Consider the group $\langle G \rangle$, each of whose elements commute with $q$ (resp. $p$) by (a). Then it follows from Claim~B and Claim~A respectively, that
\begin{equation}
\label{eq:Gpq-int-count}
\begin{split}
& \left\lvert \frac{\langle G \rangle}{I_G} \bigcap \frac{\langle p \rangle}{I_{\{p\}}} \right \rvert \leq | \{a_1 t \bmod{a} : t \in \mathbb{Z}_{a}\}| = a / a_1, \\
& \left\lvert \frac{\langle G \rangle}{I_G} \bigcap \frac{\langle q \rangle}{I_{\{q\}}} \right \rvert \leq | \{a_1 t \bmod{b} : t \in \mathbb{Z}_{b}\}| = b / a_1.
\end{split}
\end{equation}
This now implies
\begin{equation}
\label{eq:Gpq-int-count1}
\begin{split}
\left \lvert \langle G,p \rangle / I_{G \cup \{p\}} \right \rvert &= \frac{|\langle G 
\rangle / I_G| \; |\langle p \rangle / I_{\{p\}}|}{|\langle G 
\rangle / I_G \cap \langle p \rangle / I_{\{p\}}|} \geq \frac{|\langle G 
\rangle / I_G| \; a}{a / a_1} = a_1 \; |\langle G 
\rangle / I_G|, \\
\left \lvert \langle G,q \rangle / I_{G \cup \{q\}} \right \rvert &= \frac{|\langle G 
\rangle / I_G| \; |\langle q \rangle / I_{\{q\}}|}{|\langle G 
\rangle / I_G \cap \langle q \rangle / I_{\{q\}}|} \geq \frac{|\langle G 
\rangle / I_G| \; b}{b / a_1} = a_1 \; |\langle G 
\rangle / I_G|.
\end{split}
\end{equation}
For the bound on $\left \lvert \langle G,p,q \rangle / I_{G \cup \{p,q\}} \right \rvert$, consider the subgroup generated by $\overline{G} := G \cup \{p \}$. Then every element of $\langle \overline{G} \rangle$ commutes with $p$. Again by Claim~A, we have that $\left \lvert \langle \overline{G} \rangle / I_{\overline{G}} \cap \langle q \rangle / I_{\{q\}} \right \rvert \leq b / a_1$, and then it follows that 
\begin{equation}
\label{eq:Gpq-int-count2}
\left \lvert \langle G,p,q \rangle / I_{G \cup \{p,q\}} \right \rvert = \frac{|\overline{G} \rangle / I_{\overline{G}}| \; |\langle q \rangle / I_{\{q\}}|}{\left \lvert \langle \overline{G} \rangle / I_{\overline{G}} \cap \langle q \rangle / I_{\{q\}} \right \rvert} \geq a_1 \; |\langle\overline{G} \rangle / I_{\overline{G}}| \geq a_1^2 |\langle G 
\rangle / I_G|,
\end{equation}
where we used Eq.~\eqref{eq:Gpq-int-count1} in the last inequality.

We now prove Claim A, and then proof of Claim B is exactly similar. For contradiction, assume that there exists $q^j \in H_Q$ such that $j \not \equiv 0 \pmod{a_1}$. Then $0 = \llbracket p,q^{j}\rrbracket_d = cj \bmod{d}$, as $q^j$ is equivalent to some element of $H$ up to phase factors. But this implies that $j$ must be a multiple of $a_1$, as $a_1$ is the order of $c$ (specifically order implies that for any positive integer $0 < \gamma < a_1$, $c\gamma \bmod{d} \neq 0$). Thus we have a contradiction, and the claim is proved.
\end{proof}

Lemma~\ref{lem:min-gen-set-first-bound-simp} follows from parts (iii)-(v) of the next lemma. Several parts of this stronger lemma are used in subsequent proofs in Section~\ref{ssec:minimal-gen-sets}.

\begin{lemma}
\label{lem:min-gen-set-first-bound}
Let $\pi_2(S) \in \mathbb{Z}_d^{2n \times k}$ have the Smith normal form $\left(\begin{smallmatrix} D & 0 \\ 0 & 0 
\end{smallmatrix}\right)$, so that $\pi_2(S) P = Q 
\left(\begin{smallmatrix}
D & 0 \\
0 & 0
\end{smallmatrix}\right)$, for $D \in \mathbb{Z}_d^{r \times r}$ a diagonal matrix with all diagonal entries non-zero, $r \geq 1$, and invertible matrices $P \in \mathbb{Z}_d^{k \times k}$, $Q \in \mathbb{Z}_d^{2n \times 2n}$. Let $M$ denote the submodule of $\mathbb{Z}_d^{2n}$ generated by the columns of $\pi_2(S)$. Then we have the following:
\begin{enumerate}[(i)]
    \item $v \in \pi_2(\langle S \rangle)$ if and only if $v \in M$, where $\pi_2(\langle S \rangle)$ can be regarded as a multiset of size $|\langle S \rangle|$.
    \item If $T \subseteq \langle S \rangle$ is a generating set of $\langle S \rangle$, then the submodule of $\mathbb{Z}_d^{2n}$ generated by the columns of $\pi_2(T)$ equals $M$.
    \item If $T \subseteq \langle S \rangle$ is a generating set of $\langle S \rangle$, then $T$ contains a subset $T'$, with $|T'| \geq r$, such that $u \in T'$ implies $\pi_2(u) \neq 0$. Moreover, the number of invariant factors of $\pi_2(T)$ is $r$.
    \item If $r=k$, then $S$ is a generating set of $\langle S \rangle$ of the smallest size.
    \item If $r \leq k$, there exists a generating set $T = T' \cup \{p\}$ of $\langle S \rangle$ such that $\langle p \rangle = I_S$ and $|T'|=r$. Moreover, any such generating set has the following properties: (a) the columns of $\pi_2(T')$ generate the submodule $M$, (b) the matrix $\pi_2(T')$ has $r$ invariant factors, and (c) for distinct elements $q,r \in T'$, $\pi_2(q)$ and $\pi_2(r)$ are distinct and non-zero.
\end{enumerate}
\end{lemma}

\begin{proof}
For part (i), suppose that $v \in \pi_2(\langle S \rangle)$. This means that there exists $p \in \langle S \rangle$ such that $\pi_2(p)=v$. Now one can write $p = \prod_{j=0}^{\ell} h_j$, where each $h_j \in S$, which upon rearrangement of the order of the product using group commutators and simplifying the result using $s_j^d = \pm I$ for each $j$ (by Lemma~\ref{lem:max-degree}), gives $p = \lambda \prod_{j=0}^{k-1} s_j^{w_j}$, where $w := (w_0,w_1,\dots,w_{k-1}) \in \mathbb{Z}_d^k$, and $\lambda \in \langle S \rangle$ with $\pi_2(\lambda)=0$. We now conclude that $v = \pi_2(p) = \sum_{j=0}^{k-1} w_j \pi_2(s_j) = \pi_2(S) w \in M$. For the other direction, assume $v \in M$, which then implies that there exists $w \in \mathbb{Z}_d^k$, such that $v = \pi_2(S) w$. Now define $p := \prod_{j=0}^{k-1} s_j^{w_j}$. Then clearly $p \in \langle S \rangle$, and $\pi_2(p)=v$, proving that $v \in \pi_2(\langle S \rangle)$.

For part (ii), let the ordered multiset $T := \{q_0,q_1,\dots,q_{\ell -1}\} \subseteq \langle S \rangle$ be a generating set of $\langle S \rangle$, and let $M'$ be the module generated by the columns of $\pi_2(T) \in \mathbb{Z}_d^{2n \times \ell}$. Take $v \in M$. Then by part (i), $v \in \pi_2(\langle S \rangle)$. Since $\langle T \rangle = \langle S \rangle$ by assumption, we again get $v \in M'$ by part (i). Thus $M \subseteq M'$, and running the argument backwards gives $M' \subseteq M$. 

Note that by Lemma~\ref{lem:snf-nonzeros}, the minimal number of generators of $M$ is $\Theta(M) = r$. Part (ii) implies that the columns of $\pi_2(T)$ generate the submodule $M$, and so by Corollary~\ref{cor:equal-span-same-inv-factor}, $\pi_2(M)$ has $r$ invariant factors. Also, if $|\{ v \in T: \pi_2(v) \neq 0 \}| < r$, then this leads to a contradiction as it implies that the minimal number of generators of $M$ is less than $r$.

Part (iv) follows because $S$ is a generating set for $\langle S \rangle$ of size $k$, and part (iii) implies that any generating set of $\langle S \rangle$ must have size $k$.

For part (v), if $r=k$, we can take $T' = S$ and any $p \in I_S$ such that $\langle p \rangle = I_S$. So assume $r < k$, and let us denote $\hat{Q} := Q 
\left(\begin{smallmatrix}
D & 0 \\
0 & 0
\end{smallmatrix}\right) = \left( \overline{Q} \;\;\; 0 \right)$, where $\overline{Q} \in \mathbb{Z}_d^{2n \times r}$ with no non-zero columns and all columns distinct (by assumption on invertibility of $Q$, and uniqueness of Smith normal form). We define the ordered multiset $U := \{u_0,u_1,\dots,u_{k-1}\} \subseteq \langle S \rangle$, with $u_i := \prod_{j=0}^{k-1} s_j^{P_{ji}}$. The equation $\pi_2(S) P = Q 
\left(\begin{smallmatrix}
D & 0 \\
0 & 0
\end{smallmatrix}\right)$ then implies $\pi_2(u_i)$ equals the $i^{\text{th}}$ column of $\hat{Q}$ for every $i$, from which we deduce that for every $i = r,r+1,\dots,k-1$, $u_i = \omega^{\mu_i} I \in I_S$ for some $\mu_i \in \mathbb{Z}_d$. Now Theorem~\ref{thm:equiv-gen} implies that $\langle S \rangle = \langle U \rangle$, as $P$ is invertible. Also Lemma~\ref{lem:identity-gen}(ii) implies that there exists $\mu \in \mathbb{Z}_d$ such that $\langle \omega^{\mu} I \rangle = I_S$. Combining these facts we deduce that $\langle S \rangle = \langle u_0,u_1,\dots,u_{r-1}, \omega^{\mu} I\rangle$. Thus we may choose $T' = \{u_0,u_1,\dots,u_{r-1}\}$ and $p=\omega^{\mu} I$. We now prove properties (a)-(c) for any such generating set. Note that the columns of $\pi_2(T)$ and $\pi_2(T')$ both generate the same submodule. Then property (a) follows by part (ii), while property (b) follows by part (iii) and Corollary~\ref{cor:equal-span-same-inv-factor}. Property (c) follows because if $\pi_2(T')$ had two columns which are the same or if one column was zero, then the number of invariant factors of $\pi_2(T')$ (which equals the minimal number of generators for the submodule generated by the columns of $T'$) would be less than $r$, which would contradict property (b).
\end{proof}

Finally, we present the proof of Lemma~\ref{lem:d-prime-min-gen-sets}.

\begin{replemma}{lem:d-prime-min-gen-sets}
Given $S \subseteq \hwgrp$, suppose that $T := T' \cup \{p\}$ is a near-minimal generating set of $\langle S \rangle$ with $\langle p \rangle = I_S$. Let $d$ be prime. Then the following conditions are equivalent.
\begin{enumerate}[(i)]
    \item $T$ is a minimal generating set of $\langle S \rangle$.
    \item $\langle T' \rangle$ is a stabilizer subgroup of $\hwgrp$, and $I_S \neq \{I\}$.
\end{enumerate}
\end{replemma}
\begin{proof}
Let the number of invariant factors of $\pi_2(S)$ be $r \geq 1$, and suppose $T' = \{q_0,q_1,\dots,q_{r-1}\}$.
We first note a few facts about $T'$. Since $d$ is prime, $\mathbb{Z}_d$ is a field. By Lemma~\ref{lem:min-gen-set-first-bound}(v), we know that $\pi_2(T')$ has $r \geq 1$ invariant factors, and since $|T'| = r$, this implies that the kernel of the matrix $\pi_2(T')$ is trivial (as $\mathbb{Z}_d$ is a field). Now let $T'' := \{p^{\gamma_0} q_0, p^{\gamma_1} q_1, \dots, p^{\gamma_{r-1}} q_{r-1}\}$, for some $(\gamma_0,\gamma_1,\dots,\gamma_{r-1}) \in \mathbb{Z}_d^r$. Then by Lemma~\ref{lem:helper-lem-phase-mult}(iv), we know that $\langle J_{T'} \rangle = \langle J_{T''} \rangle$. Also since $\pi_2(T') = \pi_2(T'')$, we conclude that the kernel of $\pi_2(T'')$ is trivial. Combining these observations, and by using Lemma~\ref{lem:KS-generators}(iii) we deduce that $I_{T'} = I_{T''}$. Let us call this Fact~(a). The other fact we need is that, since $d$ is prime, $\langle \omega^j I \rangle = \{\omega^\ell I: \ell \in \mathbb{Z}_d\}$, for every $j \in \mathbb{Z}_d \setminus \{0\}$. Let us call this Fact~(b). We now return to the proof.

First we prove that (i) implies (ii). So assume that $T$ is a minimal generating set of $\langle S \rangle$. For contradiction, assume that $I_S = \{I\}$, which then implies $p = I$. Hence $T$ cannot be a minimal generating set as $\langle T' \rangle = \langle T', p \rangle$. Thus $I_S \neq \{I\}$. Next for contradiction again assume that $\langle T' \rangle$ is not a stabilizer group. This means that there exists $\omega^j I \in \langle T' \rangle$, for some $j \in \mathbb{Z}_d \setminus \{0\}$, and then Fact~(b) implies that $p \in \langle T' \rangle$. Thus again we conclude that $T$ is not a minimal generating set of $\langle S \rangle$.

Next we prove that (ii) implies (i). So assume that $T'$ is a stabilizer group and $I_S \neq \{I\}$, which implies that $I_{T'} = \{I\}$ and $p \neq I$, respectively. For contradiction, assume that $T$ is not a minimal generating set of $\langle S \rangle$. Then by Theorem~\ref{thm:min-gen-set-simple-form}(iii), we can conclude that there exist integers $\gamma_0, \gamma_1,\dots, \gamma_{r-1} \in \mathbb{Z}_d$, such that $p \in \langle T'' \rangle$, where $T'' := \{p^{\gamma_0} q_0, p^{\gamma_1} q_1, \dots, p^{\gamma_{r-1}} q_{r-1}\}$. Now by Fact~(a) we also have that $I_{T''} = I_{T'} = \{I\}$. Hence we can conclude that $p = I$, giving a contradiction. This proves the lemma.
\end{proof}

}
\edit{

\section{Maximum collections of non-commuting pairs}
\label{ssec:case-max-number-pairs}

While subsection~\ref{ssec:achievable_tuples} gave both necessary and sufficient conditions for a $k$-tuple $(f_0,\dots,f_{k-1}) \in (\mathbb{Z}_d\setminus\{0\})^{k}$ to be an achievable non-commuting pair relation on $n$ qudits, it is possible to give a more direct characterization of which $k$-tuples are achievable non-commuting pair relations for the case $k=nm$, i.e. when we have the maximum number of non-commuting pairs on $n$ qudits.

To simplify the presentation, let us introduce the following notation.

\textbf{Notation.} Suppose $f := (f_0,\dots,f_{k-1})$ is an achievable non-commuting pair relation on $n$ qudits. Suppose $\beta := (\beta_0,\dots,\beta_{k-1}) \in (\mathbb{Z}_d\setminus\{0\})^{k}$ be such that $\beta_i f_i \not \equiv 0 \pmod{d}$, for every $i=0,\dots,k-1$. For a fixed $f$, we will denote the set of all such $\beta$ by $\pmb{\beta}(f)$. If $\beta \in \pmb{\beta}(f)$, we will denote $f_{\beta} := (f'_0,\dots,f'_{k-1})$, where $f'_i = \beta_i f_i \mod{d}$, for every $i$. Let $\text{Sym}(k)$ denote the permutation group on the set of indices $\{0,\dots,k-1\}$. If $\sigma \in \text{Sym}(k)$, we will denote $f_\sigma := (f_{\sigma(0)},\dots,f_{\sigma(k-1)})$.

Given $f_i \in \mathbb{Z}_d\setminus\{0\}$ as in the above notation, it is useful to explicitly write down for which values of $\beta_i \in \mathbb{Z}_d\setminus\{0\}$, we have $\beta_i f_i \not \equiv 0  \pmod{d}$. This is easy to calculate. We start by writing $f_i = h \gcd (f_i, d)$, the $\gcd$ evaluated over integers, so that we have $\gcd (h,d)=1$. Now consider the set $\delta_i := \{t d / \gcd (f_i,d) : t \in \mathbb{Z}, 1 \leq t < \gcd (f_i,d)\}$. Then we claim that $\beta_i f_i \not \equiv 0  \pmod{d}$ if and only if $\beta_i \not \in \delta_i$. The forward direction is easy (we prove the contrapositive): if $\beta_i \in \delta_i$, then we have $\beta_i f_i = htd$ for some integer $t$, and so $\beta_i f_i \equiv 0 \pmod{d}$. For the converse direction, assume for contradiction that $\beta_i f_i \equiv 0 \pmod{d}$, and $\beta_i \in (\mathbb{Z}_d\setminus\{0\}) \setminus \delta_i$. This means that $d$ divides $h \gcd (f_i,d) \beta_i$, and since $\gcd (h,d)=1$ this implies that $\beta_i$ is a multiple of $d/ \gcd (f_i,d)$, which is a contradiction.

With the notation above, we state a helper lemma:

\begin{lemma}
\label{lem:helper-lem-1}
If $f := (f_0,\dots,f_{k-1})$ is an achievable non-commuting pair relation on $n$ qudits, then
\begin{enumerate}[(i)]
    \item $f_{\sigma}$ is an achievable non-commuting pair relation on $n$ qudits, for every $\sigma \in \text{Sym}(k)$.
    \item $f_{\beta}$ is an achievable non-commuting pair relation on $n$ qudits, for every $\beta \in \pmb{\beta}(f)$.
    \item $\bar{f} := (f_0,\dots,f_{k-1}, f_k)$ is an achievable non-commuting pair relation on $n+1$ qudits, for every $f_k \in \mathbb{Z}_d\setminus\{0\}$.
\end{enumerate}
\end{lemma}

\begin{proof}
Let $S=\{s_0,\dots,s_{k-1}\}$ and $T=\{t_0,\dots,t_{k-1}\}$ be a non-commuting pair that achieves $f$. For part (i), let $\sigma \in \text{Sym}(k)$. Define the ordered sets $S_{\sigma} := \{s_{\sigma(0)},\dots,s_{\sigma(k-1)}\}$ and $T_{\sigma} := \{t_{\sigma(0)},\dots,t_{\sigma(k-1)}\}$. Then $S_{\sigma}$ and $T_{\sigma}$ achieves $f_{\sigma}$.

For part (ii), let $u_0,\dots,u_{k-1} \in \mathbb{Z}_d^{2n}$ be such that $s_i = P(u_i)$, for every $i$, and let $\beta \in \pmb{\beta}(f)$. Define the ordered set $\tilde{S} := \{P(\tilde{u}_0),\dots,P(\tilde{u}_{k-1})\}$, where $\tilde{u}_i = \beta_i u_i \bmod{d}$. Then $\tilde{S},T$ achieves $\tilde{f}$, because for all $i, j \in \{0,\dots,k-1\}$, we have
\begin{equation}
    \llbracket P(\tilde{u}_i),P(\tilde{u}_j)\rrbracket_d = \llbracket s_i,s_j\rrbracket_d = \llbracket t_i,t_j\rrbracket_d = 0, \;\;\;\;\;
    \llbracket P(\tilde{u}_i),t_j\rrbracket_d =
    \begin{cases}
        \beta_i f_i \bmod{d} \;\;\; & \text{if} \;\; i=j \\
        0 \;\;\; & \text{if} \;\; i \neq j.
    \end{cases}
\end{equation}

For part (iii), let $f_k \in \mathbb{Z}_d\setminus\{0\}$ be arbitrary. Now construct the ordered sets $\bar{S} := \{\bar{s}_0,\dots,\bar{s}_k\} \subseteq \mathcal{P}_{n+1}$ and $\bar{T} := \{\bar{t}_0,\dots,\bar{t}_k\} \subseteq \mathcal{P}_{n+1}$, where $\bar{s}_i = s_i \otimes I$, $\bar{t}_i = t_i \otimes I$ (here $I$ is $d \times d$), for every $i=0,\dots,k-1$, and $\bar{s}_k = I \otimes X$, $\bar{t}_k = I \otimes Z^{f_k}$ (here $I$ is $d^n \times d^n$). Then $\bar{S}, \bar{T}$ achieves $\bar{f}$.
\end{proof}

We can now provide the complete characterization of all achievable non-commuting pair relations on $n$ qudits, when $k=nm$.
\begin{theorem}
\label{thm:nm-case}
Let $\text{Sym}(nm)$ denote the permutation group on the set of indices $\{0,1,\dots,nm-1\}$. Define $f := (f_0,\dots,f_{nm-1}) \in (\mathbb{Z}_d\setminus\{0\})^{nm}$ given by
\begin{equation}
\label{eq:nm-case-1}
    f_{i + jn} = \frac{d}{p_j^{\alpha_j}}, \;\;\; \text{for all } i=0,1,\dots,n-1, \;\; j = 0,1,\dots,m-1,
\end{equation}
and also the set
\begin{equation}
\label{eq:nm-case-possible-tuples}
    \mathcal{F} := \{(f_{\beta})_{\sigma} \in (\mathbb{Z}_d\setminus\{0\})^{nm}: \beta \in \pmb{\beta}(f), \; \sigma \in \text{Sym}(nm) \}.
\end{equation}
Then the following hold:
\begin{enumerate}[(i)]
    \item If $\beta \in \pmb{\beta}(f)$, then for every $i,j$, we have $\beta_{i+jn} \in \mathbb{Z}_d\setminus\{0\} \setminus \{t p_j^{\alpha_j} : t \in \mathbb{Z}, 1 \leq t < d/ p_j^{\alpha_j}\}$, and $(f_{\beta})_{i+jn} \in \{t d/ p_j^{\alpha_j} : t \in \mathbb{Z}, 1 \leq t < p_j^{\alpha_j}\}$.
    \item Each element of $\mathcal{F}$ is an achievable non-commuting pair relation on $n$ qudits. Conversely, if any $f' \in (\mathbb{Z}_d\setminus\{0\})^{nm}$ is an achievable non-commuting pair relation on $n$ qudits, then $f' \in \mathcal{F}$.
\end{enumerate}
\end{theorem}

\begin{proof}
(i) The statement about $\beta_{i+jn}$ follows from the argument in the paragraph above Lemma~\ref{lem:helper-lem-1}. This now implies that $(f_{\beta})_{i+jn} \in \{t d/ p_j^{\alpha_j} : t \in \mathbb{Z}, 1 \leq t < p_j^{\alpha_j}\}$.

(ii) We first prove the forward direction. We choose the ordered sets $S := \{P(u_0),\dots,P(u_{nm-1})\}$ and $T := \{P(v_0),\dots,P(v_{nm-1})\}$, where the vectors $u_0,\dots,u_{nm-1},v_0,\dots,v_{nm-1} \in \mathbb{Z}^{2n}_d$ have components defined by
\begin{equation}
    (u_{i + jn})_{\ell} = 
    \begin{cases}
        d/p_j^{\alpha_j} &\;\;\; \text{if } \ell = i \\
        0 &\;\;\; \text{if } \ell \neq i
    \end{cases}\;\;, \;\;\;\; 
    (v_{i + jn})_{\ell} = 
    \begin{cases}
        -d/p_j^{\alpha_j} &\;\;\; \text{if } \ell = i+n \\
        0 &\;\;\; \text{if } \ell \neq i+n
    \end{cases}\;\;,
\end{equation}
for all $i \in \{0,1,\dots,n-1\}$, $j \in \{0,1,\dots,m-1\}$, and $\ell \in \{0,\dots,2n-1\}$. Then one can easily check that $S,T$ achieves the non-commuting pair relation $\bar{f} = (\bar{f}_0,\dots,\bar{f}_{nm-1}) \in (\mathbb{Z}_d\setminus\{0\})^{nm}$, where $\bar{f}_i = (f_i)^2 \bmod{d}$, for every $i=0,\dots,nm-1$. Now by Corollary~\ref{cor:helper-cor-2-1}, there exists $\beta_0,\beta_1,\dots,\beta_{m-1} \in \mathbb{Z}_d\setminus\{0\}$, such that $\beta_j (d / p_j^{\alpha_j})^2 \bmod{d} = d / p_j^{\alpha_j}$, for every $j$. By Lemma~\ref{lem:helper-lem-1}(ii), we then conclude that $f$ is also an achievable non-commuting pair relation on $n$ qudits. Finally, each element of $\mathcal{F}$ is an achievable non-commuting pair relation on $n$ qudits by Lemma~\ref{lem:helper-lem-1}(i),(ii).

For the converse, suppose that the ordered sets $S = \{s_0,\dots,s_{nm-1}\}$ and $T = \{t_0, \dots,t_{nm-1}\}$ are a collection of non-commuting pairs on $n$ qudits that generate $f' = (f'_0,\dots,f'_{nm-1}) \in (\mathbb{Z}_d\setminus\{0\})^{nm}$. For every $j=0,\dots,m-1$, define the multisets $F_j := \{f'_i : f'_i \not\equiv 0 \pmod{p_j^{\alpha_j}}, \; i \in \{0,\dots,nm-1\}\}$, and then by Lemma~\ref{lem:lower-bound} we know that $|F_j| \leq n$. We now claim that (i) $|F_j| = n$, for every $j$, and (ii) $F_{\ell} \cap F_j = \emptyset$, whenever $j \neq \ell$. To see this note that if any of these two conditions does not hold, then by a simple counting argument we conclude that there exists $f'_i \not \in \cup_{j=1}^{m} F_j$. But then that would imply that $f'_i$ is divisible by every $p_j^{\alpha_j}$, or equivalently $f'_i$ is divisible by $d$. This gives a contradiction, which proves the claim. The claim also implies that each $f'_i$ is an element of exactly one $F_j$. Now fix a $F_j$ and choose any $f'_i \in F_j$. Then we must have $f'_i = k (d / p_j^{\alpha_j})$, for some $k \not \equiv 0 \pmod{p_j^{\alpha_j}}$. This immediately implies $f' \in \mathcal{F}$, as $F_j$ is arbitrary.
\end{proof}

}
\bibliographystyle{quantum}
\bibliography{bibliography}
\end{document}